\documentclass[notitlepage,reprint, aps, prl, twocolumn,superscriptaddress,nofootinbib]{revtex4-2}

\usepackage[utf8]{inputenc}
\usepackage[normalem]{ulem}
\usepackage{amsmath,amssymb,amsthm,mathtools}
\usepackage{graphicx}
\usepackage{bm}
\usepackage[dvipsnames]{xcolor}
\usepackage{physics}
\usepackage{booktabs}
\usepackage{longtable}
\usepackage{float}
\usepackage{comment}
\usepackage{lipsum}

\usepackage{algorithm}
\usepackage{algorithmicx}
\usepackage{algpseudocode}

\usepackage[unicode]{hyperref}
\hypersetup{
  bookmarksnumbered=true,
  unicode=true, colorlinks=true,
  linkcolor=NavyBlue,
  citecolor=NavyBlue,
  urlcolor=NavyBlue
}

\makeatletter \def\maketitle{ \@author@finish \title@column\titleblock@produce \suppressfloats[t]} \makeatother

\newcommand{\suppl}{Supplemental Material}

\DeclareMathOperator{\sgn}{sgn}

\newtheorem{theorem}{Theorem}
\newtheorem{lemma}{Lemma}
\newtheorem{proposition}{Proposition}
\newtheorem{corollary}{Corollary}
\theoremstyle{definition}

\theoremstyle{remark}
\newtheorem{remark}{Remark}

\newcommand{\hk}[1]{\textcolor[rgb]{0.0, 0, 0.0}{#1}}

\begin{document}

\title{Quantum Error Mitigation Simulates General Non-Hermitian Dynamics}

\author{Hiroki Kuji} \email{1225702@ed.tus.ac.jp} \affiliation{Department of Physics, Tokyo University of Science,1-3 Kagurazaka, Shinjuku, Tokyo, 162-8601, Japan} \affiliation{Department of Electrical, Electronic, and Communication Engineering, Faculty of Science and Engineering, Chuo University} \affiliation{Institute for Theoretical Physics III and Center for Integrated Quantum Science and Technology, University of Stuttgart, 70550 Stuttgart, Germany}

\author{Suguru Endo} \affiliation{NTT Computer and Data Science Laboratories, NTT Inc., Musashino 180-8585, Japan} \affiliation{NTT Research Center for Theoretical Quantum Information, NTT Inc., 3-1 Morinosato Wakanomiya, Atsugi, Kanagawa, 243-0198, Japan}

\author{Tetsuro Nikuni} \affiliation{Department of Physics, Tokyo University of Science,1-3 Kagurazaka, Shinjuku, Tokyo, 162-8601, Japan}

\author{Ryusuke Hamazaki} \email{ryusuke.hamazaki@riken.jp} \affiliation{Nonequilibrium Quantum Statistical Mechanics RIKEN Hakubi Research Team, RIKEN Pioneering Research Institute (PRI), Wako, Saitama 351-0198, Japan} \affiliation{RIKEN Center for Interdisciplinary Theoretical and Mathematical Sciences (iTHEMS), RIKEN, Wako 351-0198, Japan}

\author{Yuichiro Matsuzaki} \email{ymatsuzaki872@g.chuo-u.ac.jp} \affiliation{Department of Electrical, Electronic, and Communication Engineering, Faculty of Science and Engineering, Chuo University}

\date{\today}

\begin{abstract}
While non-Hermitian Hamiltonians enable exotic dynamical phenomena, implementing their nonunitary time evolution on near-term quantum devices remains challenging.
We propose a hardware-friendly protocol that simulates non-Hermitian dynamics without ancillas, controlled time evolution, or continuous monitoring.
The protocol combines a Gorini-Kossakowski-Sudarshan-Lindblad (GKSL) evolution via classical Gaussian white-noise averaging with stochastic quantum error mitigation (QEM) to cancel the jump contribution at the level of expectation values.
The mitigation layer uses only single-qubit operations.
We validate the method through numerical simulations of an asymmetric-hopping model and an open {\it XXZ} spin chain with non-Hermitian boundary fields.
Our work provides a programmable and ancilla-free framework for investigating exotic dynamics beyond the class of completely positive and trace-preserving dynamics using QEM.
\end{abstract}

\maketitle
\emph{Introduction---}\label{introduction} Among various candidates to describe open quantum systems, non-Hermitian physics~\cite{Hatano1997PRB,Hatano1996PRL,Ashida2020,ElGanainy2018,Bender1998} has attracted much attention recently because of its simplicity and tractability.
Non-Hermitian systems exhibit phenomena with no Hermitian counterpart, including exceptional points~\cite{kato2013perturbation,Minganti2019}, $\mathcal{PT}$-symmetry~\cite{Bender2024,Bender1999}, non-Hermitian topology~\cite{Kawabata2019,Bergholtz2021,Gong2018}, and the non-Hermitian skin effect~\cite{Yao2018,Song2019,Okuma2020,Yi2020,Longhi2022}.
These developments motivate experimentally accessible and programmable methods to probe non-Hermitian dynamics on quantum devices~\cite{Shen2025,Zhang2025,Wang2019,Jebraeilli2025,Wen2019,Liu2023,Fan2025,Takasu2020,Li2019,Chen2021,Xiao2017,Wu2024,Chen2025EP,Liu2021,Ren2022}.

Despite this motivation, implementing non-Hermitian (hence nonunitary; cf.\ imaginary-time evolution~\cite{Motta2019,McArdle2019,Kosugi2022,Silva2023,Coopmans2023,kuji2025vqnhite}) dynamics on quantum hardware remains challenging, since non-Hermitian dynamics is not completely positive and trace-preserving (CPTP)~\cite{Breuer2007}.
Native operations in both gate-based devices and analog quantum simulators are unitary, and a widely used experimental strategy is to implement the target non-Hermitian Hamiltonian dynamics by introducing ancilla qubits and postselection or engineered dissipation~\cite{wang2026nonhermitirydberg}.
Despite their generality, such constructions are not readily programmable: even when the target Hamiltonian contains only two-body interactions, such approaches may require higher-body controls due to additional couplings with the ancilla qubits, as well as model-dependent measurement or dissipation engineering, substantially increasing experimental complexity (Fig.~\ref{fig:overview}~(a)).

On the other hand, algorithmic approaches on quantum computers including Monte Carlo-based methods~\cite{li2025dynamicssimulationarbitrarynonhermitian}, linear-combination-of-unitaries techniques~\cite{An2023,An2025,Zheng2021,Berry2015,Childs2012}, and quantum-signal-processing-based methods~\cite{chan2023simulatingnonunitarydynamicsusing,Low2017,Gilyn2019}, typically rely on deep circuits with controlled unitary primitives, such as controlled time evolution on a quantum circuit, together with additional ancillas~\cite{Schlimgen2022}.
These requirements remain demanding for near-term quantum devices.
This calls for a programmable, hardware-friendly route to non-Hermitian simulation that avoids ancilla overhead and controlled time evolution.

In this Letter, we propose a general framework for realizing non-Hermitian dynamics by constructing an auxiliary Gorini-Kossakowski-Sudarshan-Lindblad (GKSL) evolution~\cite{Han2021,Gorini1976,Lindblad1976} and canceling its quantum-jump contribution via quantum error mitigation (QEM)~\cite{Cai2023,Endo2018,Sun2021,Li2017,Temme2017,PRXQuantum.4.040329}.
Figure~\ref{fig:overview}~(b) illustrates the overall workflow of our protocol.
We exploit the fact that the GKSL generator, engineered by averaging over noise trajectories~\cite{Chenu2017,Schmolke2022}, separates into a non-Hermitian Hamiltonian contribution and a jump contribution~\cite{Dalibard1992,Plenio1998,hamazaki2025intro}.
We show that the jump contribution can be systematically canceled using stochastic QEM (sQEM).
We numerically benchmark the protocol on \hk{two examples}, i.e., a one-dimensional interacting hard-core boson model \hk{with disordered asymmetric hopping}~\cite{Hamazaki2019}
\hk{and an open {\it XXZ} spin chain with non-Hermitian boundary fields}~\cite{ScaleFreeNHSE}.
Our work establishes a programmable \hk{and hardware-friendly} route to simulate exotic non-CPTP dynamics using QEM\hk{, without ancillas, controlled time evolution, or continuous monitoring}.

\medskip
\emph{Framework}---%
\begin{figure}
  \centering
  \includegraphics[width=\linewidth]{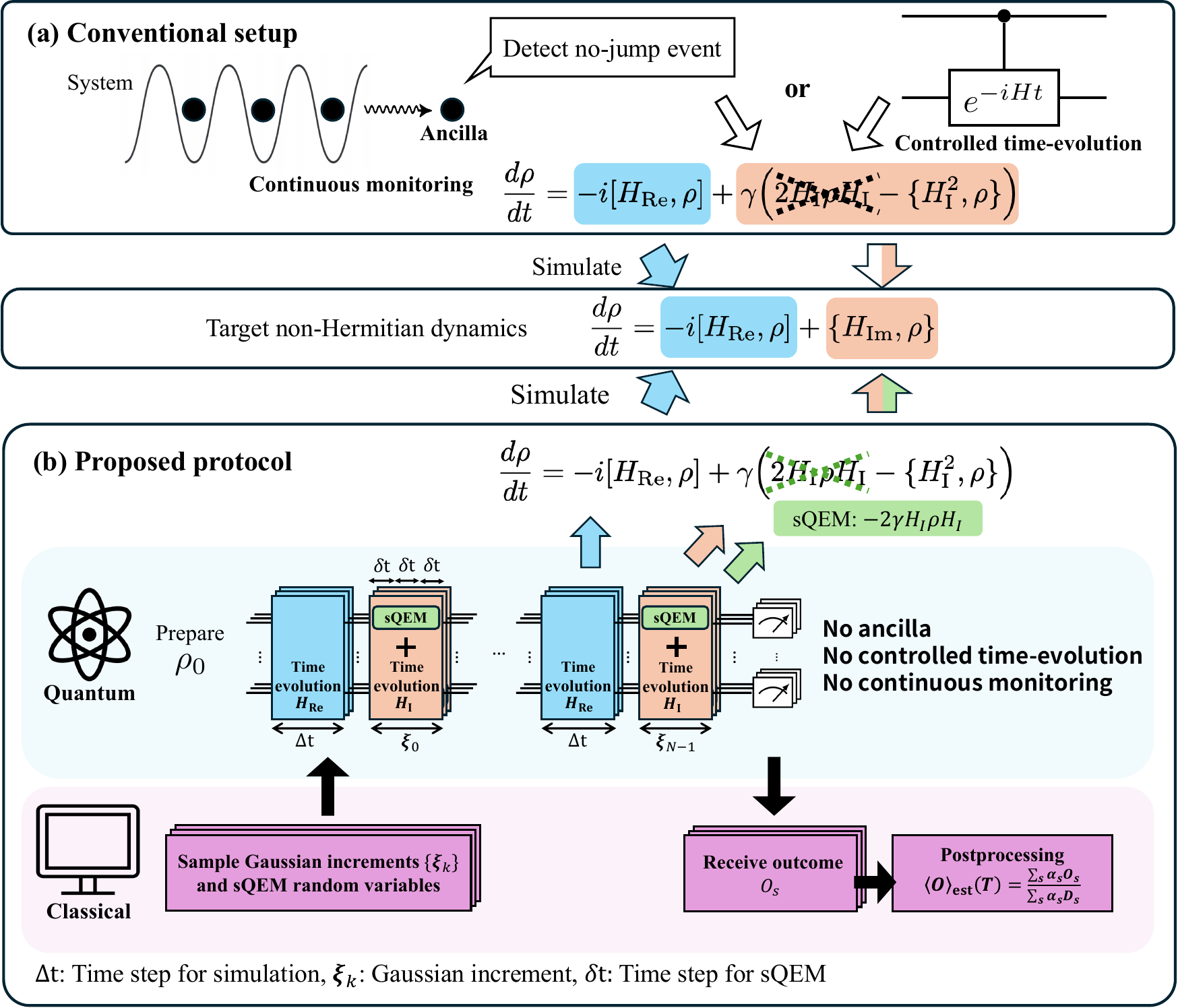}
  \caption{Overview of the proposed protocol for realizing non-Hermitian dynamics via a noise-averaged GKSL evolution and sQEM. (a) Conventional setup:
  \hk{target non-Hermitian dynamics is simulated by selecting the no-jump process through continuous monitoring, or by using controlled time evolution with ancillas}.
  (b) Proposed protocol:
  \hk{a classical processor samples the Gaussian increments $\boldsymbol{\xi}_k=\{\xi_{k,\ell}\}_{\ell \in E}$ and the sQEM random variables, while the quantum device implements the time evolution generated by $H_{\mathrm{Re}}$ and the stochastically driven unitary evolution generated by $H_{\mathrm{I}}(k)\coloneqq \sum_{\ell \in E}\xi_{k,\ell} H_{\mathrm{I},\ell}$, applies the sampled sQEM operations, and returns final measurement outcomes}.
The noise-averaged evolution reproduces the auxiliary GKSL dynamics.
Using the decomposition of the GKSL generator into a non-Hermitian (no-jump) part and a
  jump part, we implement sQEM within each trajectory to cancel the jump contribution, thereby realizing the target non-Hermitian Hamiltonian dynamics at the level of normalized
  \hk{expectation values}.
  \hk{The final estimator is obtained by classical postprocessing of the measured outcomes with sQEM signs and normalization factors.}
  Here, $\Delta t$ is the simulation time step and $\delta t$ is the time step in the sQEM cancellation map.}
  \label{fig:overview}
\end{figure}
We consider a many-body quantum system defined on a set of sites.
In this part, we present a general framework to realize a target non-Hermitian dynamics generated by
\begin{align}
  H_{\mathrm{target}}=H_{\mathrm{Re}}+iH_{\mathrm{Im}},
  \label{eq:target_hamiltonian}
\end{align}
where $H_{\mathrm{Re}}$ and $H_{\mathrm{Im}}$ are Hermitian operators.
We assume that $H_{\mathrm{Im}}$ is decomposed into local terms as $H_{\mathrm{Im}} = \sum_{\ell \in E} H_{\mathrm{Im},\ell}$, where each $H_{\mathrm{Im},\ell}$ acts nontrivially only on the sites in a finite set $\ell$ (the local support), and $E$ denotes the collection of such supports.
\hk{By the full-bond benchmark, we mean that $H_{\mathrm{Im}}$ contains a two-site term on every nearest-neighbor bond in the system. By the boundary-local benchmark, we mean that $H_{\mathrm{Im}}$ contains only single-site terms acting on the specified boundary sites.}
Throughout this paper we take $\hbar=1$.

Our goal is to \hk{estimate normalized expectation values} under the dynamics generated by $H_\mathrm{target}$ by evolving the system with an experimentally implementable stochastic Hamiltonian
\begin{align}
  H_{\mathrm{impl}}(t) = H_{\mathrm{Re}} + \sum_{\ell \in E} f_{\ell}(t) H_{\mathrm{I},\ell},
  \label{eq:simulating_hamiltonian}
\end{align}
and applying sQEM.
Here, $f_{\ell}(t)$ are independent \hk{engineered} classical Gaussian white-noise \hk{processes}\hk{, rather than uncontrolled device noise,} with $\mathbb{E}_f[f_{\ell}(t)] = 0$ and $\mathbb{E}_f\left[f_{\ell}(t) f_{\ell'}(t')\right]= 2\gamma_{\ell}\delta_{\ell\ell'}\delta(t-t')$, where $\gamma_\ell$ is the noise strength associated with \hk{the local support} $\ell$.
The operators $H_{\mathrm{I},\ell}$ are chosen such that the resulting effective non-Hermitian contribution reproduces $H_{\mathrm{Im}}$; the explicit construction is given in End Matter.
\hk{Uncontrolled device noise can, in principle, also be mitigated within the same sQEM framework by extending the quasi-probability decomposition to account for both the jump contribution and the device-noise contribution, although this increases the sampling overhead.}

Our protocol consists of two sampling \hk{components}.
\hk{One component reproduces} GKSL dynamics via the noise-averaging approach proposed in Ref.~\cite{Chenu2017}.
\hk{The other is used in sQEM to cancel the unwanted jump contribution in each experimental run.}
\hk{The random choices and sQEM post-processing are performed classically, while the corresponding evolution and final measurements are implemented on the quantum device.}

We partition the total evolution time $T$ into $N$ steps of size $\Delta t = T/N$ and define the integrated Gaussian increments
\begin{align}
  \xi_{k,\ell} \coloneqq \int_{t_k}^{t_k+\Delta t} f_{\ell}(s)ds,
  \label{eq:xi}
\end{align}
where $t_k = k\Delta t$.
\hk{The increments in Eq.~\eqref{eq:xi}} \hk{satisfy} $\xi_{k,\ell}\sim \mathcal{N}(0,2\gamma_{\ell}\Delta t)$ and \hk{are} independent across time steps $k\in \{0, \dots, N-1\}$ and \hk{local supports} $\ell$.
We implement the discretization via a Trotter formula separating $H_{\mathrm{Re}}$ and \hk{$H_{\mathrm{I}}(k)\coloneqq \sum_{\ell \in E}\xi_{k,\ell}H_{\mathrm{I},\ell}$}\hk{.
For nearest-neighbor terms, we apply an even-odd bond decomposition to both parts.}

Upon averaging over the classical noise, $\bar\rho(t)\coloneqq \mathbb{E}_{\boldsymbol{\xi}}[\rho(t)]$ obeys the GKSL generator~\cite{Chenu2017}
\begin{align}
  \mathcal{L}[\bar{\rho}] = -i[H_{\mathrm{Re}},\bar\rho] +\sum_{\ell \in E}\gamma_{\ell}\Big(2H_{\mathrm{I},\ell}\bar\rho H_{\mathrm{I},\ell}-\{H_{\mathrm{I},\ell}^2,\bar\rho\}\Big).
  \label{eq:master-lindblad-hi}
\end{align}

Next, to implement the target non-Hermitian dynamics, we decompose the GKSL generator $\mathcal{L}$ in Eq.~\eqref{eq:master-lindblad-hi} into a non-Hermitian Hamiltonian part and a jump part $\dot{\bar{\rho}} = \mathcal{L}[\bar{\rho}] = \mathcal{L}_{\mathrm{NH}}[\bar{\rho}] + \mathcal{L}_{\mathrm{J}}[\bar{\rho}]$, where $\mathcal{L}_{\mathrm{NH}}[\rho]\coloneqq -i[H_{\mathrm{Re}},\rho]-\sum_{\ell \in E}\gamma_{\ell}\{H_{\mathrm{I},\ell}^2,\rho\}$ and $\mathcal{L}_\mathrm{J}[\rho] \coloneqq \sum_{\ell \in E}2\gamma_{\ell}H_{\mathrm{I},\ell}\rho H_{\mathrm{I},\ell}.$ Here, $\mathcal{L}_{\mathrm{NH}}$ generates the non-Hermitian Hamiltonian evolution, while $\mathcal{L}_{\mathrm{J}}$ describes the jump \hk{contribution}.

To cancel $\mathcal{L}_{\mathrm{J}}$ at the level of ensemble-averaged \hk{expectation values}, we introduce a cancellation map $\mathcal{E}_{\mathcal{C}}(\delta t)\coloneqq\mathcal{I}+\delta t\mathcal{E}_{\mathcal{C}}^{(1)}+\mathcal{O}(\delta t^2)$, where $\mathcal{I}$ denotes the identity superoperator and $\mathcal{E}_{\mathcal{C}}^{(1)}$ is the first-order superoperator.
We also define the GKSL and non-Hermitian channels as $\mathcal{E}_{\mathrm{GKSL}}(\delta t) \coloneqq \exp(\delta t\mathcal{L}) = \mathcal{I}+\delta t\mathcal{L}+\mathcal{O}(\delta t^2)$ and $\mathcal{E}_{\mathrm{NH}}(\delta t) \coloneqq \exp(\delta t\mathcal{L}_{\mathrm{NH}}) = \mathcal{I}+\delta t\mathcal{L}_{\mathrm{NH}}+\mathcal{O}(\delta t^2)$, respectively.
This leads to
\begin{align}
  \mathcal{E}_{\mathrm{NH}}(\delta t) = \mathcal{E}_{\mathcal{C}}(\delta t) \mathcal{E}_{\mathrm{GKSL}}(\delta t) +\mathcal{O}(\delta t^2),
  \label{eq:cancellation-condition}
\end{align}
which implies $\mathcal{E}_{\mathcal{C}}^{(1)}=\mathcal{L}_{\mathrm{NH}}-\mathcal{L}=-\mathcal{L}_{\mathrm{J}}$ at first order in $\delta t$, i.e., $\mathcal{E}_{\mathcal{C}}^{(1)}[\cdot] = -\sum_{\ell \in E}2\gamma_{\ell} H_{\mathrm{I},\ell}[\cdot]H_{\mathrm{I},\ell}$.

We implement $\mathcal{E}_{\mathcal{C}}(\delta t)$ via sQEM using experimentally available local \hk{basis operation} channels $\{\mathcal{B}_{\ell j}\}$\hk{.
We expand} $\mathcal{E}_{\mathcal{C}}^{(1)} = q_0 \mathcal{I} + \sum_{\ell\in E}\sum_{j\ge1} q_{\ell j} \mathcal{B}_{\ell j}$, where \hk{$q_0$ and $q_{\ell j}$ are} the expansion coefficients of $\mathcal{E}_{\mathcal{C}}^{(1)}$ in \hk{this basis}.
For each \hk{local support} $\ell$, we assume that $\{\mathcal{B}_{\ell j}\}_{j\ge 1}$ spans the local \hk{superoperator space on that support.}
We then obtain the first-order quasi-probability decomposition
\begin{align}
  \mathcal{E}_{\mathcal{C}}(\delta t) = c(\delta t)\left(p_{0}\mathcal{I} + \sum_{\ell\in E}\sum_{j\ge 1}\alpha_{\ell j}p_{\ell j}\mathcal{B}_{\ell j}\right) + \mathcal{O}(\delta t^2),
  \label{eq:EC_qpd}
\end{align}
where $c(\delta t)\coloneqq 1+\delta t(q_0+\sum_{\ell\in E}\sum_{j\ge1}|q_{\ell j}|)$, $p_{\ell j}\coloneqq |q_{\ell j}|\delta t/c(\delta t)$, $p_0\coloneqq (1+q_0\delta t)/c(\delta t)$, $\alpha_{\ell j}\coloneqq \sgn(q_{\ell j})$\hk{, and $\kappa\coloneqq q_0+\sum_{\ell\in E}\sum_{j\ge1}|q_{\ell j}|$}.
By construction, $p_{0}+\sum_{\ell\in E}\sum_{j\ge1}p_{\ell j}=1$, so that $p_{0}=1-\mathcal{O}(\delta t)$ and $p_{\ell j}=\mathcal{O}(\delta t)$ for $j\ge1$.
In practice, we \hk{can} use the 16 single-qubit \hk{basis operation} channels~\cite{Endo2018} (listed in Table~\ref{tab:bases} in End Matter) and construct \hk{local} \hk{basis operation} channels by tensor products \hk{on each support}.

For $T=n \delta t$, iterating Eq.~\eqref{eq:cancellation-condition} over \hk{$n$} steps and substituting Eq.~\eqref{eq:EC_qpd} at each step yields the quasi-probability expansion for expectation values, $\Tr\left[O \rho_{\mathrm{NH}}(T)\right] = c(\delta t)^{n}\sum_{\vec{j}} \alpha_{\vec{j}} p_{\vec{j}} \Tr\left[O \rho_{\vec{j}}(T)\right] +\mathcal{O}(T\delta t)$, where $\rho_{\mathrm{NH}}(T)\coloneqq \mathcal{E}_{\mathrm{NH}}(\delta t)^n[\rho(0)]$, and $\vec{j}=(j_0,\ldots,j_{n-1})$ labels a sequence of \hk{basis operation} indices, with $j_k=0$ denoting the identity channel \hk{and $j_k=(\ell,j)$ denoting $\mathcal{B}_{\ell j}$}.
We define $p_{\vec{j}}\coloneqq \prod_{k=0}^{n-1}p_{j_k}$ and $\alpha_{\vec{j}}\coloneqq \prod_{k=0}^{n-1}\alpha_{j_k}$, with \hk{$p_{(\ell,j)}\coloneqq p_{\ell j}$, $\alpha_{(\ell,j)}\coloneqq \alpha_{\ell j}$, and} $\alpha_0\coloneqq 1$.
The state $\rho_{\vec{j}}(T)$ is obtained by inserting, at each of the $n$ steps of the GKSL evolution, the \hk{basis operation} channels specified by $\vec{j}$.
The remainder $\mathcal{O}(T\delta t)$ arises from the first-order discretization in the above derivation.
In our protocol, however, sQEM is implemented directly in continuous time, corresponding to the limit $\delta t\to 0^+$.
Therefore, sQEM itself does not introduce \hk{an additional} discretization-induced systematic error.

Finally, since the target non-Hermitian evolution is in general trace-non-preserving and basis \hk{operations} may also include trace-decreasing operations such as measurements, $\Tr\left[O \rho_{\mathrm{NH}}(T)\right]$ is generally unnormalized.
We therefore calculate the normalized expectation value
\begin{align}
    \langle O \rangle_{\mathrm{target}}(T)=\frac{\Tr\left[O \rho_{\mathrm{NH}}(T)\right]}{\Tr\left[\rho_{\mathrm{NH}}(T)\right]} =\frac{\sum_{\vec{j}} \alpha_{\vec{j}} p_{\vec{j}} \Tr\left[O \rho_{\vec{j}}(T)\right]}{\sum_{\vec{j}} \alpha_{\vec{j}} p_{\vec{j}} \Tr\left[\rho_{\vec{j}}(T)\right]},
\end{align}
with the limit $\delta t\to 0^+$.

We note that the sampling overhead can\hk{, in general,} grow exponentially with the simulated time \hk{because both the non-Hermitian normalization and the sQEM quasi-probability normalization can contribute exponential factors.}
\hk{For fixed precision, the sQEM contribution to the required number of trajectories scales as $\exp[2\kappa(L)T]$, where $\kappa(L)$ is determined by the number and strengths of the canceled channels.}
\hk{Thus, $\kappa(L)=O(1)$ when only $O(1)$ channels with bounded rates are canceled, whereas it can scale as $O(L)$ when $O(L)$ such channels are canceled.}

\hk{Our approach does not claim to eliminate the generic exponential-in-time sampling overhead in all settings.
Nevertheless, in favorable cases involving only $O(1)$ canceled channels with $L$-independent rates, the sQEM overhead does not introduce an exponential dependence on the total system size $L$.
Its practical advantage also lies in reducing the hardware requirements: the proposed protocol requires neither ancillary qubits, controlled time evolution, nor continuous monitoring.}

\hk{Details of the GKSL construction, the Trotter implementation, the sQEM implementation, and the associated sampling overhead are provided in End Matter and \suppl~\cite{supple}, Secs.~\ref{app:gaussian-average}--\hk{\ref{app:evenodd}}.}

\medskip
\emph{Full-bond non-Hermitian benchmark}---%
As \hk{a first benchmark}, we analyze the asymmetric hopping model of hard-core bosons with disorder and interaction $H = \sum_{i=0}^{L-2}\Bigl[-J\Bigl(e^{g} b_{i+1}^\dagger b_i+ e^{-g} b_i^\dagger b_{i+1}\Bigr)+ U n_i n_{i+1}\Bigr] + \sum_{i=0}^{L-1} h_i n_i$ \hk{in a one-dimensional four-site chain ($L=4$) with open boundary conditions~\cite{Hamazaki2019}}.
Here, $L$ is the number of sites, $J$ sets the hopping scale, $U$ is the nearest-neighbor interaction strength, $n_i = b_i^\dagger b_i$ is the particle-number operator at site $i$ with the annihilation operator $b_i$ of the hard-core boson, $g$ controls non-Hermiticity, and \hk{the onsite disorder is given by $h_i=h_{\mathrm{amp}}r_i$, with $h_{\mathrm{amp}}\in\{0.10, 8.0\}$ and $(r_0,r_1,r_2,r_3)=(0.9534, -0.2396, 0.8465, -0.4766)$.
The same disorder pattern is fixed for both disorder strengths.}
\hk{In this model, the non-Hermitian contribution is distributed over all nearest-neighbor bonds.
Thus, simulating the target dynamics requires canceling jump terms on all $L-1$ bonds.
Consequently, the coefficient $\kappa$ can grow with $L$.}
Note that for large $L$, this model exhibits a delocalized and non-Hermitian many-body localization regime for weak and strong disorder strengths, respectively~\cite{Hamazaki2019}.
That is, if the disorder strength is weak enough, the non-Hermitian hopping term leads to asymmetric particle transport.
Under the open boundary condition, the particles accumulate near one boundary of the system~\cite{Mu2020}.
In contrast, when the disorder strength is strong, particle transport is suppressed by many-body localization.

Via the standard hard-core-boson–to–spin-$1/2$ mapping \hk{$b_i \mapsto \sigma_i^-$, $b_i^\dagger \mapsto \sigma_i^+$, and $n_i=b_i^\dagger b_i \mapsto(1+\sigma_i^z)/2$}, where $\sigma_i^\alpha$ ($\alpha\in\{x,y,z\}$) are the Pauli operators acting on \hk{the} $i$th site and \hk{$\sigma_i^\pm=(\sigma_i^x\pm i\sigma_i^y)/2$.
Here,} $\{\ket{0}_i, \ket{1}_i\}$ denote the eigenstates defined by $\sigma_i^z\ket{0}_i=+\ket{0}_i$ and $\sigma_i^z\ket{1}_i=-\ket{1}_i$. The Hamiltonian can be written in the form of Eq.~\eqref{eq:target_hamiltonian}.
Explicitly, $H_{\mathrm{Re}}=-\frac{J\cosh g}{2}\sum_{i=0}^{L-2}\bigl(\sigma_i^x\sigma_{i+1}^x+\sigma_i^y\sigma_{i+1}^y\bigr)+\frac{U}{4}\sum_{i=0}^{L-2}(1+\sigma_i^z)(1+\sigma_{i+1}^z)+\frac{1}{2}\sum_{i=0}^{L-1} h_i(1+\sigma_i^z)$ and $H_{\mathrm{Im}} = -\frac{J\sinh g}{2}\sum_{i=0}^{L-2}\bigl(\sigma_i^x\sigma_{i+1}^y -\sigma_i^y\sigma_{i+1}^x \bigr)$.

\begin{figure}[H]
  \centering
  \includegraphics[width=\linewidth]{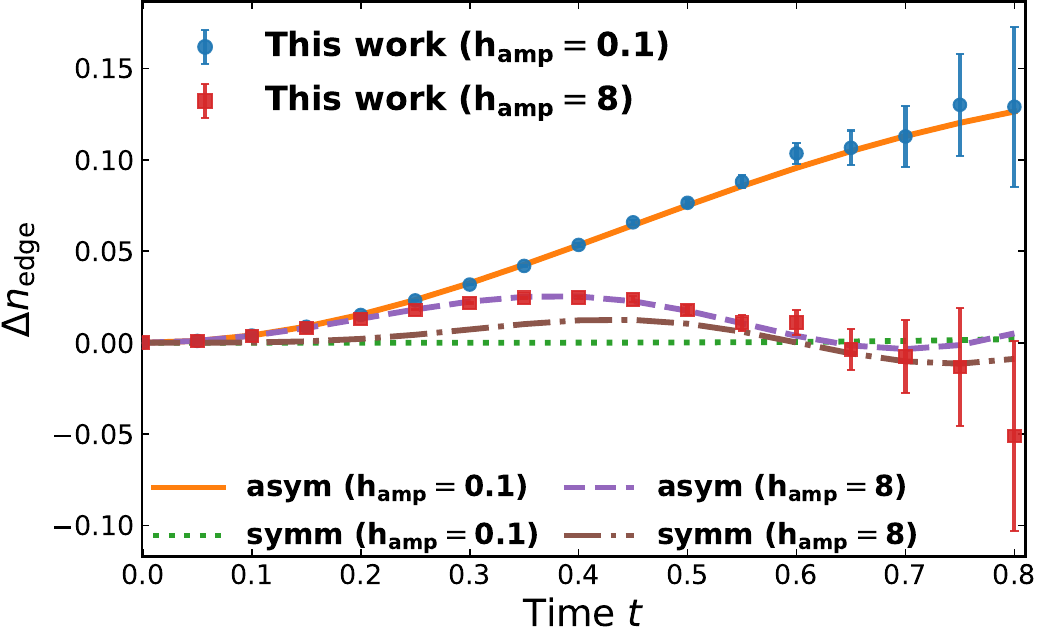}
  \caption{
    \hk{E}dge imbalance $\Delta n_{\mathrm{edge}}(t)=\langle n_{3}(t)\rangle_{\mathrm{norm}}-\langle n_{0}(t)\rangle_{\mathrm{norm}}$
    \hk{in the full-bond non-Hermitian benchmark for a four-site chain.}
    \hk{Results are shown}
    for weak disorder ($h_{\mathrm{amp}}=0.1$) and strong disorder ($h_{\mathrm{amp}}=8.0$), where the initial state is $|0110\rangle$.
    Solid orange and dashed purple lines show the exact target non-Hermitian dynamics (asym, $g=0.1$) for $h_{\mathrm{amp}}=0.1$ and $h_{\mathrm{amp}}=8.0$, respectively.
    Dotted green and dash-dotted brown lines show the corresponding symmetric-hopping model (symm, $g=0$).
    Blue circles ($h_{\mathrm{amp}}=0.1$) and red squares ($h_{\mathrm{amp}}=8.0$) show Monte Carlo estimates from our protocol.
    Parameters: $J=1.0$, $g=0.1$, $U=1.0$, $\gamma=1.0$, $\Delta t=1.0\times 10^{-3}$.
    We use $B=210240$ independent runs, each obtained by averaging over $M_{\mathrm{num}}=2000$ stochastic trajectories (total $BM_{\mathrm{num}}=4.2048\times 10^8$ trajectories).
    Error bars show jackknife standard errors over the $B$ runs.}
  \label{fig:L4_merged}
\end{figure}

In the simulation, \hk{we initialize the system in the product state $\ket{\psi_0}=\ket{0110}$ and monitor the edge imbalance $\Delta n_{\mathrm{edge}}(t)=\langle n_{3}(t)\rangle_{\mathrm{norm}}-\langle n_{0}(t)\rangle_{\mathrm{norm}}$ (see Fig.~\ref{fig:L4_merged} for parameters).}
\hk{For each time point $t$, one data point is obtained} by averaging $M_{\mathrm{num}}$ i.i.d.\ trajectory samples, where each trajectory simultaneously samples the Gaussian increments for the GKSL channel and the sQEM layer.
We then repeat the above process for $B$ independent runs $b=1,\ldots,B$ at each time to employ the jackknife resampling method~\cite{Efron1992}.
\hk{Thus, the total number of i.i.d.\ trajectories is $N_{\mathrm{traj}}=B M_{\mathrm{num}}$.}
\hk{For trajectory $s$ in run $b$, denote its parity, observable contribution, and acceptance indicator by $\alpha_s^{(b)}$, $O_s^{(b)}(t)$, and $D_s^{(b)}(t)$, respectively.
The signed numerator and denominator averages are}
\hk{$\hk{N_O^{(b)}(t)}\coloneqq M_{\mathrm{num}}^{-1}\sum_{s=1}^{M_{\mathrm{num}}}\hk{\alpha_s^{(b)}O_s^{(b)}(t)}$ and
$\hk{D^{(b)}(t)}\coloneqq M_{\mathrm{num}}^{-1}\sum_{s=1}^{M_{\mathrm{num}}}\hk{\alpha_s^{(b)}D_s^{(b)}(t)}$.}
Using all $B$ runs, we estimate the normalized expectation value
$\hk{\langle O\rangle_{N_{\mathrm{traj}}}(t)}=\sum_{b=1}^{B}\hk{N_O^{(b)}(t)}/\sum_{b=1}^{B}\hk{D^{(b)}(t)}$.
\hk{We set $O_s^{(b)}(t)=0$ whenever $D_s^{(b)}(t)=0$.}
Error bars in all figures show jackknife standard errors obtained by resampling over the $B$ independent runs, accounting for fluctuations in both the numerator and the denominator of the normalized estimator.

As reference lines, we compute (i) the exact target non-Hermitian dynamics generated by $H_{\mathrm{target}}$ via direct matrix exponentiation, and (ii) the symmetric-hopping baseline obtained by setting $g=0$ (i.e., $H_{\mathrm{Im}}=0$).
Figure~\ref{fig:L4_merged} shows that for weak disorder ($h_{\mathrm{amp}}=0.1$) the exact non-Hermitian dynamics develops a positive edge imbalance $\Delta n_{\mathrm{edge}}(t)$, whereas the symmetric-hopping baseline remains close to $\Delta n_{\mathrm{edge}}\approx 0$.
Our Monte Carlo estimates reproduce this directional bias over the time window shown.
For strong disorder ($h_{\mathrm{amp}}=8.0$), the edge-imbalance signal is substantially reduced and approaches the symmetric-hopping behavior, consistent with the intuition that strong onsite disorder inhibits transport and suppresses the drift induced by asymmetric hopping.
This disorder strength corresponds to the non-Hermitian many-body localization regime for large system sizes~\cite{Hamazaki2019}.

Because we consider short times and a modest asymmetry $g=0.1$ (i.e., $e^{2g}\simeq 1.22$), the phenomenology consistent with the ``non-Hermitian skin effect'' here appears primarily as a transport bias captured by the edge imbalance rather than a dramatic one-edge accumulation.
These results capture the non-Hermiticity-induced bias and its suppression by disorder, providing the intended qualitative validation in this $L=4$ setting.
\hk{Additional numerical results, including a minimal two-qubit validation and site-resolved occupations for the full-bond model, are presented in \suppl~\cite{supple}.}

\medskip
\hk{\emph{Boundary-local non-Hermitian benchmark---}}

We next consider an open {\it XXZ} spin chain with non-Hermitian boundary fields~\cite{ScaleFreeNHSE}, $H=-(J/4)\sum_{i=1}^{L-1}(\sigma_i^x\sigma_{i+1}^x+\sigma_i^y\sigma_{i+1}^y+\Delta\sigma_i^z\sigma_{i+1}^z)+(ig/4)(\sigma_L^z-\sigma_1^z)$, where $\Delta$ is the XXZ anisotropy.

\hk{For this boundary-field model, the imaginary part is $H_{\mathrm{Im}}=(g/2)(n_L-n_1)$.
After shifting the imaginary part by $-(g/2)I$, we obtain $H_{\mathrm{Im}}-(g/2)I=-(g/2)[n_1+(1-n_L)]$, which changes only the overall norm and hence does not affect normalized expectation values.
This form contains only two projectors, $n_1$ and $1-n_L$.
This corresponds to choosing $H_{\mathrm{I},1}=n_1$ and $H_{\mathrm{I},L}=1-n_L$ with $\gamma=g/2$.
The corresponding jump contribution consists only of the two boundary terms $g n_1[\cdot]n_1$ and $g(1-n_L)[\cdot](1-n_L)$.
Therefore, unlike the full-bond benchmark above, the number of canceled jump maps is independent of $L$, and $\kappa=O(1)$ at fixed $g$.}

\hk{Starting from the N\'eel state \hk{$|0101\cdots\rangle$}, we compute the edge imbalance $\Delta n_{\mathrm{edge}}(t)=\langle n_L(t)\rangle_{\mathrm{norm}}-\langle n_1(t)\rangle_{\mathrm{norm}}$ for $L=6,10$ and compare the estimates with both the exact non-Hermitian dynamics and the exact Hermitian \hk{dynamics} obtained by setting $g=0$ (see Fig.~\ref{fig:boundary_edge_dynamics} for parameters).}

\begin{figure}[H]
\centering
    \includegraphics[width=\linewidth]{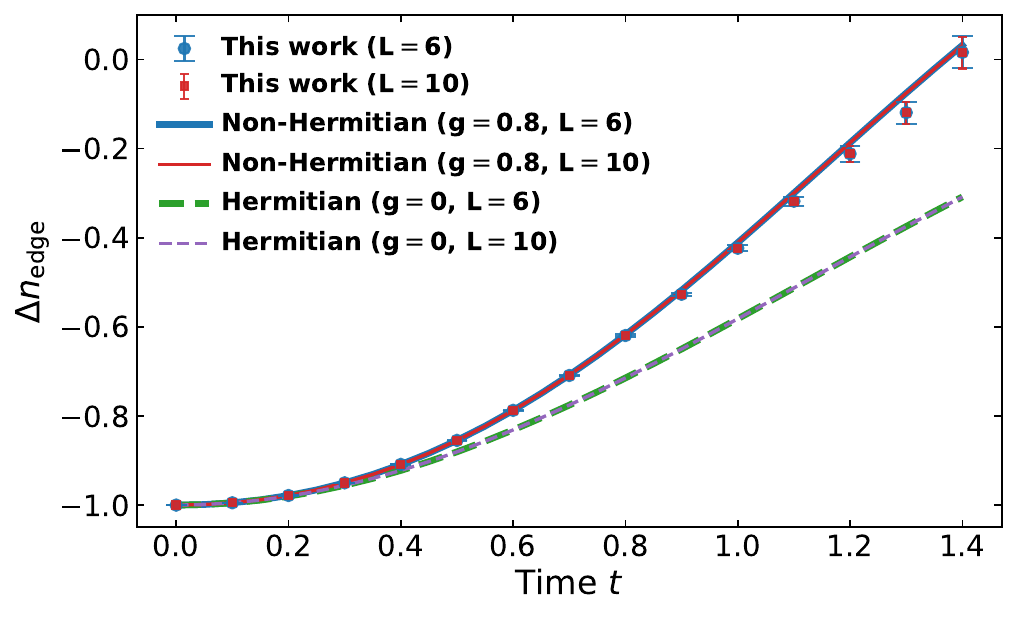}
    \caption{\hk{Edge imbalance $\Delta n_{\mathrm{edge}}(t)=\langle n_L(t)\rangle_{\mathrm{norm}}-\langle n_1(t)\rangle_{\mathrm{norm}}$ for the open {\it XXZ} spin chain with non-Hermitian boundary fields, starting from the \hk{N\'eel state $|0101\cdots\rangle$}.
Solid lines show the exact non-Hermitian dynamics at $g=0.8$, markers with error bars show Monte Carlo estimates from our protocol (This work), and dashed lines show the exact Hermitian dynamics at $g=0$, for \hk{$L=6,10$}.
The other parameters are $J=1.0$, $\Delta=0.8$.
We use $B=100$ independent runs, each averaging over $M_{\mathrm{num}}=10000$ stochastic trajectories (total $BM_{\mathrm{num}}=10^6$ trajectories per data point).
Error bars show jackknife standard errors over the $B$ runs.}}
    \label{fig:boundary_edge_dynamics}
\end{figure}

\hk{Figure~\ref{fig:boundary_edge_dynamics} shows that the Monte Carlo estimates agree with the exact boundary-field-induced dynamics within the jackknife error bars for both chain lengths.
Relative to the Hermitian dynamics at $g=0$, the edge imbalance increases more rapidly from its initial negative value, demonstrating the boundary accumulation induced by the non-Hermitian boundary fields.
With the same number of trajectories for both values of $L$, the error bars remain comparable, consistent with the $L$-independent sQEM cancellation rate in this boundary-local setting.}

\medskip
\emph{Conclusion}---%
We have proposed a general framework to realize non-Hermitian dynamics by constructing an auxiliary GKSL evolution and canceling its quantum-jump contribution via sQEM.
For targets with $k$-local interactions, the protocol requires only native $k$-local coherent-evolution blocks (without introducing higher-body controls), while the \hk{nonunitary} contribution is reconstructed by Monte Carlo sampling using single-qubit operations and measurements in the sQEM layer.
Numerical simulations on \hk{the full-bond hard-core-boson benchmark and the boundary-local open {\it XXZ} spin chain} confirm that the protocol faithfully reproduces the intended non-Hermitian dynamics, including asymmetric particle transport, its suppression by disorder\hk{, and boundary-field-induced edge dynamics.}
Importantly, when the required time evolutions are available as native analog time-evolution blocks, the protocol becomes naturally compatible with digital-analog quantum computing architectures~\cite{ParraRodriguez2020}.

Another interesting direction is to extend the framework to a broader class of non-CPTP quantum dynamics beyond non-Hermitian dynamics.
One possible example is the generalized \hk{master equation for biased ensembles of quantum trajectories}~\cite{Gupta2024,Garrahan2010,Liu2025}.

\medskip
\textit{Acknowledgments}---%
H.K. is supported by JSPS KAKENHI Grant No.~JP26KJ2001.
H.K.'s research stay in Germany is supported by JST ASPIRE Grant No.~JPMJAP24C2.
S.E. was supported by JST Moonshot R\&D Grant No.~JPMJMS2061, MEXT Q-LEAP Grant Nos.~JPMXS0120319794 and JPMXS0118067285, and JST CREST Grant Nos.~JPMJCR23I4 and JPMJCR25I4.
Y.M. is supported by JST Moonshot R\&D Grant No.~JPMJMS226C, JST CREST Grant No.~JPMJCR23I5, and JST PRESTO Grant No.~JPMJPR245B.
R.H. was supported by JSPS KAKENHI Grant No.~JP24K16982 and JST ERATO Grant No.~JPMJER2302.
\bibliographystyle{apsrev4-2}
\bibliography{ref}

@article{Chenu2017,
  title = {Quantum Simulation of Generic Many-Body Open System Dynamics Using Classical Noise},
  volume = {118},
  ISSN = {1079-7114},
  url = {http://dx.doi.org/10.1103/PhysRevLett.118.140403},
  DOI = {10.1103/physrevlett.118.140403},
  pages = {140403},
  number = {14},
  journal = {Phys. Rev. Lett.},
  publisher = {American Physical Society (APS)},
  author = {Chenu,  A. and Beau,  M. and Cao,  J. and del Campo,  A.},
  year = {2017},
  month = apr 
}

@article{Childs2021,
  title = {Theory of Trotter Error with Commutator Scaling},
  volume = {11},
  ISSN = {2160-3308},
  url = {http://dx.doi.org/10.1103/PhysRevX.11.011020},
  DOI = {10.1103/physrevx.11.011020},
  pages = {011020},
  number = {1},
  journal = {Phys. Rev. X},
  publisher = {American Physical Society (APS)},
  author = {Childs,  Andrew M. and Su,  Yuan and Tran,  Minh C. and Wiebe,  Nathan and Zhu,  Shuchen},
  year = {2021},
  month = feb 
}

@article{Cai2023,
  title = {Quantum error mitigation},
  volume = {95},
  ISSN = {1539-0756},
  url = {http://dx.doi.org/10.1103/RevModPhys.95.045005},
  DOI = {10.1103/revmodphys.95.045005},
  pages = {045005},
  number = {4},
  journal = {Rev. Mod. Phys.},
  publisher = {American Physical Society (APS)},
  author = {Cai,  Zhenyu and Babbush,  Ryan and Benjamin,  Simon C. and Endo,  Suguru and Huggins,  William J. and Li,  Ying and McClean,  Jarrod R. and O'Brien,  Thomas E.},
  year = {2023},
  month = dec 
}

@article{Endo2018,
  title = {Practical Quantum Error Mitigation for Near-Future Applications},
  volume = {8},
  ISSN = {2160-3308},
  url = {http://dx.doi.org/10.1103/PhysRevX.8.031027},
  DOI = {10.1103/physrevx.8.031027},
  pages = {031027},
  number = {3},
  journal = {Phys. Rev. X},
  publisher = {American Physical Society (APS)},
  author = {Endo,  Suguru and Benjamin,  Simon C. and Li,  Ying},
  year = {2018},
  month = jul 
}

@article{Sun2021,
  title = {Mitigating Realistic Noise in Practical Noisy Intermediate-Scale Quantum Devices},
  volume = {15},
  ISSN = {2331-7019},
  url = {http://dx.doi.org/10.1103/PhysRevApplied.15.034026},
  DOI = {10.1103/physrevapplied.15.034026},
  pages = {034026},
  number = {3},
  journal = {Phys. Rev. Applied},
  publisher = {American Physical Society (APS)},
  author = {Sun,  Jinzhao and Yuan,  Xiao and Tsunoda,  Takahiro and Vedral,  Vlatko and Benjamin,  Simon C. and Endo,  Suguru},
  year = {2021},
  month = mar 
}

@article{Ashida2020,
  title = {Non-Hermitian physics},
  volume = {69},
  ISSN = {1460-6976},
  url = {http://dx.doi.org/10.1080/00018732.2021.1876991},
  DOI = {10.1080/00018732.2021.1876991},
  number = {3},
  journal = {Adv. Phys.},
  publisher = {Informa UK Limited},
  author = {Ashida,  Yuto and Gong,  Zongping and Ueda,  Masahito},
  year = {2020},
  month = jul,
  pages = {249--435}
}

@article{Hamazaki2019,
  title = {Non-Hermitian Many-Body Localization},
  volume = {123},
  ISSN = {1079-7114},
  url = {http://dx.doi.org/10.1103/PhysRevLett.123.090603},
  DOI = {10.1103/physrevlett.123.090603},
  pages = {090603},
  number = {9},
  journal = {Phys. Rev. Lett.},
  publisher = {American Physical Society (APS)},
  author = {Hamazaki,  Ryusuke and Kawabata,  Kohei and Ueda,  Masahito},
  year = {2019},
  month = aug 
}

@article{Bender2024,
  title = {$\mathcal{PT}$-symmetric quantum mechanics},
  volume = {96},
  ISSN = {1539-0756},
  url = {http://dx.doi.org/10.1103/RevModPhys.96.045002},
  DOI = {10.1103/revmodphys.96.045002},
  pages = {045002},
  number = {4},
  journal = {Rev. Mod. Phys.},
  publisher = {American Physical Society (APS)},
  author = {Bender,  Carl M. and Hook,  Daniel W.},
  year = {2024},
  month = oct 
}

@article{Bender1999,
  title = {$\mathcal{PT}$-symmetric quantum mechanics},
  volume = {40},
  ISSN = {1089-7658},
  url = {http://dx.doi.org/10.1063/1.532860},
  DOI = {10.1063/1.532860},
  number = {5},
  journal = {J. Math. Phys.},
  publisher = {AIP Publishing},
  author = {Bender,  Carl M. and Boettcher,  Stefan and Meisinger,  Peter N.},
  year = {1999},
  month = may,
  pages = {2201--2229}
}

@book{kato2013perturbation,
  title={Perturbation theory for linear operators},
  author={Kato, Tosio},
  volume={132},
  year={2013},
  publisher={Springer Science \& Business Media}
}

@article{Yao2018,
  title = {Edge States and Topological Invariants of Non-Hermitian Systems},
  volume = {121},
  ISSN = {1079-7114},
  url = {http://dx.doi.org/10.1103/PhysRevLett.121.086803},
  DOI = {10.1103/physrevlett.121.086803},
  pages = {086803},
  number = {8},
  journal = {Phys. Rev. Lett.},
  publisher = {American Physical Society (APS)},
  author = {Yao,  Shunyu and Wang,  Zhong},
  year = {2018},
  month = aug 
}

@misc{chan2023simulatingnonunitarydynamicsusing,
      title={Simulating non-unitary dynamics using quantum signal processing with unitary block encoding}, 
      author={Hans Hon Sang Chan and David Mu{\~n}oz Ramo and Nathan Fitzpatrick},
      year={2023},
      eprint={2303.06161},
      archivePrefix={arXiv},
      primaryClass={quant-ph},
      url={https://arxiv.org/abs/2303.06161}, 
}

@misc{li2025dynamicssimulationarbitrarynonhermitian,
      title={Dynamics Simulation of Arbitrary Non-Hermitian Systems Based on Quantum Monte Carlo}, 
      author={Xiaogang Li and Kecheng Liu and Qiming Ding},
      year={2025},
      eprint={2507.11675},
      archivePrefix={arXiv},
      primaryClass={quant-ph},
      url={https://arxiv.org/abs/2507.11675}, 
}

@article{An2023,
  title = {Linear Combination of Hamiltonian Simulation for Nonunitary Dynamics with Optimal State Preparation Cost},
  volume = {131},
  ISSN = {1079-7114},
  url = {http://dx.doi.org/10.1103/PhysRevLett.131.150603},
  DOI = {10.1103/physrevlett.131.150603},
  pages = {150603},
  number = {15},
  journal = {Phys. Rev. Lett.},
  publisher = {American Physical Society (APS)},
  author = {An,  Dong and Liu,  Jin-Peng and Lin,  Lin},
  year = {2023},
  month = oct 
  }

@inbook{Efron1992,
  title = {Bootstrap Methods: Another Look at the Jackknife},
  ISBN = {9781461243809},
  ISSN = {0172-7397},
  url = {http://dx.doi.org/10.1007/978-1-4612-4380-9_41},
  DOI = {10.1007/978-1-4612-4380-9_41},
  booktitle = {Breakthroughs in Statistics},
  publisher = {Springer New York},
  author = {Efron,  Bradley},
  year = {1992},
  pages = {569--593}
}

@article{ParraRodriguez2020,
  title = {Digital-analog quantum computation},
  volume = {101},
  ISSN = {2469-9934},
  url = {http://dx.doi.org/10.1103/PhysRevA.101.022305},
  DOI = {10.1103/physreva.101.022305},
  pages = {022305},
  number = {2},
  journal = {Phys. Rev. A},
  publisher = {American Physical Society (APS)},
  author = {Parra-Rodriguez,  Adrian and Lougovski,  Pavel and Lamata,  Lucas and Solano,  Enrique and Sanz,  Mikel},
  year = {2020},
  month = feb 
}

@book{horn2012matrix,
  title={Matrix analysis},
  author={Horn, Roger A and Johnson, Charles R},
  year={2012},
  publisher={Cambridge University Press}
}

@book{papoulis2002probability,
  author = {Papoulis, Athanasios and Pillai, S. Unnikrishna},
  title = {Probability, Random Variables, and Stochastic Processes},
  edition = {4th},
  publisher = {McGraw-Hill},
  address = {Boston},
  year = {2002},
}

@book{bers1964partial,
  title={Partial differential equations},
  author={Bers, Lipman and John, Fritz and Schechter, Martin},
  year={1964},
  publisher={American Mathematical Soc.}
}

@book{dollard1979product,
  author = {Dollard, John D. and Friedman, Charles N.},
  title = {Product Integration with Applications to Differential Equations},
  publisher = {Addison-Wesley},
  address = {Reading, MA},
  year = {1979},
  series = {Encyclopedia of Mathematics and its Applications},
  volume = {10},
}

@article{hamazaki2025intro,
  title = {An Introduction to Monitored Quantum Systems and Quantum Trajectories: Spectrum, Typicality, and Phases},
  author = {Hamazaki, Ryusuke and Mochizuki, Ken and Oshima, Hisanori and Fuji, Yohei},
  journal = {Prog. Theor. Exp. Phys.},
  volume = {2026},
  number = {5},
  pages = {052A01},
  year = {2026},
  DOI = {10.1093/ptep/ptag055},
  url = {https://doi.org/10.1093/ptep/ptag055},
}

@article{Mu2020,
  title = {Emergent Fermi surface in a many-body non-Hermitian fermionic chain},
  volume = {102},
  ISSN = {2469-9969},
  url = {http://dx.doi.org/10.1103/PhysRevB.102.081115},
  DOI = {10.1103/physrevb.102.081115},
  pages = {081115},
  number = {8},
  journal = {Phys. Rev. B},
  publisher = {American Physical Society (APS)},
  author = {Mu,  Sen and Lee,  Ching Hua and Li,  Linhu and Gong,  Jiangbin},
  year = {2020},
  month = aug 
}

@article{Shen2025,
  title = {Observation of the non-Hermitian skin effect and Fermi skin on a digital quantum computer},
  volume = {16},
  ISSN = {2041-1723},
  url = {http://dx.doi.org/10.1038/s41467-025-55953-4},
  DOI = {10.1038/s41467-025-55953-4},
  pages = {1340},
  number = {1},
  journal = {Nat. Commun.},
  publisher = {Springer Science and Business Media LLC},
  author = {Shen,  Ruizhe and Chen,  Tianqi and Yang,  Bo and Lee,  Ching Hua},
  year = {2025},
  month = feb 
}

@article{Kawabata2019,
  title = {Symmetry and Topology in Non-Hermitian Physics},
  volume = {9},
  ISSN = {2160-3308},
  url = {http://dx.doi.org/10.1103/PhysRevX.9.041015},
  DOI = {10.1103/physrevx.9.041015},
  pages = {041015},
  number = {4},
  journal = {Phys. Rev. X},
  publisher = {American Physical Society (APS)},
  author = {Kawabata,  Kohei and Shiozaki,  Ken and Ueda,  Masahito and Sato,  Masatoshi},
  year = {2019},
  month = oct 
}

@article{ElGanainy2018,
  title = {Non-Hermitian physics and PT symmetry},
  volume = {14},
  ISSN = {1745-2481},
  url = {http://dx.doi.org/10.1038/nphys4323},
  DOI = {10.1038/nphys4323},
  number = {1},
  journal = {Nat. Phys.},
  publisher = {Springer Science and Business Media LLC},
  author = {El-Ganainy,  Ramy and Makris,  Konstantinos G. and Khajavikhan,  Mercedeh and Musslimani,  Ziad H. and Rotter,  Stefan and Christodoulides,  Demetrios N.},
  year = {2018},
  month = jan,
  pages = {11--19}
}

@article{Bender1998,
  title = {Real Spectra in Non-Hermitian Hamiltonians Having $\mathcal{P}\mathcal{T}$ Symmetry},
  volume = {80},
  ISSN = {1079-7114},
  url = {http://dx.doi.org/10.1103/PhysRevLett.80.5243},
  DOI = {10.1103/physrevlett.80.5243},
  number = {24},
  journal = {Phys. Rev. Lett.},
  publisher = {American Physical Society (APS)},
  author = {Bender,  Carl M. and Boettcher,  Stefan},
  year = {1998},
  month = jun,
  pages = {5243--5246}
}

@article{Zhang2025,
  title = {Observation of a non-Hermitian supersonic mode on a trapped-ion quantum computer},
  volume = {16},
  ISSN = {2041-1723},
  url = {http://dx.doi.org/10.1038/s41467-025-57930-3},
  DOI = {10.1038/s41467-025-57930-3},
  pages = {3286},
  number = {1},
  journal = {Nat. Commun.},
  publisher = {Springer Science and Business Media LLC},
  author = {Zhang,  Yuxuan and Carrasquilla,  Juan and Kim,  Yong Baek},
  year = {2025},
  month = apr 
}

@article{Song2019,
  title = {Non-Hermitian Skin Effect and Chiral Damping in Open Quantum Systems},
  volume = {123},
  ISSN = {1079-7114},
  url = {http://dx.doi.org/10.1103/PhysRevLett.123.170401},
  DOI = {10.1103/physrevlett.123.170401},
  pages = {170401},
  number = {17},
  journal = {Phys. Rev. Lett.},
  publisher = {American Physical Society (APS)},
  author = {Song,  Fei and Yao,  Shunyu and Wang,  Zhong},
  year = {2019},
  month = oct 
}

@article{Okuma2020,
  title = {Topological Origin of Non-Hermitian Skin Effects},
  volume = {124},
  ISSN = {1079-7114},
  url = {http://dx.doi.org/10.1103/PhysRevLett.124.086801},
  DOI = {10.1103/physrevlett.124.086801},
  pages = {086801},
  number = {8},
  journal = {Phys. Rev. Lett.},
  publisher = {American Physical Society (APS)},
  author = {Okuma,  Nobuyuki and Kawabata,  Kohei and Shiozaki,  Ken and Sato,  Masatoshi},
  year = {2020},
  month = feb 
}

@article{Yi2020,
  title = {Non-Hermitian Skin Modes Induced by On-Site Dissipations and Chiral Tunneling Effect},
  volume = {125},
  ISSN = {1079-7114},
  url = {http://dx.doi.org/10.1103/PhysRevLett.125.186802},
  DOI = {10.1103/physrevlett.125.186802},
  pages = {186802},
  number = {18},
  journal = {Phys. Rev. Lett.},
  publisher = {American Physical Society (APS)},
  author = {Yi,  Yifei and Yang,  Zhesen},
  year = {2020},
  month = oct 
}

@article{Longhi2022,
  title = {Self-Healing of Non-Hermitian Topological Skin Modes},
  volume = {128},
  ISSN = {1079-7114},
  url = {http://dx.doi.org/10.1103/PhysRevLett.128.157601},
  DOI = {10.1103/physrevlett.128.157601},
  pages = {157601},
  number = {15},
  journal = {Phys. Rev. Lett.},
  publisher = {American Physical Society (APS)},
  author = {Longhi,  Stefano},
  year = {2022},
  month = apr 
}

@article{Minganti2019,
  title = {Quantum exceptional points of non-Hermitian Hamiltonians and Liouvillians: The effects of quantum jumps},
  volume = {100},
  ISSN = {2469-9934},
  url = {http://dx.doi.org/10.1103/PhysRevA.100.062131},
  DOI = {10.1103/physreva.100.062131},
  pages = {062131},
  number = {6},
  journal = {Phys. Rev. A},
  publisher = {American Physical Society (APS)},
  author = {Minganti,  Fabrizio and Miranowicz,  Adam and Chhajlany,  Ravindra W. and Nori,  Franco},
  year = {2019},
  month = dec 
}

@article{Motta2019,
  title = {Determining eigenstates and thermal states on a quantum computer using quantum imaginary time evolution},
  volume = {16},
  ISSN = {1745-2481},
  url = {http://dx.doi.org/10.1038/s41567-019-0704-4},
  DOI = {10.1038/s41567-019-0704-4},
  number = {2},
  journal = {Nat. Phys.},
  publisher = {Springer Science and Business Media LLC},
  author = {Motta,  Mario and Sun,  Chong and Tan,  Adrian T. K. and O'Rourke,  Matthew J. and Ye,  Erika and Minnich,  Austin J. and Brand{\~a}o,  Fernando G. S. L. and Chan,  Garnet Kin-Lic},
  year = {2020},
  month = nov,
  pages = {205--210}
}

@article{McArdle2019,
  title = {Variational ansatz-based quantum simulation of imaginary time evolution},
  volume = {5},
  ISSN = {2056-6387},
  url = {http://dx.doi.org/10.1038/s41534-019-0187-2},
  DOI = {10.1038/s41534-019-0187-2},
  pages = {75},
  number = {1},
  journal = {npj Quantum Inf.},
  publisher = {Springer Science and Business Media LLC},
  author = {McArdle,  Sam and Jones,  Tyson and Endo,  Suguru and Li,  Ying and Benjamin,  Simon C. and Yuan,  Xiao},
  year = {2019},
  month = sep 
}

@article{Kosugi2022,
  title = {Imaginary-time evolution using forward and backward real-time evolution with a single ancilla: First-quantized eigensolver algorithm for quantum chemistry},
  volume = {4},
  ISSN = {2643-1564},
  url = {http://dx.doi.org/10.1103/PhysRevResearch.4.033121},
  DOI = {10.1103/physrevresearch.4.033121},
  pages = {033121},
  number = {3},
  journal = {Phys. Rev. Res.},
  publisher = {American Physical Society (APS)},
  author = {Kosugi,  Taichi and Nishiya,  Yusuke and Nishi,  Hirofumi and Matsushita,  Yu-ichiro},
  year = {2022},
  month = aug 
}

@article{Silva2023,
  title = {Fragmented imaginary-time evolution for early-stage quantum signal processors},
  volume = {13},
  ISSN = {2045-2322},
  url = {http://dx.doi.org/10.1038/s41598-023-45540-2},
  DOI = {10.1038/s41598-023-45540-2},
  pages = {18258},
  number = {1},
  journal = {Sci. Rep.},
  publisher = {Springer Science and Business Media LLC},
  author = {Silva,  Thais L. and Taddei,  M{\'a}rcio M. and Carrazza,  Stefano and Aolita,  Leandro},
  year = {2023},
  month = oct 
}

@article{Coopmans2023,
  title = {Predicting Gibbs-State Expectation Values with Pure Thermal Shadows},
  volume = {4},
  ISSN = {2691-3399},
  url = {http://dx.doi.org/10.1103/PRXQuantum.4.010305},
  DOI = {10.1103/prxquantum.4.010305},
  pages = {010305},
  number = {1},
  journal = {PRX Quantum},
  publisher = {American Physical Society (APS)},
  author = {Coopmans,  Luuk and Kikuchi,  Yuta and Benedetti,  Marcello},
  year = {2023},
  month = jan 
}

@misc{kuji2025vqnhite,
      title={Variational quantum-neural hybrid imaginary time evolution}, 
      author={Hiroki Kuji and Tetsuro Nikuni and Yuta Shingu},
      year={2025},
      eprint={2503.22570},
      archivePrefix={arXiv},
      primaryClass={quant-ph},
      url={https://arxiv.org/abs/2503.22570}, 
}

@article{Li2017,
  title = {Efficient Variational Quantum Simulator Incorporating Active Error Minimization},
  volume = {7},
  ISSN = {2160-3308},
  url = {http://dx.doi.org/10.1103/PhysRevX.7.021050},
  DOI = {10.1103/physrevx.7.021050},
  pages = {021050},
  number = {2},
  journal = {Phys. Rev. X},
  publisher = {American Physical Society (APS)},
  author = {Li,  Ying and Benjamin,  Simon C.},
  year = {2017},
  month = jun 
}

@article{Temme2017,
  title = {Error Mitigation for Short-Depth Quantum Circuits},
  volume = {119},
  ISSN = {1079-7114},
  url = {http://dx.doi.org/10.1103/PhysRevLett.119.180509},
  DOI = {10.1103/physrevlett.119.180509},
  pages = {180509},
  number = {18},
  journal = {Phys. Rev. Lett.},
  publisher = {American Physical Society (APS)},
  author = {Temme,  Kristan and Bravyi,  Sergey and Gambetta,  Jay M.},
  year = {2017},
  month = nov 
}

@article{Han2021,
  title = {Experimental Simulation of Open Quantum System Dynamics via Trotterization},
  volume = {127},
  ISSN = {1079-7114},
  url = {http://dx.doi.org/10.1103/PhysRevLett.127.020504},
  DOI = {10.1103/physrevlett.127.020504},
  pages = {020504},
  number = {2},
  journal = {Phys. Rev. Lett.},
  publisher = {American Physical Society (APS)},
  author = {Han,  J. and Cai,  W. and Hu,  L. and Mu,  X. and Ma,  Y. and Xu,  Y. and Wang,  W. and Wang,  H. and Song,  Y.~P. and Zou,  C.-L. and Sun,  L.},
  year = {2021},
  month = jul 
}

@article{Wang2019,
  title = {Non-Hermitian dynamics without dissipation in quantum systems},
  volume = {99},
  ISSN = {2469-9934},
  url = {http://dx.doi.org/10.1103/PhysRevA.99.063834},
  DOI = {10.1103/physreva.99.063834},
  pages = {063834},
  number = {6},
  journal = {Phys. Rev. A},
  publisher = {American Physical Society (APS)},
  author = {Wang,  Yu-Xin and Clerk,  A. A.},
  year = {2019},
  month = jun 
}

@article{Gorini1976,
  title = {Completely positive dynamical semigroups of N-level systems},
  volume = {17},
  ISSN = {1089-7658},
  url = {http://dx.doi.org/10.1063/1.522979},
  DOI = {10.1063/1.522979},
  number = {5},
  journal = {J. Math. Phys.},
  publisher = {AIP Publishing},
  author = {Gorini,  Vittorio and Kossakowski,  Andrzej and Sudarshan,  E. C. G.},
  year = {1976},
  month = may,
  pages = {821--825}
}

@article{Jebraeilli2025,
  title = {Quantum simulation of a qubit with a non-Hermitian Hamiltonian},
  volume = {111},
  ISSN = {2469-9934},
  url = {http://dx.doi.org/10.1103/PhysRevA.111.032211},
  DOI = {10.1103/physreva.111.032211},
  pages = {032211},
  number = {3},
  journal = {Phys. Rev. A},
  publisher = {American Physical Society (APS)},
  author = {Jebraeilli,  Anastashia and Geller,  Michael R.},
  year = {2025},
  month = mar 
}

@article{Liu2023,
  title = {Practical quantum simulation of small-scale non-Hermitian dynamics},
  volume = {107},
  ISSN = {2469-9934},
  url = {http://dx.doi.org/10.1103/PhysRevA.107.062608},
  DOI = {10.1103/physreva.107.062608},
  pages = {062608},
  number = {6},
  journal = {Phys. Rev. A},
  publisher = {American Physical Society (APS)},
  author = {Liu,  Hongfeng and Yang,  Xiaodong and Tang,  Kai and Che,  Liangyu and Nie,  Xinfang and Xin,  Tao and Li,  Jun and Lu,  Dawei},
  year = {2023},
  month = jun 
}

@article{Wen2019,
  title = {Experimental demonstration of a digital quantum simulation of a general $\mathcal{PT}$-symmetric system},
  volume = {99},
  ISSN = {2469-9934},
  url = {http://dx.doi.org/10.1103/PhysRevA.99.062122},
  DOI = {10.1103/physreva.99.062122},
  pages = {062122},
  number = {6},
  journal = {Phys. Rev. A},
  publisher = {American Physical Society (APS)},
  author = {Wen,  Jingwei and Zheng,  Chao and Kong,  Xiangyu and Wei,  Shijie and Xin,  Tao and Long,  Guilu},
  year = {2019},
  month = jun 
}

@article{Fan2025,
  title = {Solving non-Hermitian physics for optical manipulation on a quantum computer},
  volume = {14},
  ISSN = {2047-7538},
  url = {http://dx.doi.org/10.1038/s41377-025-01769-2},
  DOI = {10.1038/s41377-025-01769-2},
  pages = {132},
  number = {1},
  journal = {Light Sci. Appl.},
  publisher = {Springer Science and Business Media LLC},
  author = {Fan,  Yu-ang and Li,  Xiao and Wei,  Shijie and Li,  Yishan and Long,  Xinyue and Liu,  Hongfeng and Nie,  Xinfang and Ng,  Jack and Lu,  Dawei},
  year = {2025},
  month = mar 
}

@article{Zheng2021,
  title = {Universal quantum simulation of single-qubit nonunitary operators using duality quantum algorithm},
  volume = {11},
  ISSN = {2045-2322},
  url = {http://dx.doi.org/10.1038/s41598-021-83521-5},
  DOI = {10.1038/s41598-021-83521-5},
  pages = {3960},
  number = {1},
  journal = {Sci. Rep.},
  publisher = {Springer Science and Business Media LLC},
  author = {Zheng,  Chao},
  year = {2021},
  month = feb 
}

@article{An2025,
  title = {Quantum Algorithm for Linear Non-unitary Dynamics with Near-Optimal Dependence on All Parameters},
  volume = {407},
  ISSN = {1432-0916},
  url = {http://dx.doi.org/10.1007/s00220-025-05509-w},
  DOI = {10.1007/s00220-025-05509-w},
  pages = {19},
  number = {1},
  journal = {Commun. Math. Phys.},
  publisher = {Springer Science and Business Media LLC},
  author = {An,  Dong and Childs,  Andrew M. and Lin,  Lin},
  year = {2026},
  month = dec 
}

@article{Gupta2024,
  title = {Quantum jumps in driven-dissipative disordered many-body systems},
  volume = {109},
  ISSN = {2469-9934},
  url = {http://dx.doi.org/10.1103/PhysRevA.109.L050201},
  DOI = {10.1103/physreva.109.l050201},
  pages = {L050201},
  number = {5},
  journal = {Phys. Rev. A},
  publisher = {American Physical Society (APS)},
  author = {Gupta,  Sparsh and Yadalam,  Hari Kumar and Kulkarni,  Manas and Aron,  Camille},
  year = {2024},
  month = may 
}

@article{Garrahan2010,
  title = {Thermodynamics of Quantum Jump Trajectories},
  volume = {104},
  ISSN = {1079-7114},
  url = {http://dx.doi.org/10.1103/PhysRevLett.104.160601},
  DOI = {10.1103/physrevlett.104.160601},
  pages = {160601},
  number = {16},
  journal = {Phys. Rev. Lett.},
  publisher = {American Physical Society (APS)},
  author = {Garrahan,  Juan P. and Lesanovsky,  Igor},
  year = {2010},
  month = apr 
}

@article{Liu2025,
  title = {Lindbladian dynamics with loss of quantum jumps},
  volume = {111},
  ISSN = {2469-9969},
  url = {http://dx.doi.org/10.1103/PhysRevB.111.024303},
  DOI = {10.1103/physrevb.111.024303},
  pages = {024303},
  number = {2},
  journal = {Phys. Rev. B},
  publisher = {American Physical Society (APS)},
  author = {Liu,  Yu-Guo and Chen,  Shu},
  year = {2025},
  month = jan 
}

@article{Lindblad1976,
  title = {On the generators of quantum dynamical semigroups},
  volume = {48},
  ISSN = {1432-0916},
  url = {http://dx.doi.org/10.1007/BF01608499},
  DOI = {10.1007/bf01608499},
  number = {2},
  journal = {Commun. Math. Phys.},
  publisher = {Springer Science and Business Media LLC},
  author = {Lindblad,  G.},
  year = {1976},
  month = jun,
  pages = {119--130}
}

@book{Breuer2007,
    author = {Breuer, Heinz-Peter and Petruccione, Francesco},
    title = {The Theory of Open Quantum Systems},
    publisher = {Oxford University Press},
    year = {2007},
    month = {01},
    isbn = {9780199213900},
    doi = {10.1093/acprof:oso/9780199213900.001.0001},
    url = {https://doi.org/10.1093/acprof:oso/9780199213900.001.0001},
}

@article{Dalibard1992,
  title = {Wave-function approach to dissipative processes in quantum optics},
  author = {Dalibard, Jean and Castin, Yvan and M\o{}lmer, Klaus},
  journal = {Phys. Rev. Lett.},
  volume = {68},
  issue = {5},
  pages = {580--583},
  numpages = {0},
  year = {1992},
  month = {Feb},
  publisher = {American Physical Society},
  doi = {10.1103/PhysRevLett.68.580},
  url = {https://link.aps.org/doi/10.1103/PhysRevLett.68.580}
}

@article{Plenio1998,
  title = {The quantum-jump approach to dissipative dynamics in quantum optics},
  author = {Plenio, M. B. and Knight, P. L.},
  journal = {Rev. Mod. Phys.},
  volume = {70},
  issue = {1},
  pages = {101--144},
  numpages = {0},
  year = {1998},
  month = {Jan},
  publisher = {American Physical Society},
  doi = {10.1103/RevModPhys.70.101},
  url = {https://link.aps.org/doi/10.1103/RevModPhys.70.101}
}

@article{Bergholtz2021,
  title = {Exceptional topology of non-Hermitian systems},
  author = {Bergholtz, Emil J. and Budich, Jan Carl and Kunst, Flore K.},
  journal = {Rev. Mod. Phys.},
  volume = {93},
  issue = {1},
  pages = {015005},
  numpages = {31},
  year = {2021},
  month = {Feb},
  publisher = {American Physical Society},
  doi = {10.1103/RevModPhys.93.015005},
  url = {https://link.aps.org/doi/10.1103/RevModPhys.93.015005}
}

@article{Hatano1997PRB,
  title = {Vortex pinning and non-Hermitian quantum mechanics},
  author = {Hatano, Naomichi and Nelson, David R.},
  journal = {Phys. Rev. B},
  volume = {56},
  issue = {14},
  pages = {8651--8673},
  numpages = {0},
  year = {1997},
  month = {Oct},
  publisher = {American Physical Society},
  doi = {10.1103/PhysRevB.56.8651},
  url = {https://link.aps.org/doi/10.1103/PhysRevB.56.8651}
}

@article{Hatano1996PRL,
  title = {Localization Transitions in Non-Hermitian Quantum Mechanics},
  author = {Hatano, Naomichi and Nelson, David R.},
  journal = {Phys. Rev. Lett.},
  volume = {77},
  issue = {3},
  pages = {570--573},
  numpages = {0},
  year = {1996},
  month = {Jul},
  publisher = {American Physical Society},
  doi = {10.1103/PhysRevLett.77.570},
  url = {https://link.aps.org/doi/10.1103/PhysRevLett.77.570}
}

@article{Berry2015,
  title = {Simulating Hamiltonian Dynamics with a Truncated Taylor Series},
  author = {Berry, Dominic W. and Childs, Andrew M. and Cleve, Richard and Kothari, Robin and Somma, Rolando D.},
  journal = {Phys. Rev. Lett.},
  volume = {114},
  issue = {9},
  pages = {090502},
  numpages = {5},
  year = {2015},
  month = {Mar},
  publisher = {American Physical Society},
  doi = {10.1103/PhysRevLett.114.090502},
  url = {https://link.aps.org/doi/10.1103/PhysRevLett.114.090502}
}

@article{Childs2012,
  title = {Hamiltonian simulation using linear combinations of unitary operations},
  volume = {12},
  ISSN = {1533-7146},
  url = {http://dx.doi.org/10.26421/QIC12.11-12-1},
  DOI = {10.26421/qic12.11-12-1},
  number = {11 & 12},
  journal = {Quantum Inf. Comput.},
  publisher = {Rinton Press},
  author = {Childs,  Andrew M. and Wiebe,  Nathan},
  year = {2012},
  month = nov,
  pages = {901--924}
}

@article{Low2017,
  title = {Optimal Hamiltonian Simulation by Quantum Signal Processing},
  author = {Low, Guang Hao and Chuang, Isaac L.},
  journal = {Phys. Rev. Lett.},
  volume = {118},
  issue = {1},
  pages = {010501},
  numpages = {5},
  year = {2017},
  month = {Jan},
  publisher = {American Physical Society},
  doi = {10.1103/PhysRevLett.118.010501},
  url = {https://link.aps.org/doi/10.1103/PhysRevLett.118.010501}
}

@inproceedings{Gilyn2019,
  series = {STOC '19},
  title = {Quantum singular value transformation and beyond: exponential improvements for quantum matrix arithmetics},
  url = {http://dx.doi.org/10.1145/3313276.3316366},
  DOI = {10.1145/3313276.3316366},
  booktitle = {Proceedings of the 51st Annual ACM SIGACT Symposium on Theory of Computing},
  publisher = {ACM},
  author = {Gily{\'e}n,  Andr{\'a}s and Su,  Yuan and Low,  Guang Hao and Wiebe,  Nathan},
  year = {2019},
  month = jun,
  pages = {193--204},
  collection = {STOC '19}
}

@article{Schlimgen2022,
  title = {Quantum simulation of the Lindblad equation using a unitary decomposition of operators},
  author = {Schlimgen, Anthony W. and Head-Marsden, Kade and Sager, LeeAnn M. and Narang, Prineha and Mazziotti, David A.},
  journal = {Phys. Rev. Res.},
  volume = {4},
  issue = {2},
  pages = {023216},
  numpages = {7},
  year = {2022},
  month = {Jun},
  publisher = {American Physical Society},
  doi = {10.1103/PhysRevResearch.4.023216},
  url = {https://link.aps.org/doi/10.1103/PhysRevResearch.4.023216}
}

@misc{wang2026nonhermitirydberg,
      title={Non-Hermitian physics in the many-body system of Rydberg atoms}, 
      author={Ya-Jun Wang and Jun Zhang and Dong-Sheng Ding},
      year={2026},
      eprint={2602.07372},
      archivePrefix={arXiv},
      primaryClass={cond-mat.quant-gas},
      url={https://arxiv.org/abs/2602.07372}, 
}

@article{Vikstl2024,
  title = {Study of noise in virtual distillation circuits for quantum error mitigation},
  volume = {8},
  ISSN = {2521-327X},
  url = {http://dx.doi.org/10.22331/q-2024-08-14-1441},
  DOI = {10.22331/q-2024-08-14-1441},
  journal = {Quantum},
  publisher = {Verein zur Forderung des Open Access Publizierens in den Quantenwissenschaften},
  author = {Vikst{\aa}l,  Pontus and Ferrini,  Giulia and Puri,  Shruti},
  year = {2024},
  month = aug,
  pages = {1441}
}

@article{vanKempen2000,
  title = {Mean and variance of ratio estimators used in fluorescence ratio imaging},
  volume = {39},
  ISSN = {1097-0320},
  url = {http://dx.doi.org/10.1002/(SICI)1097-0320(20000401)39:4<300::AID-CYTO8>3.0.CO;2-O},
  DOI = {10.1002/(sici)1097-0320(20000401)39:4<300::aid-cyto8>3.0.co;2-o},
  number = {4},
  journal = {Cytometry},
  publisher = {Wiley},
  author = {van Kempen,  G.M.P. and van Vliet,  L.J.},
  year = {2000},
  month = apr,
  pages = {300--305}
}

@article{Schmolke2022,
  title = {Noise-Induced Quantum Synchronization},
  volume = {129},
  ISSN = {1079-7114},
  url = {http://dx.doi.org/10.1103/PhysRevLett.129.250601},
  DOI = {10.1103/physrevlett.129.250601},
  pages = {250601},
  number = {25},
  journal = {Phys. Rev. Lett.},
  publisher = {American Physical Society (APS)},
  author = {Schmolke,  Finn and Lutz,  Eric},
  year = {2022},
  month = dec 
}

@article{Takasu2020,
  title = {PT-symmetric non-Hermitian quantum many-body system using ultracold atoms in an optical lattice with controlled dissipation},
  volume = {2020},
  ISSN = {2050-3911},
  url = {http://dx.doi.org/10.1093/ptep/ptaa094},
  DOI = {10.1093/ptep/ptaa094},
  pages = {12A110},
  number = {12},
  journal = {Prog. Theor. Exp. Phys.},
  publisher = {Oxford University Press (OUP)},
  author = {Takasu,  Yosuke and Yagami,  Tomoya and Ashida,  Yuto and Hamazaki,  Ryusuke and Kuno,  Yoshihito and Takahashi,  Yoshiro},
  year = {2020},
  month = sep 
}

@article{Li2019,
  title = {Observation of parity-time symmetry breaking transitions in a dissipative Floquet system of ultracold atoms},
  volume = {10},
  ISSN = {2041-1723},
  url = {http://dx.doi.org/10.1038/s41467-019-08596-1},
  DOI = {10.1038/s41467-019-08596-1},
  pages = {855},
  number = {1},
  journal = {Nat. Commun.},
  publisher = {Springer Science and Business Media LLC},
  author = {Li,  Jiaming and Harter,  Andrew K. and Liu,  Ji and de Melo,  Leonardo and Joglekar,  Yogesh N. and Luo,  Le},
  year = {2019},
  month = feb 
}

@article{Chen2021,
  title = {Quantum Jumps in the Non-Hermitian Dynamics of a Superconducting Qubit},
  volume = {127},
  ISSN = {1079-7114},
  url = {http://dx.doi.org/10.1103/PhysRevLett.127.140504},
  DOI = {10.1103/physrevlett.127.140504},
  pages = {140504},
  number = {14},
  journal = {Phys. Rev. Lett.},
  publisher = {American Physical Society (APS)},
  author = {Chen,  Weijian and Abbasi,  Maryam and Joglekar,  Yogesh N. and Murch,  Kater W.},
  year = {2021},
  month = sep 
}

@article{Xiao2017,
  title = {Observation of topological edge states in parity--time-symmetric quantum walks},
  volume = {13},
  ISSN = {1745-2481},
  url = {http://dx.doi.org/10.1038/nphys4204},
  DOI = {10.1038/nphys4204},
  number = {11},
  journal = {Nat. Phys.},
  publisher = {Springer Science and Business Media LLC},
  author = {Xiao,  L. and Zhan,  X. and Bian,  Z. H. and Wang,  K. K. and Zhang,  X. and Wang,  X. P. and Li,  J. and Mochizuki,  K. and Kim,  D. and Kawakami,  N. and Yi,  W. and Obuse,  H. and Sanders,  B. C. and Xue,  P.},
  year = {2017},
  month = jul,
  pages = {1117--1123}
}

@article{Wu2024,
  title = {Third-order exceptional line in a nitrogen-vacancy spin system},
  volume = {19},
  ISSN = {1748-3395},
  url = {http://dx.doi.org/10.1038/s41565-023-01583-0},
  DOI = {10.1038/s41565-023-01583-0},
  number = {2},
  journal = {Nat. Nanotechnol.},
  publisher = {Springer Science and Business Media LLC},
  author = {Wu,  Yang and Wang,  Yunhan and Ye,  Xiangyu and Liu,  Wenquan and Niu,  Zhibo and Duan,  Chang-Kui and Wang,  Ya and Rong,  Xing and Du,  Jiangfeng},
  year = {2024},
  month = jan,
  pages = {160--165}
}

@article{Chen2025EP,
  title = {Quantum tomography of a third-order exceptional point in a dissipative trapped ion},
  volume = {16},
  ISSN = {2041-1723},
  url = {http://dx.doi.org/10.1038/s41467-025-62573-5},
  DOI = {10.1038/s41467-025-62573-5},
  pages = {7478},
  number = {1},
  journal = {Nat. Commun.},
  publisher = {Springer Science and Business Media LLC},
  author = {Chen,  Y.-Y. and Li,  K. and Zhang,  L. and Wu,  Y.-K. and Ma,  J.-Y. and Yang,  H.-X. and Zhang,  C. and Qi,  B.-X. and Zhou,  Z.-C. and Hou,  P.-Y. and Xu,  Y. and Duan,  L.-M.},
  year = {2025},
  month = aug 
}

@article{Liu2021,
  title = {Dynamically Encircling an Exceptional Point in a Real Quantum System},
  volume = {126},
  ISSN = {1079-7114},
  url = {http://dx.doi.org/10.1103/PhysRevLett.126.170506},
  DOI = {10.1103/physrevlett.126.170506},
  pages = {170506},
  number = {17},
  journal = {Phys. Rev. Lett.},
  publisher = {American Physical Society (APS)},
  author = {Liu,  Wenquan and Wu,  Yang and Duan,  Chang-Kui and Rong,  Xing and Du,  Jiangfeng},
  year = {2021},
  month = apr 
}

@article{Ren2022,
  title = {Chiral control of quantum states in non-Hermitian spin--orbit-coupled fermions},
  volume = {18},
  ISSN = {1745-2481},
  url = {http://dx.doi.org/10.1038/s41567-021-01491-x},
  DOI = {10.1038/s41567-021-01491-x},
  number = {4},
  journal = {Nat. Phys.},
  publisher = {Springer Science and Business Media LLC},
  author = {Ren,  Zejian and Liu,  Dong and Zhao,  Entong and He,  Chengdong and Pak,  Ka Kwan and Li,  Jensen and Jo,  Gyu-Boong},
  year = {2022},
  month = jan,
  pages = {385--389}
}

@article{Gong2018,
  title = {Topological Phases of Non-Hermitian Systems},
  volume = {8},
  ISSN = {2160-3308},
  url = {http://dx.doi.org/10.1103/PhysRevX.8.031079},
  DOI = {10.1103/physrevx.8.031079},
  pages = {031079},
  number = {3},
  journal = {Phys. Rev. X},
  publisher = {American Physical Society (APS)},
  author = {Gong,  Zongping and Ashida,  Yuto and Kawabata,  Kohei and Takasan,  Kazuaki and Higashikawa,  Sho and Ueda,  Masahito},
  year = {2018},
  month = sep 
}

@book{supple,
  title = {\rm{See Supplemental Material at [URL will be inserted by publisher] for further details.}},
}

@article{PRXQuantum.4.040329,
  title = {Noise-Assisted Digital Quantum Simulation of Open Systems Using Partial Probabilistic Error Cancellation},
  author = {Guimar\~aes, Jos\'e D. and Lim, James and Vasilevskiy, Mikhail I. and Huelga, Susana F. and Plenio, Martin B.},
  journal = {PRX Quantum},
  volume = {4},
  issue = {4},
  pages = {040329},
  numpages = {19},
  year = {2023},
  month = {Nov},
  publisher = {American Physical Society},
  doi = {10.1103/PRXQuantum.4.040329},
  url = {https://link.aps.org/doi/10.1103/PRXQuantum.4.040329}
}

@Article{ScaleFreeNHSE,
	title={{Scale-free non-Hermitian skin effect in a boundary-dissipated spin chain}},
	author={He-Ran Wang and Bo Li and Fei Song and Zhong Wang},
	journal={SciPost Phys.},
	volume={15},
	pages={191},
	year={2023},
	publisher={SciPost},
	doi={10.21468/SciPostPhys.15.5.191},
	url={https://scipost.org/10.21468/SciPostPhys.15.5.191},
}
\clearpage
\newpage

\normalsize \renewcommand{\theequation}{A\arabic{equation}} \setcounter{equation}{0}

\onecolumngrid
\section*{End Matter}\label{sec:end_matter}
\twocolumngrid
\emph{The construction of $H_{\mathrm{I},\ell}$ from $H_{\mathrm{Im}}$---}
Here, we describe the construction of the stochastic generators used to simulate the target non-Hermitian dynamics generated by $H_{\mathrm{target}} = H_{\mathrm{Re}} + i H_{\mathrm{Im}}$.
We assume that $H_{\mathrm{Im}}$ is decomposed into local terms as
\begin{align}
  H_{\mathrm{Im}} = \sum_{\ell \in E} H_{\mathrm{Im},\ell},
  \label{eq:HIm_decomp}
\end{align}
where each $H_{\mathrm{Im},\ell}$ acts only on the sites specified by $\ell$ and $E$ denotes the set of such local labels.
For this purpose, we impose the condition up to an identity shift
\begin{align}
  H_{\mathrm{Im}} - E_{\mathrm{shift}} I = - \sum_{\ell \in E} \gamma_{\ell} H_{\mathrm{I},\ell}^{2},
  \label{eq:condition_Hi_multi}
\end{align}
where $E_{\mathrm{shift}}$ is chosen such that the left-hand side is negative semidefinite.
Note that the shift $E_{\mathrm{shift}}I$ only rescales the overall norm and does not affect values of observables with respect to normalized time-evolved states.

To determine $E_\mathrm{shift}$, recall that the Hamiltonian $H_\mathrm{Im}$ is decomposed as Eq.~\eqref{eq:HIm_decomp}.
For each $\ell$, we take the spectral decomposition of $H_{\mathrm{Im},\ell}$ on the local Hilbert space $\mathcal{H}_\ell$ defined on the sites in $\ell$: $H_{\mathrm{Im},\ell} = U_{\ell} \mathrm{diag}(\mu_{\ell,z}) U_{\ell}^{\dagger}$.
Here, $\{\mu_{\ell,z}\}_z$ is the set of eigenvalues of $ H_{\mathrm{Im},\ell}$.
Let us choose a local shift $e_{\ell}$ satisfying $e_{\ell}\ge \max_z(\mu_{\ell,z})$, so that $-e_{\ell}I_{\ell}+H_{\mathrm{Im},\ell}$ is negative semidefinite.
Then define $H_{\mathrm{I},\ell}= U_{\ell} \mathrm{diag}(\lambda_{\ell,z}) U_{\ell}^{\dagger}$ and $\lambda_{\ell,z} = \sqrt{(e_{\ell}-\mu_{\ell,z})/\gamma_{\ell}}$.
By construction,
\begin{align}
  H_{\mathrm{I},\ell}^{2}
  = U_{\ell} \mathrm{diag}\left(\frac{e_{\ell}-\mu_{\ell,z}}{\gamma_{\ell}}\right)U_{\ell}^{\dagger}
  = \frac{e_{\ell}I_{\ell}-H_{\mathrm{Im},\ell}}{\gamma_{\ell}},
  \label{eq:Hi_local_sq}
\end{align}
where $I_{\ell}$ is the identity operator on $\mathcal{H}_\ell$.
Summing Eq.~\eqref{eq:Hi_local_sq} over $\ell$ and putting the local identities $I_{\ell}$ into the full Hilbert space yields $- \sum_{\ell \in E} \gamma_{\ell} H_{\mathrm{I},\ell}^{2}=\sum_{\ell \in E} H_{\mathrm{Im},\ell} - \Big(\sum_{\ell \in E} e_{\ell}\Big) I=H_{\mathrm{Im}} - E_{\mathrm{shift}} I$, with $E_{\mathrm{shift}}:=\sum_{\ell \in E} e_{\ell}$, \hk{thereby satisfying} Eq.~\eqref{eq:condition_Hi_multi}.
This construction keeps each $H_{\mathrm{I},\ell}$ strictly local with the same locality pattern as the decomposition in Eq.~\eqref{eq:HIm_decomp}, enabling an implementation compatible with digital-analog quantum computing with independently driven stochastic unitarities on each \hk{local support}.

\medskip
\emph{Local basis operations for sQEM---}
We use experimentally available single-qubit basis operation channels for quasi-probability decompositions.
Following Ref.~\cite{Endo2018}, we employ the sixteen basis operations listed in Table~\ref{tab:bases}.
The notation $[A]$ denotes the channel $\rho\mapsto A\rho A^\dagger$.

We construct local basis operation channels on each \hk{local support} $\ell\in E$ by tensor products of the single-qubit basis operation channels acting on the \hk{sites in $\ell$}.
We use a single index $j$ to enumerate a fixed list of \hk{local} basis operation channels on \hk{support} $\ell$ and write them as $\{\mathcal{B}_{\ell j}\}_{j\ge 0}$.
We take $j=0$ for the identity channel and define, for all $\ell\in E$, $\mathcal{B}_{\ell,0}\coloneqq \mathcal{I}$ with the parity $\alpha_{\ell,0}\coloneqq 1$.
Since $\mathcal{B}_{\ell,0}$ and $\alpha_{\ell,0}$ are independent of $\ell$, we also write $\mathcal{B}_0\coloneqq \mathcal{I}$ and $\alpha_0\coloneqq 1$.
\begin{table*}[t]
\begin{center}
\begin{tabular}{clclclcl}
\toprule
0 & ~~$[I]$ (no operation) &
1 & ~~$[\sigma^{\mathrm{x}}] $ &
2 & ~~$[\sigma^{\mathrm{y}}] $ &
3 & ~~$[\sigma^{\mathrm{z}}] $ \\
4 & ~~$[R_{\mathrm{x}}] = [\frac{1}{\sqrt{2}}(I +i \sigma^{\mathrm{x}})]$ &
5 & ~~$[R_{\mathrm{y}}] = [\frac{1}{\sqrt{2}}(I +i \sigma^{\mathrm{y}})] $ &
6 & ~~$[R_{\mathrm{z}}] = [\frac{1}{\sqrt{2}}(I+i \sigma^{\mathrm{z}})] $ &
7 & ~~$[R_{\mathrm{y}z}] = [\frac{1}{\sqrt{2}}(\sigma^{\mathrm{y}} + \sigma^{\mathrm{z}})]$ \\
8 & ~~$[R_{\mathrm{z}x}] = [\frac{1}{\sqrt{2}}(\sigma^{\mathrm{z}} + \sigma^{\mathrm{x}})] $ &
9 & ~~$[R_{\mathrm{x}y}] = [\frac{1}{\sqrt{2}}(\sigma^{\mathrm{x}} + \sigma^{\mathrm{y}})]$ &
10 & ~~$[\pi_{\mathrm{x}}] = [\frac{1}{2}(I + \sigma^{\mathrm{x}})] $ &
11 & ~~$[\pi_{\mathrm{y}}] = [\frac{1}{2}(I + \sigma^{\mathrm{y}})]$ \\
12 & ~~$[\pi_{\mathrm{z}}] = [\frac{1}{2}(I + \sigma^{\mathrm{z}})] $ &
13 & ~~$[\pi_{\mathrm{y}z}] = [\frac{1}{2}(\sigma^{\mathrm{y}} +i \sigma^{\mathrm{z}})]$ &
14 & ~~$[\pi_{\mathrm{z}x}] = [\frac{1}{2}(\sigma^{\mathrm{z}} +i \sigma^{\mathrm{x}})]$ &
15 & ~~$[\pi_{\mathrm{x}y}] = [\frac{1}{2}(\sigma^{\mathrm{x}} +i \sigma^{\mathrm{y}})] $ \\
\bottomrule
\end{tabular}
\end{center}
\caption{
Sixteen single-qubit basis operation channels~\cite{Endo2018}.
$[I]$ denotes the identity operation (no operation). $[\sigma^\alpha]\quad (\alpha=x,y,z)$ corresponds to applying Pauli matrices.
$[R]$ denotes single-qubit rotations. $[\pi]$ denotes projective measurements.
}
\label{tab:bases}
\end{table*}

\medskip
\emph{Error analysis---}
Here, we \hk{analyze the} root-mean-square error (RMSE) at a final time $T$ \hk{and derive requirements on $\Delta t$ and $N_{\mathrm{traj}}$ for a target accuracy $\varepsilon$}.

Let $\mathbb{E}_{\boldsymbol{\xi}}[\cdot]$ denote the ensemble average over the Gaussian increments\hk{, and let $\mathbb{E}_{\rm all}[\cdot]$ denote the average over both the Gaussian increments and the sQEM sampling}.
\hk{We denote the accumulated sQEM normalization factor by $\Lambda(T)\coloneqq e^{\kappa T}$.}

We \hk{define the normalized expectation values of the target evolution and the Trotterized implementation with infinite sampling} as
\begin{align}
\langle O\rangle_{\mathrm{target}}(T)
&\coloneqq
\frac{\Tr\left[O\rho_{\mathrm{NH}}(T)\right]}
{\Tr\left[\rho_{\mathrm{NH}}(T)\right]},\\
\hk{\langle O\rangle_{\mathrm{impl}}(T)}
&\hk{\coloneqq
\frac{\Tr\left[O\rho_{\mathrm{NH}}^{(\mathrm{impl})}(T)\right]}
{\Tr\left[\rho_{\mathrm{NH}}^{(\mathrm{impl})}(T)\right]}.}
\end{align}
Here, $\rho_{\mathrm{NH}}(T)$ is the ideal unnormalized non-Hermitian state\hk{, whereas $\rho_{\mathrm{NH}}^{(\mathrm{impl})}(T)$ is the corresponding unnormalized state from the Trotterized implementation with infinite sampling}.

For each \hk{trajectory} $s$, let $\alpha_s\in\{\pm1\}$ be its parity, let $D_s(T)\in\{0,1\}$ be its acceptance indicator, and let $O_s(T)\in[-\|O\|,\|O\|]$ be its observable contribution, with $O_s(T)=0$ when $D_s(T)=0$.
We define \hk{$X_s(T)\coloneqq\alpha_sO_s(T)$ and $Y_s(T)\coloneqq\alpha_sD_s(T)$ and write the finite-sample estimator as}
\begin{align}
\hk{\langle O\rangle_{N_{\mathrm{traj}}}(T)
\coloneqq
\frac{\sum_{s=1}^{N_{\mathrm{traj}}}X_s(T)}
{\sum_{s=1}^{N_{\mathrm{traj}}}Y_s(T)}.}
\end{align}

We then define the \hk{overall} RMSE as
\begin{align}
\mathrm{RMSE}_{\mathrm{all}}(T)
\coloneqq
\sqrt{\mathbb{E}_{\mathrm{all}}\left[\bigl(\hk{\langle O\rangle_{N_{\mathrm{traj}}}(T)}-\langle O\rangle_{\mathrm{target}}(T)\bigr)^{2}\right]}.
\label{eq:em_rmse_all_def}
\end{align}
We separate it into the systematic error due to Trotterization and the statistical error due to finite sampling,
\begin{align}\label{eq:eps_sys}
\hk{\epsilon_{\mathrm{sys}}(T)}
&\hk{\coloneqq
\left|\langle O\rangle_{\mathrm{impl}}(T)-\langle O\rangle_{\mathrm{target}}(T)\right|,}\\
\hk{\epsilon_{\mathrm{stat}}(T)}
&\hk{\coloneqq
\sqrt{\mathbb{E}_{\mathrm{all}}\left[
\bigl(\langle O\rangle_{N_{\mathrm{traj}}}(T)-\langle O\rangle_{\mathrm{impl}}(T)\bigr)^2
\right]}.}
\end{align}
\hk{The statistical error includes both the variance and the bias of the finite-sample ratio estimator.
The Cauchy--Schwarz inequality gives
$\mathrm{RMSE}_{\mathrm{all}}(T)
\le
\epsilon_{\mathrm{sys}}(T)+\epsilon_{\mathrm{stat}}(T)$.}
For $\mathrm{RMSE}_{\mathrm{all}}(T)\le \varepsilon$, \hk{it is sufficient to require}
\begin{align}
\epsilon_{\mathrm{sys}}(T)\le \frac{\varepsilon}{2},
\qquad
\epsilon_{\mathrm{stat}}(T)\le \frac{\varepsilon}{2},
\label{eq:em_sufficient_split}
\end{align}

First, we consider the systematic error from \hk{the} Trotter formula between $H_{\mathrm{Re}}$ and $H_{\mathrm{I}}(k)\coloneqq \sum_{\ell \in E}\xi_{k,\ell} H_{\mathrm{I},\ell}$, where $\xi_{k,\ell} \coloneqq \int_{k\Delta t}^{k\Delta t+\Delta t} f_{\ell}(s)ds$, and \hk{from the even-odd bond splitting used for nearest-neighbor terms in} the GKSL channel construction.
The relevant errors are classified by (i) Lie-Trotter splitting between $H_{\mathrm{Re}}$ and $H_{\mathrm{I}}(k)$, (ii) Lie-Trotter even-odd splitting of $H_{\mathrm{Re}}$, and (iii) Suzuki-Trotter even-odd splitting for $H_{\mathrm{I}}(k)$.

\hk{Under the assumptions of Theorem~\ref{thm:bound_sys_error} and Corollary~\ref{coro:STeo_vanish} in \suppl~\cite{supple},} we have
\begin{align}
\epsilon_\mathrm{sys}(T)&\leq
\frac{4\|O\|\hk{\Lambda(T)}\mathbb{E}_{\boldsymbol{\xi}}\left[\sum_{k=0}^{N-1}\left\|U^{\mathrm{impl,ST(eo)}}_k-U^{\mathrm{ex}}_k\right\|\right]}{\min\{|\mathrm{Tr}[\rho_{\rm NH}(T)]|,\,|\mathrm{Tr}[\rho^{\rm (impl)}_{\rm NH}(T)]|\}}
\notag\\
&\le
\frac{4\|O\|\hk{\Lambda(T)}\,C^{\hk{\mathrm{ST(eo)}}}\,T\sqrt{\Delta t}}
{\min\{|\mathrm{Tr}[\rho_{\rm NH}(T)]|,\,|\mathrm{Tr}[\rho^{\rm (impl)}_{\rm NH}(T)]|\}}
\label{eq:suff_sys_app}
\end{align}
for some constant $C^{\mathrm{ST(eo)}}$ independent of $\Delta t$ and $T$.
\hk{Assuming that the denominator remains bounded away from zero as $\Delta t\to0$, this bound can be made smaller than $\varepsilon/2$ by reducing $\Delta t$.}

\hk{Applying Theorem~\ref{thm:rmse_ratio_sqem_full} in \suppl~\cite{supple} to the trajectory pairs $(X_s,Y_s)$ gives}
\begin{align}
\hk{\epsilon_{\mathrm{stat}}(T)}
\hk{\approx}
\hk{
\frac{\sqrt{\mathrm{Var}_{\mathrm{all}}\!\left(
X_s(T)-\langle O\rangle_{\mathrm{impl}}(T)Y_s(T)
\right)}}{
\sqrt{N_{\mathrm{traj}}}\,
\left|\mathbb{E}_{\mathrm{all}}[Y_s(T)]\right|
}.}
\label{eq:joint_ratio_rmse_end_matter}
\end{align}
\hk{Since
$|X_s(T)-\langle O\rangle_{\mathrm{impl}}(T)Y_s(T)|
\le(\|O\|+|\langle O\rangle_{\mathrm{impl}}(T)|)D_s(T)$,
defining $p_{\mathrm{succ}}(T)\coloneqq\mathbb{E}_{\mathrm{all}}[D_s(T)]$ and using
$\Tr[\rho_{\mathrm{NH}}^{(\mathrm{impl})}(T)]=\Lambda(T)\mathbb{E}_{\mathrm{all}}[Y_s(T)]$ give}
\begin{align}
\hk{\epsilon_{\mathrm{stat}}(T)
\ \lesssim\
\frac{\|O\|+|\langle O\rangle_{\mathrm{impl}}(T)|}{\sqrt{N_{\mathrm{traj}}}}
\frac{\Lambda(T)\sqrt{p_{\mathrm{succ}}(T)}}{|\Tr[\rho_{\mathrm{NH}}^{(\mathrm{impl})}(T)]|},}
\label{eq:rmse_bound_psucc_Lambda}
\end{align}
\hk{For the case of the equal error budget in Eq.~\eqref{eq:em_sufficient_split}, requirement of $\epsilon_{\mathrm{stat}}(T)\le \varepsilon/2$ is
satisfied if we take the sample-size as}
\begin{align}
\hk{N_{\mathrm{traj}}}
&\hk{\gtrsim
\left(
\frac{\bigl(\|O\|+|\langle O\rangle_{\mathrm{impl}}(T)|\bigr)\Lambda(T)\sqrt{p_{\mathrm{succ}}(T)}}
{(\varepsilon/2)|\Tr[\rho_{\mathrm{NH}}^{(\mathrm{impl})}(T)]|}
\right)^2 ,}
\label{eq:Ntraj_requirement_end_matter}
\end{align}
where 
we use Eq.~\eqref{eq:rmse_bound_psucc_Lambda}.
\hk{Thus the sQEM \hk{quasi-probability} normalization contributes the factor $\Lambda(T)^2=\exp(2\kappa T)$ to the required number of trajectories, in addition to the non-Hermitian normalization and success-probability factors.}

\clearpage
\newpage

\setcounter{page}{1} \renewcommand{\thepage}{Supplemental Material -- \arabic{page}/\hk{19}}

\title{Supplemental Material for\\[4pt] ``Quantum Error Mitigation Simulates General Non-Hermitian Dynamics''} \renewcommand{\thesection}{SM\arabic{section}} \renewcommand{\theequation}{SM\arabic{equation}} \renewcommand{\thefigure}{SM\arabic{figure}} \renewcommand{\thetable}{SM\arabic{table}} \setcounter{section}{0} \setcounter{equation}{0} \setcounter{figure}{0} \setcounter{table}{0}

\setcounter{secnumdepth}{3}

\maketitle
\onecolumngrid \vspace{-2.5em}
\section{GKSL construction via Gaussian noise averaging}\label{app:gaussian-average}
In this section, we briefly review the stochastic-Hamiltonian noise-averaging construction of a GKSL generator proposed in Ref.~\cite{Chenu2017}, and make the discrete-time implementation used in our setup explicit.
We set $\hbar=1$ throughout.

\subsection{Averaged one-step map}
We consider the multi-channel noise Hamiltonian
\begin{align}
    H_{\mathrm{impl}}(t)=H_{\mathrm{Re}}+\sum_{\ell\in E} f_{\ell}(t)H_{\mathrm{I},\ell},
\end{align}
where $\{f_{\ell}(t)\}_{\ell\in E}$ are independent Gaussian white noises with $\mathbb{E}_f[f_{\ell}(t)]=0$ and $\mathbb{E}_f[f_{\ell}(t)f_{\ell'}(t')]=2\gamma_{\ell} \delta_{\ell\ell'} \delta(t-t')$.

Let $k=0,1,\ldots,N-1$ label discrete time steps with $t_k\coloneqq k\Delta t$ and $N\coloneqq T/\Delta t$, and let $\ell\in E$ label the \hk{local supports}.
We partition the total time into steps of size $\Delta t$ and define the Gaussian increments
\begin{align}
\xi_{k,\ell}\coloneqq \int_{t_k}^{t_k+\Delta t} f_\ell(s) ds,
\qquad t_k=k\Delta t,
\end{align}
so that $\xi_{k,\ell}\sim\mathcal N(0,2\gamma_\ell \Delta t)$, i.e., $\xi_{k,\ell}$ is Gaussian with mean $0$ and variance $2\gamma_\ell \Delta t$, and the set $\boldsymbol{\xi}_k =\{\xi_{k,\ell}\}_{\ell\in E}$ is independent across $\ell$ and also independent across $k$.

For concreteness, we use the Lie-Trotter formula, so that the single-step unitary is
\begin{align}
  U_k^{\mathrm{prod}} \coloneqq e^{-i H_{\mathrm{Re}} \Delta t} V_k^{\mathrm{prod}},
  \qquad
  V_k^{\mathrm{prod}} \coloneqq \prod_{\ell\in E} e^{-i \xi_{k,\ell} H_{\mathrm{I},\ell}},
  \label{eq:app-Uk-multi}
\end{align}
where $\prod_{\ell\in E}$ denotes an ordered product over a fixed order $E=\{\ell_1,\ldots,\ell_{|E|}\}$ chosen once and for all.
Hence, the state update becomes
\begin{align}
\rho_{k+1}
=
U_k^{\mathrm{prod}}\rho_k U_k^{\mathrm{prod}\dagger}
=
e^{-i H_{\mathrm{Re}} \Delta t}
\Bigl(
V_k^{\mathrm{prod}} \rho_k V_k^{\mathrm{prod}\dagger}
\Bigr)
e^{+i H_{\mathrm{Re}} \Delta t}.
\label{eq:app-time-evolved}
\end{align}
For notational convenience, we also define the past and partial increment histories $\boldsymbol{\xi}_{<k}\coloneqq \{\boldsymbol{\xi}_j\}_{j=0}^{k-1}$ and $\boldsymbol{\xi}_{\le k}\coloneqq \{\boldsymbol{\xi}_j\}_{j=0}^{k}$.
Let $\bar{\rho}_k \coloneqq \mathbb{E}_{\boldsymbol{\xi}_{<k}}[\rho_k]$.
Because $\rho_k$ depends only on the past increments $\boldsymbol{\xi}_{<k}$ and the Gaussian increments are independent, the collection $\boldsymbol{\xi}_k$ is independent of $\rho_k$.

Define the one-step noise-averaging channel at step $k$ by
\begin{align}
\Phi^{\mathrm{prod}}_{\Delta t,k}[\rho]
\coloneqq
\mathbb{E}_{\boldsymbol{\xi}_k}
\left[
V_k^{\mathrm{prod}} \rho V_k^{\mathrm{prod}\dagger}
\right].
\label{eq:app-Phi-def}
\end{align}
Therefore,
\begin{align}
\bar{\rho}_{k+1}
&=
\mathbb{E}_{\boldsymbol{\xi}_{\le k}}\left[U_k^{\mathrm{prod}}\rho_k U_k^{\mathrm{prod}\dagger}\right]\notag\\
&=
e^{-iH_{\mathrm{Re}}\Delta t}
\mathbb{E}_{\boldsymbol{\xi}_{< k}}
\left[
\mathbb{E}_{\boldsymbol{\xi}_k}\left[V_k^{\mathrm{prod}} \rho_k V_k^{\mathrm{prod}\dagger}\right]
\right]
e^{+iH_{\mathrm{Re}}\Delta t}\notag\\
&=
e^{-iH_{\mathrm{Re}}\Delta t}
\mathbb{E}_{\boldsymbol{\xi}_{< k}}
\left[
\Phi^{\mathrm{prod}}_{\Delta t,k}[\rho_k]
\right]
e^{+iH_{\mathrm{Re}}\Delta t}\notag\\
&=
e^{-iH_{\mathrm{Re}}\Delta t}
\Phi^{\mathrm{prod}}_{\Delta t,k}[\bar{\rho}_k]
e^{+iH_{\mathrm{Re}}\Delta t}.
\label{eq:app-avg-step-map}
\end{align}
Since the increment law $\xi_{k,\ell}\sim\mathcal{N}(0,2\gamma_\ell\Delta t)$ does not depend on $k$, the channel $\Phi^{\mathrm{prod}}_{\Delta t,k}$ is the same for all $k$.
Hence, we drop the step index and write $\Phi^{\mathrm{prod}}_{\Delta t,k}\eqqcolon \Phi^{\mathrm{prod}}_{\Delta t}$.

\subsection{Exact evaluation of \texorpdfstring{$\Phi^{\mathrm{prod}}_{\Delta t}$}{Phi\_Delta} using Gaussianity}

We introduce the adjoint superoperator based on the commutator as
\begin{align}
\mathrm{ad}_{H_{\mathrm{I},\ell}}[\rho]\coloneqq [H_{\mathrm{I},\ell},\rho].
\end{align}
Then, each unitary transformation can be written as
\begin{align}
e^{-i\xi_\ell H_{\mathrm{I},\ell}} \rho e^{+i\xi_\ell H_{\mathrm{I},\ell}}
=
e^{-i\xi_\ell \mathrm{ad}_{H_{\mathrm{I},\ell}}}[\rho],
\end{align}
and hence
\begin{align}
V_k^{\mathrm{prod}} \rho V_k^{\mathrm{prod}\dagger}
=
\left(\prod_{\ell} e^{-i\xi_\ell \mathrm{ad}_{H_{\mathrm{I},\ell}}}\right)[\rho],
\label{eq:app-V-ad}
\end{align}
where the product is ordered consistently with $V_k^{\mathrm{prod}}$.

Since the variables $\boldsymbol{\xi}_k$ are independent, the average Eq.~\eqref{eq:app-Phi-def} can be performed sequentially.
For each channel $\ell$, define the single-channel Gaussian-averaging map
\begin{align}
\Phi^{(\ell)}_{\Delta t}
\coloneqq
\int_{-\infty}^{\infty} d\xi
\frac{1}{\sqrt{4\pi\gamma_\ell\Delta t}}
\exp\left(-\frac{\xi^2}{4\gamma_\ell\Delta t}\right)
e^{-i\xi \mathrm{ad}_{H_{\mathrm{I},\ell}}}.
\label{eq:app-Phi-ell-int}
\end{align}
Using the identity concerning Fourier transform (for $a>0$),
\begin{align}
\int_{-\infty}^{\infty} d\xi
\frac{1}{\sqrt{4\pi a}}
\exp\left(-\frac{\xi^2}{4a}\right)e^{-i\xi x}
=
\exp(-a x^2),
\label{eq:app-gaussian-fourier}
\end{align}
and applying it with $x=\mathrm{ad}_{H_{\mathrm{I},\ell}}$ and $a=\gamma_\ell\Delta t$, we obtain
\begin{align}
\Phi^{(\ell)}_{\Delta t}
=
\exp\left(-\gamma_\ell\Delta t \mathrm{ad}_{H_{\mathrm{I},\ell}}^{ 2}\right),
\qquad
\mathrm{ad}_{H_{\mathrm{I},\ell}}^{ 2}[\rho]=[H_{\mathrm{I},\ell},[H_{\mathrm{I},\ell},\rho]].
\label{eq:app-Phi-ell-closed}
\end{align}

Consequently, fixing the same ordering $E=\{\ell_1,\ldots,\ell_{|E|}\}$ as in $V_k^{\mathrm{prod}}$, we have the exact representation
\begin{align}
\Phi^{\mathrm{prod}}_{\Delta t}
=
\Phi^{(\ell_{|E|})}_{\Delta t}\circ \cdots \circ \Phi^{(\ell_1)}_{\Delta t}.
\label{eq:app-Phi-multi-closed}
\end{align}
Substituting Eq.~\eqref{eq:app-Phi-multi-closed} into Eq.~\eqref{eq:app-avg-step-map} gives the closed form of the averaged one-step map.

\subsection{Continuous-time limit and GKSL form}

To obtain a differential equation, we expand the obtained maps to first order in $\Delta t$:
\begin{align}
e^{-iH_{\mathrm{Re}}\Delta t}\rho e^{+iH_{\mathrm{Re}}\Delta t}
&= \rho - i\Delta t [H_{\mathrm{Re}},\rho] + \mathcal{O}(\Delta t^2),\\
\Phi^{(\ell)}_{\Delta t}(\rho)
&=
\rho - \gamma_\ell\Delta t [H_{\mathrm{I},\ell},[H_{\mathrm{I},\ell},\rho]] + \mathcal{O}(\Delta t^2).
\end{align}
Since Eq.~\eqref{eq:app-Phi-multi-closed} is a composition over $\ell$, to first order we obtain
\begin{align}
\Phi^{\mathrm{prod}}_{\Delta t}[\rho]
=
\rho
-\Delta t\sum_{\ell \in E}\gamma_\ell [H_{\mathrm{I},\ell},[H_{\mathrm{I},\ell},\rho]]
+\mathcal{O}(\Delta t^2).
\end{align}
Using this in Eq.~\eqref{eq:app-avg-step-map} and taking $\Delta t\to 0$ yields
\begin{align}
\frac{d\bar{\rho}}{dt}
=
-i[H_{\mathrm{Re}},\bar{\rho}]
-\sum_{\ell \in E}\gamma_\ell [H_{\mathrm{I},\ell},[H_{\mathrm{I},\ell},\bar{\rho}]].
\label{eq:app-master-double-comm-multi}
\end{align}
Finally, using the identity
\begin{align}
-[H_{\mathrm{I},\ell},[H_{\mathrm{I},\ell},\rho]]
=
2H_{\mathrm{I},\ell}\rho H_{\mathrm{I},\ell}-\{H_{\mathrm{I},\ell}^2,\rho\},
\end{align}
we obtain the GKSL form
\begin{align}
\frac{d\bar{\rho}}{dt}
=
-i[H_{\mathrm{Re}},\bar{\rho}]
+\sum_{\ell \in E}\gamma_\ell\Bigl(2H_{\mathrm{I},\ell}\bar{\rho}H_{\mathrm{I},\ell}-\{H_{\mathrm{I},\ell}^2,\bar{\rho}\}\Bigr),
\end{align}
which coincides with Eq.~\eqref{eq:master-lindblad-hi} of the main text.

\section{Stochastic QEM implementation and sampling overhead}
\label{app:sqem_impl}
In this section, we review the stochastic QEM (sQEM) and the associated sampling overhead, following the formulation of Ref.~\cite{Sun2021}, and adapt it to our \hk{local-support} setting.

\subsection{Quasi-probability decomposition form}
We implement the cancellation map $\mathcal{E}_{\mathcal{C}}(\delta t)$ via a quasi-probability decomposition over $\{\mathcal{B}_{\ell j}\}$ for sufficiently small $\delta t$:
\begin{align}
\mathcal{E}_{\mathcal{C}}(\delta t)
=
\bigl(1+q_{0}\delta t\bigr)\mathcal{I}
+\sum_{\ell\in E}\sum_{j\ge 1} q_{\ell j} \delta t \mathcal{B}_{\ell j}
+\mathcal{O}(\delta t^{2}).
\label{eq:EC_near_identity_local}
\end{align}
Following Ref.~\cite{Sun2021}, we rewrite the first-order decomposition in a quasi-probability form
\begin{align}
\mathcal{E}_{\mathcal{C}}(\delta t)
=
c(\delta t)\left(
p_{0} \mathcal{B}_{0}
+
\sum_{\ell\in E}\sum_{j\ge 1}\alpha_{\ell j} \tilde p_{\ell j} \delta t \mathcal{B}_{\ell j}
\right)
+\mathcal{O}(\delta t^{2}),
\label{eq:EC_qpd_local}
\end{align}
where the parameters are chosen as
\begin{align}
c(\delta t)
&\coloneqq 1+\Bigl(q_{0}+\sum_{\ell\in E}\sum_{j\ge 1}|q_{\ell j}|\Bigr)\delta t, \notag\\
p_{0}
&\coloneqq \frac{1+q_{0}\delta t}{c(\delta t)}
= 1-\sum_{\ell\in E}\sum_{j\ge 1}\tilde p_{\ell j} \delta t, \notag\\
\tilde p_{\ell j}
&\coloneqq \frac{p_{\ell j}}{\delta t}
= \frac{|q_{\ell j}|}{c(\delta t)},
\qquad
\alpha_{\ell j}\coloneqq \sgn(q_{\ell j})\in\{\pm 1\}.
\label{eq:EC_qpd_params_local}
\end{align}
Thus $p_0=1-\mathcal{O}(\delta t)$ and $\tilde p_{\ell j}=\mathcal{O}(1)$ for $j\ge 1$.

For $T=n \delta t$, iterating Eq.~\eqref{eq:cancellation-condition} of the main text and substituting Eq.~\eqref{eq:EC_qpd_local} at each step yield a discrete-time expansion at the channel level.
Let $\vec{j}=(j_0,\ldots,j_{n-1})$ denote a sequence of \hk{basis operation} indices with $j_k=0$ for the identity channel and $j_k=(\ell,j)$ for a non-identity basis operation channel $\mathcal{B}_{\ell j}$ with $j\ge 1$.
Define $p_{(\ell,j)}\coloneqq \tilde p_{\ell j} \delta t$ and $\alpha_{(\ell,j)}\coloneqq \alpha_{\ell j}$, together with $p_{0}$ and $\alpha_{0}\coloneqq 1$.
We write $p_{\vec{j}}\coloneqq \prod_{k=0}^{n-1}p_{j_k}$ and $\alpha_{\vec{j}}\coloneqq \prod_{k=0}^{n-1}\alpha_{j_k}$.
Then
\begin{align}
\mathcal{E}_{\mathrm{NH}}(\delta t)^{n}
=
c(\delta t)^n\sum_{\vec{j}}
\alpha_{\vec{j}} p_{\vec{j}}
\mathcal{B}_{j_{n-1}}\circ \mathcal{E}_{\mathrm{GKSL}}(\delta t)\circ \cdots \circ
\mathcal{B}_{j_{0}}\circ \mathcal{E}_{\mathrm{GKSL}}(\delta t)
+\mathcal{O}(T\delta t),
\label{eq:qpd_channel_expansion_appB}
\end{align}
which implies
\begin{align}
\rho_{\mathrm{NH}}(T)
=
c(\delta t)^n\sum_{\vec{j}}
\alpha_{\vec{j}} p_{\vec{j}}
\rho_{\vec{j}}(T)
+\mathcal{O}(T\delta t),
\label{eq:qpd_state_expansion_appB}
\end{align}
where $\rho_{\mathrm{NH}}(T)\coloneqq \mathcal{E}_{\mathrm{NH}}(\delta t)^{n}[\rho(0)]$ and
\begin{align}
\rho_{\vec{j}}(T)\coloneqq
\Bigl(\mathcal{B}_{j_{n-1}}\circ \mathcal{E}_{\mathrm{GKSL}}(\delta t)\circ \cdots \circ
\mathcal{B}_{j_{0}}\circ \mathcal{E}_{\mathrm{GKSL}}(\delta t)\Bigr)[\rho(0)].
\label{eq:rho_j_sequence}
\end{align}
Consequently, for any observable $O$,
\begin{align}
\Tr\left[O \rho_{\mathrm{NH}}(T)\right]
=
c(\delta t)^n\sum_{\vec{j}} \alpha_{\vec{j}} p_{\vec{j}}
\Tr\left[O \rho_{\vec{j}}(T)\right]
+\mathcal{O}(T\delta t).
\label{eq:qpd_obs_expansion_appB}
\end{align}
The remainder $\mathcal{O}(T\delta t)$ is a discretization artifact of the above discrete-time derivation.
In the next subsection, we take the limit $\delta t\to0^+$ and realize the same scheme directly in continuous time as a jump process.

\subsection{Stochastic implementation in continuous time}
We now describe the Monte Carlo realization of the continuous-time limit $\delta t\to 0^+$~\cite{Sun2021}.
We treat only the non-identity terms $(\ell,j)$ with $j\ge 1$ as jump events.
We maintain a trajectory-parity variable $\alpha\in\{\pm 1\}$, initialized as $\alpha=1$, which accumulates the jump parities $\alpha_{\ell j}$ along the trajectory; in the discrete-time representation this corresponds to $\alpha_{\vec{j}}=\prod_{k}\alpha_{j_k}$ for the sampled sequence $\vec{j}$.

We define the total jump rate
\begin{align}
\Gamma_{\mathrm{tot}}\coloneqq \sum_{\ell\in E}\sum_{j\ge 1}\tilde p_{\ell j},
\label{eq:Gamma_tot_local}
\end{align}
where $\tilde p_{\ell j}$ are the corresponding event rates.

For time-independent rates, the no-jump survival probability $Q(t)$ and the jump-time density $P(t)$ in the interval $[t,t+dt]$ are
\begin{align}
Q(t)&=\exp(-\Gamma_{\mathrm{tot}}t),
\label{eq:survival_Q}\\
P(t) dt&=\Gamma_{\mathrm{tot}}e^{-\Gamma_{\mathrm{tot}}t} dt.
\label{eq:survival_density_local}
\end{align}
To sample the next jump time, generate $u\in[0,1]$ uniformly and solve $u=Q(t_{\mathrm{jp}})$, which gives the jump time
\begin{align}
t_{\mathrm{jp}}=-\ln u/\Gamma_{\mathrm{tot}}.
\label{eq:jump_time_local}
\end{align}
Conditioned on a jump, we choose the mark $(\ell,j)$ with probability
\begin{align}
\Pr[(\ell,j)] = \tilde p_{\ell j}/\Gamma_{\mathrm{tot}}.
\label{eq:mark_dist_local}
\end{align}
We then apply the corresponding basis operation channel $\mathcal{B}_{\ell j}$ and update the trajectory parity as
\begin{align}
\alpha \leftarrow \alpha_{\ell j} \alpha.
\label{eq:alpha_update_local}
\end{align}

Between sQEM events, we implement the unitary evolution generated by \hk{the} stochastic Hamiltonian for the sampled duration.
\hk{Averaging this evolution over the Gaussian-noise realizations reproduces the corresponding GKSL evolution.}
Repeating this procedure until the total evolution time reaches $T$ yields one \hk{trajectory}.
This construction is the $\delta t\to 0^+$ limit of the discrete-time quasi-probability scheme, where at each step a non-identity basis operation occurs with probability $\tilde p_{\ell j}\delta t$.
In this limit the jump process becomes a Poisson process with survival probability Eq.~\eqref{eq:survival_Q}, and averaging over trajectories reproduces the same error-mitigated channel as the continuous-time formulation in Ref.~\cite{Sun2021}.

\section{\hk{Statistical error analysis}}
\label{app:statistical_error}
\hk{We quantify the statistical error of the trajectory estimator, including both Gaussian and sQEM randomness.
Let $N_{\mathrm{traj}}$ be the number of i.i.d.\ trajectories used to estimate the normalized expectation value at a fixed time $T$.}

The quasi-probability decomposition introduces an overall normalization factor that accumulates with time.
Taking the continuous-time limit $\delta t\to 0^+$ with $n=T/\delta t$, we define the overhead factor
\begin{align}
\Lambda(T)\coloneqq \lim_{\delta t\to 0^+} c(\delta t)^{T/\delta t}.
\label{eq:Lambda_def_app}
\end{align}
Using Eq.~\eqref{eq:EC_qpd_params_local}, we obtain
\begin{align}
\Lambda(T)=\exp(\kappa T),
\qquad
\kappa\coloneqq q_{0}+\sum_{\ell\in E}\sum_{j\ge 1}|q_{\ell j}|.
\label{eq:Lambda_exp_kappa_local}
\end{align}

If all basis operation channels are trace preserving, then for each trajectory $s=1,\ldots,\hk{N_{\mathrm{traj}}}$ we measure the observable $O$ at time $T$ and multiply the outcome by the \hk{trajectory} parity $\alpha\in\{\pm1\}$.
Since $|O|\le \|O\|$, the magnitude of each signed contribution is bounded by $\Lambda(T)\|O\|$, and the RMSE of the sample average satisfies
\begin{align}
\mathrm{RMSE}(T)\ \le\ \frac{\Lambda(T)\|O\|}{\sqrt{\hk{N_{\mathrm{traj}}}}}.
\label{eq:Ns_rmse_bound}
\end{align}

If some basis operation channels are trace decreasing, as is the case for the projective-measurement basis operations $[\pi]$ in Table~\ref{tab:bases} in the End Matter, then a sampled trajectory can be unnormalized.
In this case, the \hk{normalized expectation value obtained from the Trotterized implementation with infinite sampling} is
\begin{align}
\langle O\rangle_{\hk{\mathrm{impl}}}(T)&=
\frac{\Tr\left[O \rho_{\mathrm{NH}}^{\hk{(\mathrm{impl})}}(T)\right]}{\Tr\left[\rho_{\mathrm{NH}}^{\hk{(\mathrm{impl})}}(T)\right]}=
\frac{\Lambda(T)\sum_{\vec{j}} \alpha_{\vec{j}} p_{\vec{j}}
\Tr\left[O\hk{\rho_{\vec{j}}(T)}\right]}{\Lambda(T)\sum_{\vec{j}} \alpha_{\vec{j}} p_{\vec{j}}
\Tr\left[\hk{\rho_{\vec{j}}(T)}\right]},
\label{eq:ONH_normalized_def_appB}
\end{align}
\hk{Here, $\rho_{\vec{j}}(T)$ is defined in Eq.~\eqref{eq:rho_j_sequence}.}
Accordingly, both the numerator and the denominator are estimated from the same trajectory sample.

Along a given trajectory $s\in \hk{\{1,\dots,N_{\mathrm{traj}}\}}$, suppose that measurement basis operations occur at random event indices $w=1,\ldots,n_{\mathrm{meas}}(s)$.
For each such measurement basis operation event $w$, we define an indicator $\eta^{(w)}_s \in \{0,1\}$, where $\eta^{(w)}_s=1$ if the required measurement outcome associated with that basis operation is obtained, and $\eta^{(w)}_s=0$ otherwise.
Importantly, even when $\eta^{(w)}_s=0$ occurs, the trial is still counted in the total number of Monte Carlo samples, i.e., we do not remove the shot from the dataset; instead, we record a zero contribution as described below.

The overall acceptance indicator for \hk{each trajectory} is then $D_s(T)\coloneqq \prod_{w=1}^{n_{\mathrm{meas}}(s)} \eta^{(w)}_s \in \{0,1\}$, so that $D_s(T)=1$ if and only if all required measurement-basis-operation outcomes along the trajectory are obtained, and $D_s(T)=0$ otherwise.

We denote by $O_s(T)\in[-\|O\|,\|O\|]$ the outcome of the final measurement of $O$ at time $T$ when the trajectory is accepted.
If the trajectory is rejected, we set $O_s(T)\coloneqq 0$ when $D_s(T)=0$.
We \hk{denote the} per-trajectory parity \hk{by} $\alpha_s\in\{\pm1\}$.
The \hk{ratio} estimator is
\begin{align}
    \langle O\rangle_{\hk{N_{\mathrm{traj}}}}(T)=\frac{\sum_{s=1}^{\hk{N_{\mathrm{traj}}}} \alpha_s O_s(T)}{\sum_{s=1}^{\hk{N_{\mathrm{traj}}}} \alpha_s D_s(T)}.
    \label{eq:estimator_appB}
\end{align}
\hk{If the denominator vanishes, we set $\langle O\rangle_{N_{\mathrm{traj}}}(T)=0$.}

Next we consider the RMSE of the normalized estimator in the presence of trace-decreasing basis operation channels~\cite{Vikstl2024,vanKempen2000}.

\begin{theorem}[RMSE of the normalized \hk{trajectory} ratio estimator in the presence of trace-decreasing basis operation channels]
\label{thm:rmse_ratio_sqem_full}
Fix a final time $T$ \hk{and a time step $\Delta t$}.
For each trajectory $s=1,\dots,\hk{N_{\mathrm{traj}}}$, define \hk{$X_s \coloneqq \alpha_s O_s(T)$ and $Y_s \coloneqq \alpha_s D_s(T)$}, where $\alpha_s\in\{\pm1\}$, $D_s(T)\in\{0,1\}$, and $O_s(T)\in[-\|O\|,\|O\|]$ with the convention $O_s(T)=0$ whenever $D_s(T)=0$.
Assume $\{(X_s,Y_s)\}_{s=1}^{\hk{N_{\mathrm{traj}}}}$ are i.i.d.\ \hk{under the trajectory distribution, including both Gaussian and sQEM randomness, and that $\mathbb{E}_{\mathrm{all}}[Y_s]\neq0$}.
\hk{The quasi-probability reconstruction in Eq.~\eqref{eq:trajectory_state_reconstruction} gives}
\begin{align}
\hk{\Lambda(T)\mathbb{E}_{\mathrm{all}}[X_s]}
&\hk{=\Tr\left[O\rho_{\mathrm{NH}}^{(\mathrm{impl})}(T)\right],}\\
\hk{\Lambda(T)\mathbb{E}_{\mathrm{all}}[Y_s]}
&\hk{=\Tr\left[\rho_{\mathrm{NH}}^{(\mathrm{impl})}(T)\right],}
\end{align}
\hk{and hence
$\langle O\rangle_{\mathrm{impl}}(T)=\mathbb{E}_{\mathrm{all}}[X_s]/\mathbb{E}_{\mathrm{all}}[Y_s]$.}
Consider the normalized estimator
\begin{align}
\hk{\langle O\rangle_{N_{\mathrm{traj}}}(T)}
\coloneqq
\frac{\sum_{s=1}^{\hk{N_{\mathrm{traj}}}}X_s}
{\sum_{s=1}^{\hk{N_{\mathrm{traj}}}}Y_s}
=\frac{\bar X}{\bar Y},
\qquad
\hk{\bar X\coloneqq \frac{1}{N_{\mathrm{traj}}}\sum_{s=1}^{N_{\mathrm{traj}}}X_s,\quad
\bar Y\coloneqq \frac{1}{N_{\mathrm{traj}}}\sum_{s=1}^{N_{\mathrm{traj}}}Y_s.}
\end{align}
Then, as $\hk{N_{\mathrm{traj}}}\to\infty$,
\begin{align}
\sqrt{\hk{N_{\mathrm{traj}}}}
\hk{\bigl(\langle O\rangle_{N_{\mathrm{traj}}}(T)-\langle O\rangle_{\mathrm{impl}}(T)\bigr)}
\ \xrightarrow[]{d}\
\mathcal{N}\hk{\left(
0,\frac{\mathrm{Var}_{\mathrm{all}}\!\left(
X_s-\langle O\rangle_{\mathrm{impl}}(T)Y_s
\right)}
{\mathbb{E}_{\mathrm{all}}[Y_s]^2}
\right)\hk{.}}
\label{eq:thm_asymp_normal}
\end{align}
\hk{In the large-sample regime where $\bar Y$ remains bounded away from zero, the statistical error is asymptotically approximated as}
\begin{align}
\hk{\epsilon_{\mathrm{stat}}(T)}
\hk{\coloneqq
\sqrt{\mathbb{E}_{\mathrm{all}}\!\left[
\bigl(\langle O\rangle_{N_{\mathrm{traj}}}(T)-\langle O\rangle_{\mathrm{impl}}(T)\bigr)^2
\right]}}
\hk{\approx}
\hk{
\frac{
\sqrt{\mathrm{Var}_{\mathrm{all}}\!\left(
X_s-\langle O\rangle_{\mathrm{impl}}(T)Y_s
\right)}
}{
\sqrt{N_{\mathrm{traj}}}\,|\mathbb{E}_{\mathrm{all}}[Y_s]|
}.}
\label{eq:thm_rmse_asymp}
\end{align}
\hk{Similarly, the second-order expansion gives the following asymptotic approximation for the finite-sample bias:}
\begin{align}
\hk{\mathrm{Bias}_{N_{\mathrm{traj}}}(T)}
\hk{\coloneqq
\mathbb{E}_{\mathrm{all}}\!\left[\langle O\rangle_{N_{\mathrm{traj}}}(T)\right]
-\langle O\rangle_{\mathrm{impl}}(T)}
\hk{\approx}
\hk{\frac{1}{N_{\mathrm{traj}}}
\frac{
\langle O\rangle_{\mathrm{impl}}(T)\mathrm{Var}_{\mathrm{all}}(Y_s)
-\mathrm{Cov}_{\mathrm{all}}(X_s,Y_s)
}{
\mathbb{E}_{\mathrm{all}}[Y_s]^2
}.}
\label{eq:thm_bias_second_order}
\end{align}
and using only boundedness and the convention $O_s(T)=0$ when $D_s(T)=0$,
\begin{align}
\hk{\epsilon_{\mathrm{stat}}(T)}
\ \lesssim\
\frac{\hk{\|O\|+|\langle O\rangle_{\mathrm{impl}}(T)|}}{\sqrt{\hk{N_{\mathrm{traj}}}}}
\sqrt{\frac{\hk{\mathbb{E}_{\mathrm{all}}}[D_s(T)]}{\hk{\mathbb{E}_{\mathrm{all}}}[\alpha_s D_s(T)]^2}}.
\label{eq:thm_rmse_bound}
\end{align}
\end{theorem}

\begin{proof}
From definitions $|X_s|=|\alpha_s O_s(T)|\le \|O\|$ and $|Y_s|=|\alpha_s D_s(T)|=D_s(T)\le 1$, so $\hk{\mathbb{E}_{\mathrm{all}}}[X_s^2]\le\|O\|^2$ and $\hk{\mathbb{E}_{\mathrm{all}}}[Y_s^2]\le1$.
\hk{We consider the large-sample regime in which $\bar Y$ is nonzero and remains bounded away from zero.}
Introducing the mean fluctuations
\hk{$\Delta X\coloneqq\bar X-\mathbb{E}_{\mathrm{all}}[X_s]$ and
$\Delta Y\coloneqq\bar Y-\mathbb{E}_{\mathrm{all}}[Y_s]$}, the multivariate central limit theorem yields
\begin{align}
\sqrt{\hk{N_{\mathrm{traj}}}}
\begin{pmatrix}
    \Delta X \\
    \Delta Y
\end{pmatrix}
\xrightarrow[]{d}\ \mathcal{N}(0,\Sigma),
\qquad
\Sigma=\mathrm{Cov}\left(\begin{pmatrix}X_s\\ Y_s\end{pmatrix}\right)=
\begin{pmatrix}
\mathrm{Var}(X_s) & \mathrm{Cov}(X_s,Y_s)\\
\mathrm{Cov}(X_s,Y_s) & \mathrm{Var}(Y_s)
\end{pmatrix}.
\label{eq:clt_vector}
\end{align}

Define $g(x,y)\coloneqq x/y$, so that
\hk{$g(\bar X,\bar Y)=\langle O\rangle_{N_{\mathrm{traj}}}(T)$ and
$g(\mathbb{E}_{\mathrm{all}}[X_s],\mathbb{E}_{\mathrm{all}}[Y_s])
=\langle O\rangle_{\mathrm{impl}}(T)$}.
A first-order Taylor expansion around
\hk{$(\mathbb{E}_{\mathrm{all}}[X_s],\mathbb{E}_{\mathrm{all}}[Y_s])$} gives
\begin{align}
g(\bm{m}+\bm{h})\ \approx\ g(\bm{m})+\nabla g(\bm{m})^{\mathsf T}\bm{h},
\qquad
\bm{m} \coloneqq
\hk{\begin{pmatrix}
\mathbb{E}_{\mathrm{all}}[X_s]\\
\mathbb{E}_{\mathrm{all}}[Y_s]
\end{pmatrix}},
\quad
\bm{h} \coloneqq
\begin{pmatrix}\Delta X\\ \Delta Y\end{pmatrix}.
\end{align}
For $g(x,y)=x/y$, the gradient is
\begin{align}
\nabla g\hk{\left(
\mathbb{E}_{\mathrm{all}}[X_s],
\mathbb{E}_{\mathrm{all}}[Y_s]
\right)}
=
\hk{\begin{pmatrix}
1/\mathbb{E}_{\mathrm{all}}[Y_s]\\
-\langle O\rangle_{\mathrm{impl}}(T)/\mathbb{E}_{\mathrm{all}}[Y_s]
\end{pmatrix}}
\eqqcolon \bm{a}.
\label{eq:grad_g}
\end{align}
Therefore,
\begin{align}
\hk{\langle O\rangle_{N_{\mathrm{traj}}}(T)
-\langle O\rangle_{\mathrm{impl}}(T)}
=
\hk{g(\bar X,\bar Y)
-g\!\left(\mathbb{E}_{\mathrm{all}}[X_s],
\mathbb{E}_{\mathrm{all}}[Y_s]\right)}
=
g(\bm{m}+\bm{h})-g(\bm{m})
\approx
\bm{a}^{\mathsf T}
\begin{pmatrix}\Delta X\\ \Delta Y\end{pmatrix}.
\label{eq:taylor_first}
\end{align}
Multiplying Eq.~\eqref{eq:taylor_first} by $\sqrt{\hk{N_{\mathrm{traj}}}}$ gives
\begin{align}
\sqrt{\hk{N_{\mathrm{traj}}}}
\hk{\bigl(\langle O\rangle_{N_{\mathrm{traj}}}(T)-\langle O\rangle_{\mathrm{impl}}(T)\bigr)}
\approx \bm{a}^{\mathsf T}\sqrt{\hk{N_{\mathrm{traj}}}} \bm{h}.
\end{align}
By the multivariate central limit theorem Eq.~\eqref{eq:clt_vector}, we have
\hk{$\sqrt{N_{\mathrm{traj}}}\bm{h}
=\sqrt{N_{\mathrm{traj}}}\bigl((\bar X,\bar Y)^{\mathsf T}
-(\mathbb{E}_{\mathrm{all}}[X_s],\mathbb{E}_{\mathrm{all}}[Y_s])^{\mathsf T}\bigr)
\xrightarrow[]{d}\bm{W}$} with $\bm{W}\sim \mathcal{N}(0,\Sigma)$.
Since $\bm{a}$ is deterministic, the continuous mapping theorem implies $\bm{a}^{\mathsf T}\sqrt{\hk{N_{\mathrm{traj}}}} \bm{h}\xrightarrow[]{d} \bm{a}^{\mathsf T}\bm{W}$.
Moreover, a linear functional of a multivariate Gaussian is Gaussian, so $\bm{a}^{\mathsf T}\bm{W}$ is a one-dimensional normal variable with
\begin{align}
\hk{\mathbb{E}_{\mathrm{all}}}[\bm{a}^{\mathsf T}\bm{W}]&=\bm{a}^{\mathsf T}\hk{\mathbb{E}_{\mathrm{all}}}[\bm{W}]=0,\\
\mathrm{Var}(\bm{a}^{\mathsf T}\bm{W})&=\hk{\mathbb{E}_{\mathrm{all}}}\left[(\bm{a}^{\mathsf T}\bm{W})^2\right]
=\hk{\mathbb{E}_{\mathrm{all}}}\left[\bm{a}^{\mathsf T}\bm{W}\bm{W}^{\mathsf T}\bm{a}\right]
=\bm{a}^{\mathsf T}\hk{\mathbb{E}_{\mathrm{all}}}[\bm{W}\bm{W}^{\mathsf T}]\bm{a}
=\bm{a}^{\mathsf T}\Sigma \bm{a},
\label{eq:var_linear_gaussian}
\end{align}
which yields
\hk{$\sqrt{N_{\mathrm{traj}}}
\bigl(\langle O\rangle_{N_{\mathrm{traj}}}(T)-\langle O\rangle_{\mathrm{impl}}(T)\bigr)
\xrightarrow[]{d}\mathcal{N}(0,\bm{a}^{\mathsf T}\Sigma\bm{a})$}.
A direct expansion gives
\begin{align}
\bm{a}^\top \Sigma \bm{a}
&=
\hk{\begin{pmatrix}
\dfrac{1}{\mathbb{E}_{\mathrm{all}}[Y_s]} &
-\dfrac{\langle O\rangle_{\mathrm{impl}}(T)}
{\mathbb{E}_{\mathrm{all}}[Y_s]}
\end{pmatrix}}
\begin{pmatrix}
\mathrm{Var}(X_s) & \mathrm{Cov}(X_s,Y_s)\\
\mathrm{Cov}(X_s,Y_s) & \mathrm{Var}(Y_s)
\end{pmatrix}
\hk{\begin{pmatrix}
\dfrac{1}{\mathbb{E}_{\mathrm{all}}[Y_s]}\\
-\dfrac{\langle O\rangle_{\mathrm{impl}}(T)}
{\mathbb{E}_{\mathrm{all}}[Y_s]}
\end{pmatrix}}\notag\\
&=
\hk{\frac{
\mathrm{Var}(X_s)
+\langle O\rangle_{\mathrm{impl}}(T)^2\mathrm{Var}(Y_s)
-2\langle O\rangle_{\mathrm{impl}}(T)\mathrm{Cov}(X_s,Y_s)
}{
\mathbb{E}_{\mathrm{all}}[Y_s]^2
}}
=
\hk{\frac{
\mathrm{Var}\!\left(X_s-\langle O\rangle_{\mathrm{impl}}(T)Y_s\right)
}{
\mathbb{E}_{\mathrm{all}}[Y_s]^2
},}
\label{eq:asymp_var_calc}
\end{align}
which proves Eq.~\eqref{eq:thm_asymp_normal}.

We next derive the second-order bias expansion in Eq.~\eqref{eq:thm_bias_second_order}.
Using \hk{$\langle O\rangle_{N_{\mathrm{traj}}}(T)
=(\mathbb{E}_{\mathrm{all}}[X_s]+\Delta X)/
(\mathbb{E}_{\mathrm{all}}[Y_s]+\Delta Y)$} and
\hk{$\langle O\rangle_{\mathrm{impl}}(T)
=\mathbb{E}_{\mathrm{all}}[X_s]/\mathbb{E}_{\mathrm{all}}[Y_s]$}, we write
\begin{align}
\hk{\langle O\rangle_{N_{\mathrm{traj}}}(T)}
=
\hk{\frac{\mathbb{E}_{\mathrm{all}}[X_s]+\Delta X}
{\mathbb{E}_{\mathrm{all}}[Y_s]+\Delta Y}}
=
\hk{\frac{\mathbb{E}_{\mathrm{all}}[X_s]+\Delta X}
{\mathbb{E}_{\mathrm{all}}[Y_s]}}
\hk{\frac{1}{
1+\Delta Y/\mathbb{E}_{\mathrm{all}}[Y_s]}.}
\end{align}
We expand the denominator and obtain
\begin{align}
\hk{\langle O\rangle_{N_{\mathrm{traj}}}(T)}
&=
\hk{\frac{\mathbb{E}_{\mathrm{all}}[X_s]+\Delta X}
{\mathbb{E}_{\mathrm{all}}[Y_s]}
\left(
1-\frac{\Delta Y}{\mathbb{E}_{\mathrm{all}}[Y_s]}
+\frac{\Delta Y^2}{\mathbb{E}_{\mathrm{all}}[Y_s]^2}
\right)}
+O(\hk{N_{\mathrm{traj}}}^{-3/2})
\notag\\
&=
\hk{\frac{\mathbb{E}_{\mathrm{all}}[X_s]}
{\mathbb{E}_{\mathrm{all}}[Y_s]}
+\frac{\Delta X}{\mathbb{E}_{\mathrm{all}}[Y_s]}
-\frac{\mathbb{E}_{\mathrm{all}}[X_s]}
{\mathbb{E}_{\mathrm{all}}[Y_s]^2}\Delta Y
-\frac{\Delta X}{\mathbb{E}_{\mathrm{all}}[Y_s]^2}\Delta Y}
\notag\\
&\hk{\quad
+\frac{\mathbb{E}_{\mathrm{all}}[X_s]}
{\mathbb{E}_{\mathrm{all}}[Y_s]^3}\Delta Y^2
+\frac{\Delta X}{\mathbb{E}_{\mathrm{all}}[Y_s]^3}\Delta Y^2}
+O(\hk{N_{\mathrm{traj}}}^{-3/2})
\notag\\
&=
\hk{\langle O\rangle_{\mathrm{impl}}(T)
+\left(
\frac{\Delta X}{\mathbb{E}_{\mathrm{all}}[Y_s]}
-\frac{\langle O\rangle_{\mathrm{impl}}(T)}
{\mathbb{E}_{\mathrm{all}}[Y_s]}\Delta Y
\right)}
\notag\\
&\hk{\quad
+\left(
\frac{\langle O\rangle_{\mathrm{impl}}(T)}
{\mathbb{E}_{\mathrm{all}}[Y_s]^2}\Delta Y^2
-\frac{1}{\mathbb{E}_{\mathrm{all}}[Y_s]^2}\Delta X\Delta Y
\right)}
+O(\hk{N_{\mathrm{traj}}}^{-3/2}),
\label{eq:second_order_expand}
\end{align}
\hk{where we used
$\mathbb{E}_{\mathrm{all}}[X_s]/\mathbb{E}_{\mathrm{all}}[Y_s]
=\langle O\rangle_{\mathrm{impl}}(T)$ and
$\mathbb{E}_{\mathrm{all}}[X_s]/\mathbb{E}_{\mathrm{all}}[Y_s]^3
=\langle O\rangle_{\mathrm{impl}}(T)/\mathbb{E}_{\mathrm{all}}[Y_s]^2$.}

Taking the ensemble average of Eq.~\eqref{eq:second_order_expand} and using
\hk{$\mathbb{E}_{\mathrm{all}}[\Delta X]
=\mathbb{E}_{\mathrm{all}}[\bar X-\mathbb{E}_{\mathrm{all}}[X_s]]=0$}
and $\hk{\mathbb{E}_{\mathrm{all}}[\Delta Y]=0}$, we find
\begin{align}
\hk{\mathbb{E}_{\mathrm{all}}\!\left[
\langle O\rangle_{N_{\mathrm{traj}}}(T)
\right]-\langle O\rangle_{\mathrm{impl}}(T)}
&=
\hk{\frac{
\langle O\rangle_{\mathrm{impl}}(T)\mathbb{E}_{\mathrm{all}}[\Delta Y^2]
-\mathbb{E}_{\mathrm{all}}[\Delta X\Delta Y]
}{
\mathbb{E}_{\mathrm{all}}[Y_s]^2
}}
+o\left(\frac{1}{\hk{N_{\mathrm{traj}}}}\right).
\label{eq:bias_reduce_to_moments}
\end{align}
Since \hk{$\Delta X=\bar X-\mathbb{E}_{\mathrm{all}}[X_s]$ and
$\Delta Y=\bar Y-\mathbb{E}_{\mathrm{all}}[Y_s]$} are sample-mean fluctuations,
\begin{align}
\hk{\mathbb{E}_{\mathrm{all}}}[\Delta Y^2]=\mathrm{Var}(\bar Y)=\frac{\mathrm{Var}(Y_s)}{\hk{N_{\mathrm{traj}}}},
\qquad
\hk{\mathbb{E}_{\mathrm{all}}}[\Delta X\Delta Y]=\mathrm{Cov}(\bar X,\bar Y)=\frac{\mathrm{Cov}(X_s,Y_s)}{\hk{N_{\mathrm{traj}}}}.
\label{eq:mean_var_cov_scaling}
\end{align}
Substituting Eq.~\eqref{eq:mean_var_cov_scaling} into Eq.~\eqref{eq:bias_reduce_to_moments} yields \eqref{eq:thm_bias_second_order}.

We now derive the RMSE scaling Eq.~\eqref{eq:thm_rmse_asymp}.
By definition,
\begin{align}
\hk{\epsilon_{\mathrm{stat}}(T)}^2
=
\hk{\mathbb{E}_{\mathrm{all}}\!\left[
\bigl(\langle O\rangle_{N_{\mathrm{traj}}}(T)
-\langle O\rangle_{\mathrm{impl}}(T)\bigr)^2
\right]}
=
\hk{\mathrm{Var}_{\mathrm{all}}\!\left(
\langle O\rangle_{N_{\mathrm{traj}}}(T)
\right)+\mathrm{Bias}_{N_{\mathrm{traj}}}(T)^2.}
\label{eq:rmse_def_expand}
\end{align}
From Eq.~\eqref{eq:thm_asymp_normal}, one has
\hk{$\mathrm{Var}_{\mathrm{all}}\!\left(\langle O\rangle_{N_{\mathrm{traj}}}(T)\right)
=\mathrm{Var}_{\mathrm{all}}\!\left(X_s-\langle O\rangle_{\mathrm{impl}}(T)Y_s\right)/
\left(N_{\mathrm{traj}}\mathbb{E}_{\mathrm{all}}[Y_s]^2\right)
+o(N_{\mathrm{traj}}^{-1})$}, while Eq.~\eqref{eq:thm_bias_second_order} implies
\hk{$\mathrm{Bias}_{N_{\mathrm{traj}}}(T)=O(N_{\mathrm{traj}}^{-1})$} and hence its square is $O(\hk{N_{\mathrm{traj}}}^{-2})$.
Therefore Eq.~\eqref{eq:rmse_def_expand} gives
\begin{align}
\hk{\epsilon_{\mathrm{stat}}(T)}^2
=
\hk{\frac{
\mathrm{Var}_{\mathrm{all}}\!\left(
X_s-\langle O\rangle_{\mathrm{impl}}(T)Y_s
\right)
}{
N_{\mathrm{traj}}\mathbb{E}_{\mathrm{all}}[Y_s]^2
}}
+o\left(\frac{1}{\hk{N_{\mathrm{traj}}}}\right),
\end{align}
and taking the square root yields Eq.~\eqref{eq:thm_rmse_asymp}.

Finally, we prove the coarse bound Eq.~\eqref{eq:thm_rmse_bound}.
Using
\hk{$X_s-\langle O\rangle_{\mathrm{impl}}(T)Y_s
=\alpha_s\bigl(O_s(T)-\langle O\rangle_{\mathrm{impl}}(T)D_s(T)\bigr)$}
and $\alpha_s^2=1$, we have
\begin{align}
\hk{\bigl(X_s-\langle O\rangle_{\mathrm{impl}}(T)Y_s\bigr)^2
=\bigl(O_s(T)-\langle O\rangle_{\mathrm{impl}}(T)D_s(T)\bigr)^2.}
\end{align}
If $D_s(T)=0$, then $O_s(T)=0$ by convention and the right-hand side is zero.
If $D_s(T)=1$, then
\hk{$|O_s(T)-\langle O\rangle_{\mathrm{impl}}(T)|
\le\|O\|+|\langle O\rangle_{\mathrm{impl}}(T)|$}.
Combining both cases gives
\hk{$\bigl(X_s-\langle O\rangle_{\mathrm{impl}}(T)Y_s\bigr)^2
\le\bigl(\|O\|+|\langle O\rangle_{\mathrm{impl}}(T)|\bigr)^2D_s(T)$}
and therefore
\begin{align}
\mathrm{Var}_{\hk{\mathrm{all}}}\!\left(
X_s-\hk{\langle O\rangle_{\mathrm{impl}}(T)}Y_s
\right)
\le
\hk{\mathbb{E}_{\mathrm{all}}\!\left[
\bigl(X_s-\langle O\rangle_{\mathrm{impl}}(T)Y_s\bigr)^2
\right]}
\le
\hk{\bigl(\|O\|+|\langle O\rangle_{\mathrm{impl}}(T)|\bigr)^2
\mathbb{E}_{\mathrm{all}}[D_s(T)].}
\end{align}
Using
\hk{$\mathbb{E}_{\mathrm{all}}[Y_s]
=\mathbb{E}_{\mathrm{all}}[\alpha_sD_s(T)]$}
then yields Eq.~\eqref{eq:thm_rmse_bound}.
\end{proof}

As described in End Matter, we set
\begin{align}
p_{\mathrm{succ}}(T)= \hk{\mathbb{E}_{\mathrm{all}}}[D_s(T)],
\qquad
p_{\mathrm{fail}}(T)=1-p_{\mathrm{succ}}(T).
\end{align}
Then Eq.~\eqref{eq:thm_rmse_bound} yields
\begin{align}
\hk{\epsilon_{\mathrm{stat}}(T)}
\ \lesssim\
\frac{\hk{\|O\|+|\langle O\rangle_{\mathrm{impl}}(T)|}}{\sqrt{\hk{N_{\mathrm{traj}}}}}
\sqrt{\frac{p_{\mathrm{succ}}(T)}{\hk{\mathbb{E}_{\mathrm{all}}}[\alpha_s D_s(T)]^2}}.
\label{eq:rmse_bound_psucc}
\end{align}
From \hk{Eqs.~\eqref{eq:ONH_normalized_def_appB} and~\eqref{eq:estimator_appB}}, we can rewrite $\Tr[\rho_{\mathrm{NH}}^{\hk{(\mathrm{impl})}}(T)] = \Lambda(T) \hk{\mathbb{E}_{\mathrm{all}}}[\alpha_s D_s(T)]$.
Then Eq.~\eqref{eq:rmse_bound_psucc} gives
\begin{align}
\hk{\epsilon_{\mathrm{stat}}(T)}
\ \lesssim\
\frac{\hk{\|O\|+|\langle O\rangle_{\mathrm{impl}}(T)|}}{\sqrt{\hk{N_{\mathrm{traj}}}}}
\frac{\Lambda(T)\sqrt{p_{\mathrm{succ}}(T)}}{|\Tr[\rho_{\mathrm{NH}}^{\hk{(\mathrm{impl})}}(T)]|}.
\label{eq:rmse_bound_psucc_Lambda_supp}
\end{align}

\section{Auxiliary lemmas}\label{app:aux_lemmas}
We collect several standard formula\hk{s} used in this paper.
See, e.g., Refs.~\cite{bers1964partial, dollard1979product} for the Duhamel formula, Ref.~\cite{horn2012matrix,dollard1979product} for matrix-norm inequalities and Ref.~\cite{papoulis2002probability} for Gaussian moments.

\begin{lemma}[Duhamel formula (variation of constants)]
\label{lem:duhamel}
Let $H$ be time-independent and let $D(\tau)$ satisfy
\begin{align}
\frac{d}{d\tau}D(\tau) = -iH D(\tau) + R(\tau),
\qquad
D(0)=D_0 .
\label{eq:duhamel_ode}
\end{align}
Then
\begin{align}
D(\tau)
= e^{-i\tau H}D_0
+ \int_{0}^{\tau} ds e^{-i(\tau-s)H} R(s).
\label{eq:duhamel_sol}
\end{align}
\end{lemma}

\begin{proof}
Multiply Eq.~\eqref{eq:duhamel_ode} from the left by $e^{+i\tau H}$ and use the product rule to obtain $\frac{d}{d\tau}\bigl(e^{+i\tau H}D(\tau)\bigr)=e^{+i\tau H}R(\tau)$.
Integrating from $0$ to $\tau$ and multiplying by $e^{-i\tau H}$ gives Eq.~\eqref{eq:duhamel_sol}.
\end{proof}

\begin{lemma}[Gaussian absolute moments]
\label{lem:gaussian_moments}
Let $X\sim\mathcal{N}(0,\sigma^2)$.
Then
\begin{align}
\mathbb{E}[|X|]=\sigma\sqrt{\frac{2}{\pi}},
\qquad
\mathbb{E}[X^2]=\sigma^2,
\qquad
\mathbb{E}[|X|^3]=\frac{2\sqrt{2}}{\sqrt{\pi}}\sigma^3.
\label{eq:gaussian_moments}
\end{align}
\end{lemma}

\begin{proof}
The identity $\mathbb{E}[X^2]=\sigma^2$ follows directly from $\mathrm{Var}(X)=\sigma^2$ and $\mathbb{E}[X]=0$.
The absolute moments can be obtained by evaluating the corresponding Gaussian integrals.
\end{proof}

\begin{lemma}[Unitary operator bound in trace norm]\label{lem:unitary_trace_lipschitz}
Let $\rho\ge0$ with $\Tr\rho>0$.
Let $U$ and $V$ be unitary operators and define $\tau:= U\rho U^\dagger$ and $\sigma:= V\rho V^\dagger$.
Then
\begin{align}
\|\tau-\sigma\|_1 \le 2\,(\Tr\rho)\,\|U-V\|.
\label{eq:unitary_trace_lipschitz}
\end{align}
In particular, if $\Tr\rho=1$ then $\|\tau-\sigma\|_1\le 2\|U-V\|$.
\end{lemma}

\begin{proof}
We expand
\begin{align}
\tau-\sigma
&=U\rho U^\dagger - V\rho V^\dagger
=(U-V)\rho U^\dagger + V\rho(U^\dagger - V^\dagger).
\end{align}
Using the triangle inequality and $\|AXB\|_1\le \|A\|\,\|X\|_1\,\|B\|$, we obtain
\begin{align}
\|\tau-\sigma\|_1
&\le \|(U-V)\rho U^\dagger\|_1 + \|V\rho(U^\dagger - V^\dagger)\|_1 \notag\\
&\le \|U-V\|\,\|\rho\|_1\,\|U^\dagger\| + \|V\|\,\|\rho\|_1\,\|U^\dagger - V^\dagger\|.
\end{align}
Since $\|\rho\|_1=\Tr\rho$, $\|U\|=\|V\|=1$, and $\|U^\dagger - V^\dagger\|=\|U-V\|$, we get $\|\tau-\sigma\|_1\le 2(\Tr\rho)\|U-V\|$.
\end{proof}

\begin{lemma}[Bound evaluation of normalized expectation values]\label{lem:normalized_stability}
Let \hk{$\tau$ be a Hermitian trace-class operator and $\sigma \ge 0$ be positive semidefinite, with $t:=\Tr\tau\neq0$ and $s:=\Tr\sigma>0$}.
Let $O=O^\dagger$ be a bounded observable and define
\begin{align}
\langle O\rangle_\tau \coloneqq \frac{\Tr[O\tau]}{t},\qquad
\langle O\rangle_\sigma \coloneqq \frac{\Tr[O\sigma]}{s}.
\end{align}
Then
\begin{align}
\bigl|\langle O\rangle_\tau-\langle O\rangle_\sigma\bigr|
\le
\frac{2\|O\|}{\min\{\hk{|t|},s\}}\;\|\tau-\sigma\|_1.
\label{eq:ratio_stability_trace}
\end{align}
\end{lemma}

\begin{proof}
We start from
\begin{align}
\left|\langle O\rangle_\tau-\langle O\rangle_\sigma\right|
&=
\left|\frac{\Tr[O\tau]}{t}-\frac{\Tr[O\sigma]}{s}\right|
=
\left|\frac{s\Tr[O\tau]-t\Tr[O\sigma]}{ts}\right|
\notag\\
&=
\left|\frac{s\Tr[O(\tau-\sigma)]+(s-t)\Tr[O\sigma]}{ts}\right|
\notag\\
&\le
\frac{|\Tr[O(\tau-\sigma)]|}{\hk{|t|}}
+
\frac{|\Tr[O\sigma]|}{\hk{|t|}s}\,|s-t|.
\label{eq:ratio_split_app}
\end{align}
For the first term, H\"older's inequality gives
\begin{align}
|\Tr[O(\tau-\sigma)]|\le \|O\|\,\|\tau-\sigma\|_1.
\end{align}
For the second term, since $\sigma\ge0$ we have $\|\sigma\|_1=\Tr\sigma=s$, and thus
\begin{align}
|\Tr[O\sigma]|\le \|O\|\,\|\sigma\|_1=\|O\|\,s.
\end{align}
Substituting these bounds into Eq.~\eqref{eq:ratio_split_app} yields
\begin{align}
\left|\frac{\Tr[O\tau]}{t}-\frac{\Tr[O\sigma]}{s}\right|
\le
\frac{\|O\|}{\hk{|t|}}\|\tau-\sigma\|_1
+
\|O\|\frac{|s-t|}{\hk{|t|}}.
\end{align}
Moreover, $|s-t|=|\Tr(\sigma-\tau)|\le \|\sigma-\tau\|_1=\|\tau-\sigma\|_1$, hence
\begin{align}
\left|\frac{\Tr[O\tau]}{t}-\frac{\Tr[O\sigma]}{s}\right|
\le
\frac{2\|O\|}{\hk{|t|}}\|\tau-\sigma\|_1.
\end{align}
\hk{Since $\min\{|t|,s\}\le|t|$,}
\begin{align}
\hk{\frac{2\|O\|}{|t|}\|\tau-\sigma\|_1
\le
\frac{2\|O\|}{\min\{|t|,s\}}\|\tau-\sigma\|_1,}
\end{align}
\hk{which gives Eq.~\eqref{eq:ratio_stability_trace}.}
\end{proof}

\paragraph{\hk{Systematic-error bound.}}

Fix a realization of the Gaussian increments $\boldsymbol{\xi}$ and fix an sQEM sample.
\hk{Since the trajectory is fixed, we omit its index $s$ below.}
\hk{Let $\alpha\in\{\pm1\}$ denote the accumulated parity of this sample.}
The sample specifies the sequence of basis-operation channels used along the trajectory.
Since each basis-operation channel in Table~\ref{tab:bases} is of single-Kraus form $[A]:\rho\mapsto A\rho A^\dagger$ and is trace non-increasing, we can represent the channel at step $k$ by a single operator $\pi_k$ satisfying $\pi_k^\dagger\pi_k\le I$ and hence $\|\pi_k\|\le 1$.
In particular, measurement basis operations are trace-decreasing, so individual trajectory contributions can be unnormalized.

Let $U_k$ and $\tilde U_k$ denote the exact and implemented one-step unitary propagators used at Trotter step $k$ under the fixed $\boldsymbol{\xi}$.
Conditioned on the same sQEM record, the only difference between the two evolutions is $U_k\mapsto \tilde U_k$.
Define the (possibly unnormalized) intermediate states recursively by
\begin{align}
\rho_0 := \tilde\rho_0 := \rho_{\rm in},\qquad
\rho_{k+1} := \pi_k U_k\rho_k U_k^\dagger \pi_k^\dagger,\qquad
\tilde\rho_{k+1} := \pi_k \tilde U_k\tilde\rho_k \tilde U_k^\dagger \pi_k^\dagger,
\label{eq:recurrence_states_note_final}
\end{align}
for $k=0,1,\dots,N-1$.

Because $\pi_k^\dagger\pi_k\le I$, we have
\begin{align}
\Tr[\tilde\rho_{k+1}]
= \Tr\!\left[\pi_k^\dagger\pi_k\, \tilde U_k\tilde\rho_k \tilde U_k^\dagger\right]
\le \Tr[\tilde\rho_k],
\end{align}
and similarly $\Tr[\rho_{k+1}]\le \Tr[\rho_k]$.
Since $\Tr[\rho_{\rm in}]=1$, this implies $\Tr[\tilde\rho_k]\le 1$ and $\Tr[\rho_k]\le 1$ for all $k$.
Moreover, for any trace-class operator $X$,
\begin{align}
\|\pi_k X \pi_k^\dagger\|_1 \le \|\pi_k\|^2 \|X\|_1 \le \|X\|_1.
\label{eq:pi_contractive_trace_final}
\end{align}

Using Eq.~\eqref{eq:pi_contractive_trace_final} and the triangle inequality, we obtain
\begin{align}
\|\tilde\rho_{k+1}-\rho_{k+1}\|_1
&=
\bigl\|\pi_k\bigl(\tilde U_k\tilde\rho_k \tilde U_k^\dagger
- U_k\rho_k U_k^\dagger\bigr)\pi_k^\dagger\bigr\|_1
\notag\\
&\le
\bigl\|\tilde U_k\tilde\rho_k \tilde U_k^\dagger
- U_k\rho_k U_k^\dagger\bigr\|_1
\notag\\
&\le
\bigl\|\tilde U_k\tilde\rho_k \tilde U_k^\dagger
- U_k\tilde\rho_k U_k^\dagger\bigr\|_1
+
\bigl\|U_k(\tilde\rho_k-\rho_k)U_k^\dagger\bigr\|_1
\notag\\
&=
\bigl\|\tilde U_k\tilde\rho_k \tilde U_k^\dagger
- U_k\tilde\rho_k U_k^\dagger\bigr\|_1
+
\|\tilde\rho_k-\rho_k\|_1.
\label{eq:delta_recursion_1_final}
\end{align}
Applying Lemma~\ref{lem:unitary_trace_lipschitz} to $\rho=\tilde\rho_k$, $U=\tilde U_k$, and $V=U_k$ gives
\begin{align}
\bigl\|\tilde U_k\tilde\rho_k \tilde U_k^\dagger
- U_k\tilde\rho_k U_k^\dagger\bigr\|_1
\le 2\,\Tr[\tilde\rho_k]\;\|\tilde U_k-U_k\|
\le 2\|\tilde U_k-U_k\|.
\label{eq:unitary_term_bound_nohat_final}
\end{align}
Combining Eqs.~\eqref{eq:delta_recursion_1_final} and~\eqref{eq:unitary_term_bound_nohat_final} yields
\begin{align}
\|\tilde\rho_{k+1}-\rho_{k+1}\|_1 \le \|\tilde\rho_k-\rho_k\|_1 + 2\|\tilde U_k-U_k\|.
\end{align}
Iterating from $k=0$ to $N-1$ and using $\tilde\rho_0=\rho_0$ gives
\begin{align}
\|\tilde\rho_N-\rho_N\|_1
\le 2\sum_{k=0}^{N-1}\|\tilde U_k-U_k\|.
\label{eq:trajectory_trace_bound_note_final}
\end{align}

Define the averaged (generally unnormalized) non-Hermitian states by
\begin{align}
\rho_{\rm NH}(T) := \hk{\Lambda(T)}\mathbb{E}_{\mathrm{all}}[\hk{\alpha}\rho_N],\qquad
\rho_{\rm NH}^{(\mathrm{impl})}(T) := \hk{\Lambda(T)}\mathbb{E}_{\mathrm{all}}[\hk{\alpha}\tilde\rho_N].
\label{eq:trajectory_state_reconstruction}
\end{align}
Using $\|\mathbb{E}[X]\|_1\le \mathbb{E}[\|X\|_1]$ together with Eq.~\eqref{eq:trajectory_trace_bound_note_final}, we obtain
\begin{align}
\bigl\|\rho_{\rm NH}^{(\mathrm{impl})}(T)-\rho_{\rm NH}(T)\bigr\|_1
&=
\hk{\Lambda(T)}\bigl\|\mathbb{E}_{\mathrm{all}}[\hk{\alpha}(\tilde\rho_N-\rho_N)]\bigr\|_1
\le
\hk{\Lambda(T)}\mathbb{E}_{\mathrm{all}}\bigl[\|\tilde\rho_N-\rho_N\|_1\bigr]
\notag\\
&\le
2\,\hk{\Lambda(T)}\mathbb{E}_{\boldsymbol{\xi}}\!\left[\sum_{k=0}^{N-1}\|\tilde U_k-U_k\|\right],
\label{eq:avg_trace_bound_note_final}
\end{align}
where the last step uses \hk{$|\alpha|=1$ and the fact that} $\|\tilde U_k-U_k\|$ depends only on the Gaussian increments $\boldsymbol{\xi}$\hk{, so the remaining $\mathbb{E}_{\rm all}$ average reduces to $\mathbb{E}_{\boldsymbol{\xi}}$}.

Finally, the systematic error between the normalized expectation values is
\begin{align}
\epsilon_{\rm sys}(T)
=
\left|
\frac{\Tr[O\rho_{\rm NH}^{(\mathrm{impl})}(T)]}{\Tr[\rho_{\rm NH}^{(\mathrm{impl})}(T)]}
-
\frac{\Tr[O\rho_{\rm NH}(T)]}{\Tr[\rho_{\rm NH}(T)]}
\right|.
\end{align}
\hk{Here $\rho_{\rm NH}(T)\ge0$, whereas $\rho_{\rm NH}^{(\mathrm{impl})}(T)$ is a signed Hermitian reconstruction; we assume that both traces are nonzero.}
Applying Lemma~\ref{lem:normalized_stability} to $\tau=\rho_{\rm NH}^{(\mathrm{impl})}(T)$ and $\sigma=\rho_{\rm NH}(T)$ yields the following theorem:
\begin{theorem}[Bound for the systematic error]\label{thm:bound_sys_error}
\begin{align}
\epsilon_{\rm sys}(T)
\le
\frac{2\|O\|}{\min\{\hk{|\Tr[\rho_{\rm NH}^{(\mathrm{impl})}(T)]|},\,\hk{|\Tr[\rho_{\rm NH}(T)]|}\}}
\bigl\|\rho_{\rm NH}^{(\mathrm{impl})}(T)-\rho_{\rm NH}(T)\bigr\|_1.
\end{align}
Moreover, combining this with Eq.~\eqref{eq:avg_trace_bound_note_final} gives the desired systematic bound:
\begin{align}
\epsilon_{\rm sys}(T)
\le
\frac{4\|O\|\hk{\Lambda(T)}}{\min\{\hk{|\Tr[\rho_{\rm NH}^{(\mathrm{impl})}(T)]|},\,\hk{|\Tr[\rho_{\rm NH}(T)]|}\}}
\mathbb{E}_{\boldsymbol{\xi}}\!\left[\sum_{k=0}^{N-1}\|\tilde U_k-U_k\|\right].
\end{align}
\end{theorem}
The remaining step in End Matter Eq.~\eqref{eq:suff_sys_app} follows by upper bounding $\mathbb{E}_{\boldsymbol{\xi}}[\sum_{k=0}^{N-1}\|\tilde U_k-U_k\|]$ using the Trotter-analysis bounds, i.e., Corollary~\ref{coro:STeo_vanish} in Sec.~\ref{app:evenodd}.

\section{Trotter bounds for the \texorpdfstring{$H_{\mathrm{Re}}$-$H_{\mathrm{I}}(k)$}{Hr-Kk} splitting}
\label{app:trotter_bounds}
In this section, we derive error bounds for the trotterization between $H_{\mathrm{Re}}$ and $H_{\mathrm{I}}(k)=\sum_{\ell\in E}\xi_{k,\ell}H_{\mathrm{I},\ell}$.
The proof strategy follows the commutator-scaling approach in Ref.~\cite{Childs2021}.
We explicitly evaluate the resulting bounds for our setting with the stochastic-noise operator $H_{\mathrm{I}}(k)$.
Since $\xi_{k,\ell}\sim\mathcal{N}(0,2\gamma_\ell\Delta t)$, the resulting step-size scaling differs from the standard deterministic case.

Specifically, we provide one-step bounds for the Lie-Trotter and Suzuki-Trotter formulas used in the main text.

\subsection{\texorpdfstring{Lie-Trotter one-step bounds}{Lie-Trotter one-step bounds}}
\label{appE:LT}

\begin{lemma}[One-step Lie-Trotter bound with $(\Delta t,\{\xi_{k,\ell}\}_{\ell})$]
\label{lem:LT_HrHi}
Assume that $H_{\mathrm{Re}}$ and $\{H_{\mathrm{I},\ell}\}_{\ell\in E}$ are Hermitian.
For any $\Delta t > 0$ and any real increments $\{\xi_{k,\ell}\}_{\ell\in E}$, define the stochastic-noise operator
\begin{align}
H_{\mathrm{I}}(k) \coloneqq \sum_{\ell \in E}\xi_{k,\ell} H_{\mathrm{I},\ell},
\label{eq:HI_k_def}
\end{align}
and the \hk{Lie-Trotter} and exact one-step propagators
\begin{align}
U_k^{\mathrm{LT}} \coloneqq e^{-i\Delta t H_{\mathrm{Re}}} e^{-iH_{\mathrm{I}}(k)},
\qquad
U_k^{\mathrm{ex}} \coloneqq \exp\left[-i\left(\Delta t H_{\mathrm{Re}}+H_{\mathrm{I}}(k)\right)\right].
\label{eq:Uk_LT_Uk_ex_defs}
\end{align}
Then
\begin{align}
\left\|U_k^{\mathrm{LT}}-U_k^{\mathrm{ex}}\right\|
\le \frac{\Delta t}{2} \|[H_{\mathrm{Re}},H_{\mathrm{I}}(k)]\|
\le \frac{\Delta t}{2}\sum_{\ell \in E}|\xi_{k,\ell}| \bigl\|[H_{\mathrm{Re}},H_{\mathrm{I},\ell}]\bigr\|.
\label{eq:LT_HrHi_final}
\end{align}
\end{lemma}

\begin{proof}
Introduce an interpolation parameter $\tau\in[0,1]$ and define
\begin{align}
S^{\mathrm{LT}}(\tau)
&\coloneqq \exp(-i\tau \Delta t H_{\mathrm{Re}}) \exp(-i\tau H_{\mathrm{I}}(k)),\\
U(\tau)
&\coloneqq \exp\left[-i\tau\left(\Delta t H_{\mathrm{Re}}+H_{\mathrm{I}}(k)\right)\right],
\end{align}
so that $S^{\mathrm{LT}}(0)=U(0)=\mathbb{I}$, $S^{\mathrm{LT}}(1)=U_k^{\mathrm{LT}}$, and $U(1)=U_k^{\mathrm{ex}}$.

Using the product rule,
\begin{align}
\frac{d}{d\tau}S^{\mathrm{LT}}(\tau)
&=
-i\left(\Delta t H_{\mathrm{Re}}+H_{\mathrm{I}}(k)\right)S^{\mathrm{LT}}(\tau)
-i\Bigl[e^{-i\tau \Delta t H_{\mathrm{Re}}} H_{\mathrm{I}}(k) e^{+i\tau \Delta t H_{\mathrm{Re}}}-H_{\mathrm{I}}(k)\Bigr]S^{\mathrm{LT}}(\tau),
\label{eq:dS1_tau_HrHi}
\end{align}
whereas
\begin{align}
\frac{d}{d\tau}U(\tau)=-i\left(\Delta t H_{\mathrm{Re}}+H_{\mathrm{I}}(k)\right)U(\tau).
\label{eq:dU_tau_HrHi}
\end{align}

Let $D^{\mathrm{LT}}(\tau)\coloneqq S^{\mathrm{LT}}(\tau)-U(\tau)$.
Subtracting Eq.~\eqref{eq:dU_tau_HrHi} from Eq.~\eqref{eq:dS1_tau_HrHi} gives
\begin{align}
\frac{d}{d\tau}D^{\mathrm{LT}}(\tau)
=
-i\left(\Delta t H_{\mathrm{Re}}+H_{\mathrm{I}}(k)\right)D^{\mathrm{LT}}(\tau)+R^{\mathrm{LT}}(\tau),
\qquad D^{\mathrm{LT}}(0)=0,
\label{eq:D_inhom_HrHi}
\end{align}
with
\begin{align}
R^{\mathrm{LT}}(\tau)\coloneqq
-i\Bigl[e^{-i\tau \Delta t H_{\mathrm{Re}}} H_{\mathrm{I}}(k) e^{+i\tau \Delta t H_{\mathrm{Re}}}-H_{\mathrm{I}}(k)\Bigr]S^{\mathrm{LT}}(\tau).
\label{eq:R_tau_HrHi}
\end{align}
By Lemma~\ref{lem:duhamel},
\begin{align}
D^{\mathrm{LT}}(1)=\int_{0}^{1}d\tau
e^{-i(1-\tau)\left(\Delta t H_{\mathrm{Re}}+H_{\mathrm{I}}(k)\right)} R^{\mathrm{LT}}(\tau).
\label{eq:D1_int_HrHi}
\end{align}

Define $G^{\mathrm{LT}}(s)\coloneqq e^{-isH_{\mathrm{Re}}}H_{\mathrm{I}}(k) e^{isH_{\mathrm{Re}}}$.
Then
\begin{align}
\frac{d}{ds}G^{\mathrm{LT}}(s)=-i e^{-isH_{\mathrm{Re}}}[H_{\mathrm{Re}},H_{\mathrm{I}}(k)]e^{isH_{\mathrm{Re}}},
\end{align}
hence
\begin{align}
e^{-i\tau \Delta t H_{\mathrm{Re}}} H_{\mathrm{I}}(k) e^{+i\tau \Delta t H_{\mathrm{Re}}}-H_{\mathrm{I}}(k)
=
-i\int_{0}^{\tau\Delta t}dse^{-isH_{\mathrm{Re}}}[H_{\mathrm{Re}},H_{\mathrm{I}}(k)]e^{isH_{\mathrm{Re}}}.
\label{eq:conj_minus_HIk}
\end{align}

Since $H_{\mathrm{Re}}$ and $H_{\mathrm{I}}(k)$ are Hermitian, all exponentials above are unitary and have norm $1$.
Using the triangle inequality and submultiplicativity,
\begin{align}
\|D^{\mathrm{LT}}(1)\|
&\le \int_0^1 d\tau\|R^{\mathrm{LT}}(\tau)\|
\le \int_0^1 d\tau
\left\|e^{-i\tau \Delta t H_{\mathrm{Re}}} H_{\mathrm{I}}(k) e^{+i\tau \Delta t H_{\mathrm{Re}}}-H_{\mathrm{I}}(k)\right\|.
\end{align}
From Eq.~\eqref{eq:conj_minus_HIk},
\begin{align}
\left\|e^{-i\tau \Delta t H_{\mathrm{Re}}} H_{\mathrm{I}}(k) e^{+i\tau \Delta t H_{\mathrm{Re}}}-H_{\mathrm{I}}(k)\right\|
&\le \int_{0}^{|\tau\Delta t|}ds\|[H_{\mathrm{Re}},H_{\mathrm{I}}(k)]\|
=|\tau\Delta t|\|[H_{\mathrm{Re}},H_{\mathrm{I}}(k)]\|.
\end{align}
Therefore,
\begin{align}
\|D^{\mathrm{LT}}(1)\|
\le \int_0^1 d\tau|\tau\Delta t|\|[H_{\mathrm{Re}},H_{\mathrm{I}}(k)]\|
= \frac{\Delta t}{2} \|[H_{\mathrm{Re}},H_{\mathrm{I}}(k)]\|.
\end{align}
Since $D^{\mathrm{LT}}(1)=S^{\mathrm{LT}}(1)-U(1)=U_k^{\mathrm{LT}}-U_k^{\mathrm{ex}}$, we obtain the first inequality in Eq.~\eqref{eq:LT_HrHi_final}.
The second inequality follows from
\begin{align}
[H_{\mathrm{Re}},H_{\mathrm{I}}(k)]=\sum_{\ell \in E}\xi_{k,\ell}[H_{\mathrm{Re}},H_{\mathrm{I},\ell}],
\qquad
\|[H_{\mathrm{Re}},H_{\mathrm{I}}(k)]\|\le \sum_{\ell \in E}|\xi_{k,\ell}| \bigl\|[H_{\mathrm{Re}},H_{\mathrm{I},\ell}]\bigr\|.
\end{align}
\end{proof}

\begin{theorem}[\hk{One-step scaling of the Lie-Trotter splitting}]
\label{prop:LT_scaling_Nstep}
Let $\xi_{k,\ell}\sim\mathcal{N}(0,2\gamma_{\ell}\Delta t)$ be independent across $k$ and $\ell$.
Then the \hk{Lie-Trotter} splitting error satisfies the one-step bound
\begin{align}
\mathbb{E}_{\boldsymbol{\xi}}\left[\left\|U_k^{\mathrm{LT}}-U_k^{\mathrm{ex}}\right\|\right]
&\le
\Delta t\sum_{\ell \in E}\sqrt{\frac{\gamma_{\ell}\Delta t}{\pi}} \bigl\|[H_{\mathrm{Re}},H_{\mathrm{I},\ell}]\bigr\|.
\label{eq:LT_onestep_expectation_app}
\end{align}
\end{theorem}

\begin{proof}
By Lemma~\ref{lem:gaussian_moments}, $\mathbb{E}_{\boldsymbol{\xi}}[|\xi_{k,\ell}|]=\sqrt{4\gamma_{\ell}\Delta t/\pi}$.
\end{proof}

\subsection{\texorpdfstring{Suzuki-Trotter one-step bounds}{Suzuki-Trotter one-step bounds}}
\label{appE:ST}

We will also use Lemma~\ref{lem:duhamel}.
We repeatedly use the following standard identity: if $F(0)=F'(0)=0$, then
\begin{align}
F(\tau)=\int_0^\tau d\tau_1\int_0^{\tau_1} d\tau_2 F''(\tau_2).
\label{eq:twice_integrated_identity}
\end{align}

Fix a time step $k$ and real increments $\{\xi_{k,\ell}\}_{\ell\in E}$, and use the same $H_{\mathrm{I}}(k)$ as in Sec.~\ref{appE:LT} and $U_k^{\mathrm{ex}}$ defined in Eq.~\eqref{eq:Uk_LT_Uk_ex_defs}.
Define the Suzuki-Trotter formula
\begin{align}
U_k^{\mathrm{ST}}
\coloneqq
\exp\left(-i\frac{\Delta t}{2}H_{\mathrm{Re}}\right)
\exp(-iH_{\mathrm{I}}(k))
\exp\left(-i\frac{\Delta t}{2}H_{\mathrm{Re}}\right).
\label{eq:Uk_strang}
\end{align}

\begin{lemma}[One-step Suzuki-Trotter bound for stochastically driven unitary]
\label{lem:strang_bound_HrHi}
Assume $H_{\mathrm{Re}}$ and $\{H_{\mathrm{I},\ell}\}_{\ell\in E}$ are Hermitian.
Then, for any \hk{time step $\Delta t>0$} and any real increments $\{\xi_{k,\ell}\}_{\ell\in E}$,
\begin{align}
\left\|
U_k^{\mathrm{ST}}-U_k^{\mathrm{ex}}
\right\|
&\le
\frac{\Delta t}{12} \bigl\|[H_{\mathrm{I}}(k),[H_{\mathrm{I}}(k),H_{\mathrm{Re}}]]\bigr\|
+
\frac{\Delta t^{ 2}}{24} \bigl\|[H_{\mathrm{Re}},[H_{\mathrm{Re}},H_{\mathrm{I}}(k)]]\bigr\|.
\label{eq:strang_bound_final}
\end{align}
\end{lemma}

\begin{proof}
Introduce $\tau\in[0,1]$:
\begin{align}
U(\tau)
&\coloneqq \exp\left[-i\tau\left(\Delta t H_{\mathrm{Re}}+H_{\mathrm{I}}(k)\right)\right],\\
S^{\mathrm{ST}}(\tau)
&\coloneqq \exp\left(-i\tau\frac{\Delta t}{2}H_{\mathrm{Re}}\right)
\exp(-i\tau H_{\mathrm{I}}(k))
\exp\left(-i\tau\frac{\Delta t}{2}H_{\mathrm{Re}}\right),
\end{align}
so that $U(1)=U_k^{\mathrm{ex}}$ and $S^{\mathrm{ST}}(1)=U_k^{\mathrm{ST}}$.
Let $D^{\mathrm{ST}}(\tau)\coloneqq S^{\mathrm{ST}}(\tau)-U(\tau)$.

Write $S^{\mathrm{ST}}(\tau)=W_1(\tau)W_2(\tau)W_3(\tau)$ with $W_1(\tau)=e^{-i\tau\frac{\Delta t}{2}H_{\mathrm{Re}}}$, $W_2(\tau)=e^{-i\tau H_{\mathrm{I}}(k)}$, $W_3(\tau)=e^{-i\tau\frac{\Delta t}{2}H_{\mathrm{Re}}}$.
A direct differentiation gives
\begin{align}
\frac{d}{d\tau}S^{\mathrm{ST}}(\tau)
=
-i\left(\Delta t H_{\mathrm{Re}}+H_{\mathrm{I}}(k)\right)S^{\mathrm{ST}}(\tau)
+
R^{\mathrm{ST}}(\tau),
\end{align}
where $R^{\mathrm{ST}}(\tau)=W_1(\tau) T^{\mathrm{ST}}(\tau) W_2(\tau)W_3(\tau)$ and
\begin{align}
T^{\mathrm{ST}}(\tau)
=
\Bigl(e^{-i\tau \mathrm{ad}_{H_{\mathrm{I}}(k)}}-\mathbb{I}\Bigr)\left(-i\frac{\Delta t}{2}H_{\mathrm{Re}}\right)
+
\Bigl(e^{+i\tau\frac{\Delta t}{2} \mathrm{ad}_{H_{\mathrm{Re}}}}-\mathbb{I}\Bigr)\left(+iH_{\mathrm{I}}(k)\right).
\end{align}
Since $U'(\tau)=-i(\Delta t H_{\mathrm{Re}}+H_{\mathrm{I}}(k))U(\tau)$, we have
\begin{align}
\frac{d}{d\tau}D^{\mathrm{ST}}(\tau)
=
-i\left(\Delta t H_{\mathrm{Re}}+H_{\mathrm{I}}(k)\right)D^{\mathrm{ST}}(\tau)+R^{\mathrm{ST}}(\tau),
\qquad D^{\mathrm{ST}}(0)=0.
\end{align}
Lemma~\ref{lem:duhamel} therefore yields
\begin{align}
D^{\mathrm{ST}}(1)=\int_0^1 d\tau
e^{-i(1-\tau)\left(\Delta t H_{\mathrm{Re}}+H_{\mathrm{I}}(k)\right)} R^{\mathrm{ST}}(\tau).
\end{align}

Note that $T^{\mathrm{ST}}(0)=0$ and $(T^{\mathrm{ST}})'(0)=0$ (the linear terms cancel by the symmetry of the Suzuki-Trotter formula), so the twice-integrated identity Eq.~\eqref{eq:twice_integrated_identity} applies to $T^{\mathrm{ST}}$.
Differentiating twice gives
\begin{align}
(T^{\mathrm{ST}})''(\tau)
=
\frac{i\Delta t}{2}
e^{-i\tau \mathrm{ad}_{H_{\mathrm{I}}(k)}}\Bigl([H_{\mathrm{I}}(k),[H_{\mathrm{I}}(k),H_{\mathrm{Re}}]]\Bigr)
-\frac{i\Delta t^{ 2}}{4}
e^{+i\tau\frac{\Delta t}{2} \mathrm{ad}_{H_{\mathrm{Re}}}}\Bigl([H_{\mathrm{Re}},[H_{\mathrm{Re}},H_{\mathrm{I}}(k)]]\Bigr).
\end{align}
Using Eq.~\eqref{eq:twice_integrated_identity}, the triangle inequality, and unitary invariance of the norm, we obtain
\begin{align}
\|T^{\mathrm{ST}}(\tau)\|
\le
\tau^2\left(
\frac{\Delta t}{4} \bigl\|[H_{\mathrm{I}}(k),[H_{\mathrm{I}}(k),H_{\mathrm{Re}}]]\bigr\|
+
\frac{\Delta t^{ 2}}{8} \bigl\|[H_{\mathrm{Re}},[H_{\mathrm{Re}},H_{\mathrm{I}}(k)]]\bigr\|
\right).
\end{align}
Therefore,
\begin{align}
\|D^{\mathrm{ST}}(1)\|
\le \int_0^1 d\tau\|T^{\mathrm{ST}}(\tau)\|
=
\frac{\Delta t}{12} \bigl\|[H_{\mathrm{I}}(k),[H_{\mathrm{I}}(k),H_{\mathrm{Re}}]]\bigr\|
+
\frac{\Delta t^{ 2}}{24} \bigl\|[H_{\mathrm{Re}},[H_{\mathrm{Re}},H_{\mathrm{I}}(k)]]\bigr\|,
\end{align}
and since $D^{\mathrm{ST}}(1)=U_k^{\mathrm{ST}}-U_k^{\mathrm{ex}}$, this proves Eq.~\eqref{eq:strang_bound_final}.
\end{proof}

\begin{theorem}[\hk{One-step scaling of the Suzuki-Trotter splitting}]
\label{prop:ST_scaling_Nstep}
Let $\xi_{k,\ell}\sim\mathcal{N}(0,2\gamma_{\ell}\Delta t)$ be independent across $k$ and $\ell$, and let $U_k^{\mathrm{ST}}$ and $U_k^{\mathrm{ex}}$ be defined in Eq.~\eqref{eq:Uk_strang} and Eq.~\eqref{eq:Uk_LT_Uk_ex_defs}.
Then, using Lemma~\ref{lem:strang_bound_HrHi} and Gaussian moment identities, we obtain the one-step scaling
\begin{align}
\mathbb{E}_{\boldsymbol{\xi}}\left[\left\|U_k^{\mathrm{ST}}-U_k^{\mathrm{ex}}\right\|\right]
&\le
\frac{\Delta t}{12}\sum_{\ell,\ell'\in E}\mathbb{E}_{\boldsymbol{\xi}}\left[|\xi_{k,\ell}||\xi_{k,\ell'}|\right] \|[H_{\mathrm{I},\ell},[H_{\mathrm{I},\ell'},H_{\mathrm{Re}}]]\|
+\frac{\Delta t^{ 2}}{24}\sum_{\ell\in E}\mathbb{E}_{\boldsymbol{\xi}}[|\xi_{k,\ell}|] \|[H_{\mathrm{Re}},[H_{\mathrm{Re}},H_{\mathrm{I},\ell}]]\| \notag\\
&=
\frac{\Delta t^{ 2}}{6}\sum_{\ell\in E}\gamma_\ell \|[H_{\mathrm{I},\ell},[H_{\mathrm{I},\ell},H_{\mathrm{Re}}]]\|
+\frac{\Delta t^{ 2}}{3\pi}\sum_{\substack{\ell,\ell'\in E\\ \ell\neq \ell'}}\sqrt{\gamma_\ell\gamma_{\ell'}} \|[H_{\mathrm{I},\ell},[H_{\mathrm{I},\ell'},H_{\mathrm{Re}}]]\|\notag\\
&+\frac{\Delta t^{ 5/2}}{12\sqrt{\pi}}\sum_{\ell\in E}\sqrt{\gamma_\ell} \|[H_{\mathrm{Re}},[H_{\mathrm{Re}},H_{\mathrm{I},\ell}]]\|.
\label{eq:ST_onestep_expectation_app}
\end{align}
\end{theorem}

\begin{proof}
The one-step estimate Eq.~\eqref{eq:ST_onestep_expectation_app} follows from Lemma~\ref{lem:strang_bound_HrHi} together with Gaussian moment identities (e.g., $\mathbb{E}_{\boldsymbol{\xi}}[\xi_{k,\ell}^{2}]=2\gamma_\ell\Delta t$ and independence across $\ell$).
\end{proof}

\section{Trotter bound for the even-odd bond splitting}
\label{app:evenodd}
In this section, we quantify the trotterization error introduced by the even-odd bond splitting.
We first derive one-step and $N$-step bounds for the Lie-Trotter even-odd splitting and show that its bound contains an $\mathcal{O}(T)$ contribution that is not controlled by the step size $\Delta t$.
We then show that adopting the Suzuki-Trotter even--odd splitting for the stochastically driven unitary resolves this issue and yields a step-size-controlled bound.
The proof strategy follows the commutator-scaling approach in Ref.~\cite{Childs2021}.
We explicitly evaluate the resulting bounds for our setting with the stochastic-noise operator $H_{\mathrm{I}}(k)=\sum_{\ell\in E}\xi_{k,\ell}H_{\mathrm{I},\ell}$.
Since $\xi_{k,\ell}\sim\mathcal{N}(0,2\gamma_\ell\Delta t)$, the resulting step-size scaling differs from the standard deterministic case.
\subsection{Setup and the issue with Lie-Trotter even-odd splitting for the stochastically driven unitary}
\label{appF:setup_issue}

We first partition the set of bonds $E$ into two layers,
\begin{align}
E = E_{\mathrm{even}}\dot\cup E_{\mathrm{odd}},
\end{align}
such that within each layer the corresponding bond operators mutually commute.

For the coherent part, we decompose the Hamiltonian into two commuting bond layers,
\begin{align}
H_{\mathrm{Re}} &\coloneqq H_{\mathrm{Re}}^{(\mathrm{even})}+H_{\mathrm{Re}}^{(\mathrm{odd})},\\
H_{\mathrm{Re}}^{(\alpha)} &\coloneqq \sum_{\ell\in E_{\alpha}} H_{\mathrm{Re},\ell},
\quad (\alpha\in{\mathrm{even},\mathrm{odd}}),
\end{align}
where each layer consists of mutually commuting bonds:
\begin{align}
\bigl[H_{\mathrm{Re},\ell},H_{\mathrm{Re},\ell'}\bigr]=0
\quad \text{if}\quad
\ell,\ell'\in E_{\mathrm{even}}
\ \text{or}\
\ell,\ell'\in E_{\mathrm{odd}}.
\end{align}

For the multi-channel noise part, we assume a set of Hermitian operators $\{H_{\mathrm{I},\ell}\}_{\ell\in E}$ and, at each step $k$, independent Gaussian random variables $\{\xi_{k,\ell}\}_{\ell\in E}$ with
\begin{align}
\xi_{k,\ell}\sim \mathcal N(0,2\gamma_\ell \Delta t).
\end{align}
Using the same bond partition, we define the random layer generators
\begin{align}
H_{\mathrm{I}}^{(\alpha)}(k) &\coloneqq \sum_{\ell\in E_{\alpha}} \xi_{k,\ell} H_{\mathrm{I},\ell}
\quad (\alpha\in{\mathrm{even},\mathrm{odd}}),\\
H_{\mathrm{I}}(k) &\coloneqq H_{\mathrm{I}}^{(\mathrm{even})}(k)+H_{\mathrm{I}}^{(\mathrm{odd})}(k),
\end{align}
where each layer consists of mutually commuting bonds:
\begin{align}
[H_{\mathrm{I},\ell},H_{\mathrm{I},\ell'}]=0
\quad \text{if}\quad
\ell,\ell'\in E_{\mathrm{even}}
\ \text{or}\
\ell,\ell'\in E_{\mathrm{odd}}.
\end{align}

We define the one-step propagators for each component:
\begin{align}
V_{\mathrm{Re},k} &\coloneqq  \exp(-i\Delta t H_{\mathrm{Re}}),
&
V_{\mathrm{Re},k}^{\mathrm{eo}} &\coloneqq  \exp(-i\Delta t H_{\mathrm{Re}}^{(\mathrm{even})})\exp(-i\Delta t H_{\mathrm{Re}}^{(\mathrm{odd})}),
\\
V_{\mathrm{I},k} &\coloneqq  \exp(-i H_{\mathrm{I}}(k)),
&
V_{\mathrm{I},k}^{\mathrm{eo}} &\coloneqq  \exp(-i H_{\mathrm{I}}^{(\mathrm{even})}(k)) \exp(-i H_{\mathrm{I}}^{(\mathrm{odd})}(k)).
\end{align}
Note that, by the commutativity within each layer,
\begin{align}
\exp(-i H_{\mathrm{I}}^{(\mathrm{even})}(k))
= \prod_{\ell\in E_{\mathrm{even}}} \exp(-i\xi_{k,\ell}H_{\mathrm{I},\ell}),
\qquad
\exp(-i H_{\mathrm{I}}^{(\mathrm{odd})}(k))
= \prod_{\ell\in E_{\mathrm{odd}}} \exp(-i\xi_{k,\ell}H_{\mathrm{I},\ell}),
\end{align}
with arbitrary order inside each product.

We implement one step as
\begin{align}
U_k^{\hk{\mathrm{impl,LT(eo)}}} \coloneqq  V_{\mathrm{Re},k}^{\mathrm{eo}} V_{\mathrm{I},k}^{\mathrm{eo}}.
\label{eq:Uk_impl_eo}
\end{align}
As an intermediate reference, we \hk{recall} the Lie-Trotter split step
\begin{align}
U_k^{\mathrm{LT}} =  V_{\mathrm{Re},k} V_{\mathrm{I},k}
= \exp(-i\Delta t H_{\mathrm{Re}}) \exp(-iH_{\mathrm{I}}(k)),
\label{eq:Uk_LT_eo}
\end{align}
and the exact one-step propagator under the discretized integrated-noise model
\begin{align}
U_k^{\mathrm{ex}} =  \exp\left[-i\left(\Delta t H_{\mathrm{Re}}+H_{\mathrm{I}}(k)\right)\right].
\label{eq:Uk_exact_eo}
\end{align}

\begin{lemma}[One-step Lie-Trotter even-odd bound for $H_{\mathrm{Re}}$]
\label{lem:LT_evenodd_Hr}
Assume $H_{\mathrm{Re}}^{(\mathrm{even})}$ and $H_{\mathrm{Re}}^{(\mathrm{odd})}$ are Hermitian.
Then, for any \hk{time step $\Delta t>0$},
\begin{align}
\left\|
V_{\mathrm{Re},k}^{\mathrm{eo}}-V_{\mathrm{Re},k}
\right\|
\le \frac{\Delta t^{ 2}}{2} \bigl\|[H_{\mathrm{Re}}^{(\mathrm{even})},H_{\mathrm{Re}}^{(\mathrm{odd})}]\bigr\|.
\label{eq:LT_evenodd_Hr_final}
\end{align}
\end{lemma}

\begin{proof}
Apply Lemma~\ref{lem:LT_HrHi} with $H_{\mathrm{Re}}\rightarrow H_{\mathrm{Re}}^{(\mathrm{even})}$ and $H_{\mathrm{I}}(k)\rightarrow \Delta t H_{\mathrm{Re}}^{(\mathrm{odd})}$.
\end{proof}

\begin{lemma}[One-step Lie-Trotter even-odd bound for the stochastically driven unitary]
\label{lem:LT_evenodd_Hi}
Assume all $H_{\mathrm{I},\ell}$ are Hermitian and commute within each layer ($\ell\in E_{\mathrm{even}}$ or $\ell\in E_{\mathrm{odd}}$).
Then, for any fixed realization $\{\xi_{k,\ell}\}_{\ell\in E}$,
\begin{align}
\left\|
V_{\mathrm{I},k}^{\mathrm{eo}}-V_{\mathrm{I},k}
\right\|
&\le
\frac{1}{2} \bigl\|[H_{\mathrm{I}}^{(\mathrm{even})}(k),H_{\mathrm{I}}^{(\mathrm{odd})}(k)]\bigr\|
\label{eq:LT_evenodd_Hi_final}\\
&\le
\frac{1}{2}\sum_{\ell\in E_{\mathrm{even}}}\sum_{\ell'\in E_{\mathrm{odd}}}
|\xi_{k,\ell}| |\xi_{k,\ell'}|\bigl\|[H_{\mathrm{I},\ell},H_{\mathrm{I},\ell'}]\bigr\|.
\label{eq:LT_evenodd_Hi_final_expanded}
\end{align}
\end{lemma}

\begin{proof}
This is the same Lie-Trotter bound as Lemma~\ref{lem:LT_evenodd_Hr}, applied to $H_{\mathrm{I}}^{(\mathrm{even})}(k)+H_{\mathrm{I}}^{(\mathrm{odd})}(k)$.
The expanded bound follows from bilinearity of the commutator and the triangle inequality.
\end{proof}

\begin{proposition}[One-step error for the Lie-Trotter even-odd implementation]
\label{lem:onestep_combined_eo}
Assume all components are Hermitian.
Then, for any fixed realization $\{\xi_{k,\ell}\}_{\ell\in E}$,
\begin{align}
\left\|U_k^{\hk{\mathrm{impl,LT(eo)}}}-U_k^{\mathrm{ex}}\right\|
\le
\left\|V_{\mathrm{Re},k}^{\mathrm{eo}}-V_{\mathrm{Re},k}\right\|
+\left\|V_{\mathrm{I},k}^{\mathrm{eo}}-V_{\mathrm{I},k}\right\|
+\frac{\Delta t}{2} \bigl\|[H_{\mathrm{Re}},H_{\mathrm{I}}(k)]\bigr\|.
\label{eq:onestep_combined_eo_bound}
\end{align}
Moreover, using $\|[H_{\mathrm{Re}},H_{\mathrm{I}}(k)]\|\le \sum_{\ell\in E}|\xi_{k,\ell}|\|[H_{\mathrm{Re}},H_{\mathrm{I},\ell}]\|$,
\begin{align}
\left\|U_k^{\hk{\mathrm{impl,LT(eo)}}}-U_k^{\mathrm{ex}}\right\|
&\le
\frac{\Delta t^{ 2}}{2} \bigl\|[H_{\mathrm{Re}}^{(\mathrm{even})},H_{\mathrm{Re}}^{(\mathrm{odd})}]\bigr\|
+\frac{1}{2}\sum_{\ell\in E_{\mathrm{even}}}\sum_{\ell'\in E_{\mathrm{odd}}}
|\xi_{k,\ell}| |\xi_{k,\ell'}|\bigl\|[H_{\mathrm{I},\ell},H_{\mathrm{I},\ell'}]\bigr\|
\notag\\
&\quad\qquad\qquad\qquad\qquad\qquad\qquad\qquad\qquad\qquad
+\frac{\Delta t}{2}\sum_{\ell\in E}|\xi_{k,\ell}|\bigl\|[H_{\mathrm{Re}},H_{\mathrm{I},\ell}]\bigr\|.
\label{eq:onestep_combined_eo_bound_expanded}
\end{align}
\end{proposition}

\begin{proof}
Use Lemmas~\ref{lem:LT_evenodd_Hr} and~\ref{lem:LT_evenodd_Hi} and the triangle inequality.
The remaining $H_{\mathrm{Re}}$-$H_{\mathrm{I}}(k)$ term is bounded by the one-step Lie-Trotter bound (see Sec.~\ref{app:trotter_bounds} for details).
\end{proof}

\begin{theorem}[$N$-step accumulation for the Lie-Trotter even-odd implementation]
\label{prop:Nstep_Lie_eo}

Taking the ensemble average over the independent Gaussian increments and using Lemma~\ref{lem:gaussian_moments}, we obtain
\begin{align}
\mathbb{E}_{\boldsymbol{\xi}}\left[\sum_{k=0}^{N-1}\left\|\hk{U^{\mathrm{impl,LT(eo)}}_k}-U^{\mathrm{ex}}_k\right\|\right]
&\le
\frac{T\Delta t}{2} \bigl\|[H_{\mathrm{Re}}^{(\mathrm{even})},H_{\mathrm{Re}}^{(\mathrm{odd})}]\bigr\|
+\frac{2T}{\pi}\sum_{\ell\in E_{\mathrm{even}}}\sum_{\ell'\in E_{\mathrm{odd}}}
\sqrt{\gamma_\ell\gamma_{\ell'}}\bigl\|[H_{\mathrm{I},\ell},H_{\mathrm{I},\ell'}]\bigr\|
\notag\\
&\qquad\qquad\qquad\qquad\qquad\qquad\qquad\qquad\qquad
+T\sum_{\ell\in E}\sqrt{\frac{\gamma_\ell\Delta t}{\pi}}\bigl\|[H_{\mathrm{Re}},H_{\mathrm{I},\ell}]\bigr\|.
\label{eq:Nstep_expectation_combined_eo_prop}
\end{align}
\end{theorem}

\begin{proof}
Taking the ensemble average and using independence of the Gaussian increments together with Lemma~\ref{lem:gaussian_moments} yields Eq.~\eqref{eq:Nstep_expectation_combined_eo_prop}.
\end{proof}

\begin{remark}
The bound in Theorem~\ref{prop:Nstep_Lie_eo} means the ensemble-averaged bound Eq.~\eqref{eq:Nstep_expectation_combined_eo_prop} contains a contribution of the form
\begin{align}
\frac{2T}{\pi}\sum_{\ell\in E_{\mathrm{even}}}\sum_{\ell'\in E_{\mathrm{odd}}}
\sqrt{\gamma_\ell\gamma_{\ell'}}\bigl\|[H_{\mathrm{I},\ell},H_{\mathrm{I},\ell'}]\bigr\|,
\end{align}
which scales as $\mathcal O(T)$ and therefore does not vanish as $\Delta t\to 0$.
Hence, the Lie-Trotter even-odd splitting of the stochastically driven unitary may yield an error contribution that is not controlled by the step size.
In the next subsection, we show that this issue is resolved by adopting a Suzuki-Trotter even-odd splitting for $V_{\mathrm{I},k}$.
\end{remark}

\subsection{Suzuki-Trotter even-odd splitting for the stochastically driven unitary}
\label{app:evenodd_hi_strang_full}

We improve the even-odd splitting inside the multi-channel stochastically driven unitary by using a Suzuki-Trotter formula between the even and odd bond layers, while keeping the Lie-Trotter even-odd splitting inside $H_{\mathrm{Re}}$ as before.

We use the same notation as in Sec.~\ref{appF:setup_issue}.
In particular, $H_{\mathrm{Re}} = H_{\mathrm{Re}}^{(\mathrm{even})}+H_{\mathrm{Re}}^{(\mathrm{odd})}$ and $H_{\mathrm{I}}(k)=H_{\mathrm{I}}^{(\mathrm{even})}(k)+H_{\mathrm{I}}^{(\mathrm{odd})}(k)$.

We define
\begin{align}
V_{\mathrm{I},k}^{\mathrm{(even,half)}}\coloneqq e^{-i\frac{1}{2}H_{\mathrm{I}}^{(\mathrm{even})}(k)},
\qquad
V_{\mathrm{I},k}^{\mathrm{(odd)}}\coloneqq e^{-iH_{\mathrm{I}}^{(\mathrm{odd})}(k)}.
\end{align}
The Suzuki-Trotter even-odd splitting is
\begin{align}
V_{\mathrm{I},k}^{\mathrm{ST(eo)}}
\coloneqq
V_{\mathrm{I},k}^{\mathrm{(even,half)}}
V_{\mathrm{I},k}^{\mathrm{(odd)}}
V_{\mathrm{I},k}^{\mathrm{(even,half)}}.
\label{eq:VI_ST_eo_def}
\end{align}

We keep
\begin{align}
V_{\mathrm{Re},k} \coloneqq e^{-i\Delta t H_{\mathrm{Re}}},
\qquad
V_{\mathrm{Re},k}^{\mathrm{eo}} \coloneqq e^{-i\Delta t H_{\mathrm{Re}}^{(\mathrm{even})}} e^{-i\Delta t H_{\mathrm{Re}}^{(\mathrm{odd})}},
\end{align}
and implement one step as
\begin{align}
U_k^{\mathrm{impl,ST(eo)}} \coloneqq V_{\mathrm{Re},k}^{\mathrm{eo}} V_{\mathrm{I},k}^{\mathrm{ST(eo)}}.
\label{eq:Uk_impl_STeo}
\end{align}
We also define $U_k^{\mathrm{LT}}\coloneqq V_{\mathrm{Re},k}V_{\mathrm{I},k}$ and $U_k^{\mathrm{ex}}\coloneqq \exp[-i(\Delta t H_{\mathrm{Re}}+H_{\mathrm{I}}(k))]$ as before.

\begin{lemma}[One-step Suzuki-Trotter even-odd bound for the stochastically driven unitary]
\label{lem:VI_evenodd_ST_onestep}
Assume that all $H_{\mathrm{I},\ell}$ are Hermitian and commute within each layer.
Then, for any fixed realization $\{\xi_{k,\ell}\}_{\ell\in E}$,
\begin{align}
  \left\|V_{\mathrm{I},k}^{\mathrm{ST(eo)}}-V_{\mathrm{I},k}\right\|
  \le
  \frac{1}{12} \bigl\|[H_{\mathrm{I}}^{(\mathrm{odd})}(k),[H_{\mathrm{I}}^{(\mathrm{odd})}(k),H_{\mathrm{I}}^{(\mathrm{even})}(k)]]\bigr\|
  +
  \frac{1}{24} \bigl\|[H_{\mathrm{I}}^{(\mathrm{even})}(k),[H_{\mathrm{I}}^{(\mathrm{even})}(k),H_{\mathrm{I}}^{(\mathrm{odd})}(k)]]\bigr\|.
\label{eq:VI_STeo_onestep_bound}
\end{align}
\end{lemma}

\begin{proof}
This is a direct specialization of the standard one-step Suzuki-Trotter bound to $\exp[-i(H_{\mathrm{I}}^{(\mathrm{even})}(k)+H_{\mathrm{I}}^{(\mathrm{odd})}(k))]$ and $e^{-iH_{\mathrm{I}}^{(\mathrm{even})}(k)/2}e^{-iH_{\mathrm{I}}^{(\mathrm{odd})}(k)}e^{-iH_{\mathrm{I}}^{(\mathrm{even})}(k)/2}$.
\end{proof}

\begin{proposition}[One-step error for the Suzuki-Trotter even-odd implementation]
\label{prop:onestep_combined_STeo}
Assume that all components are Hermitian.
Then, for any fixed realization $\{\xi_{k,\ell}\}_{\ell\in E}$,
\begin{align}
\left\|U_k^{\mathrm{impl,ST(eo)}}-U_k^{\mathrm{ex}}\right\|
\le
\left\|V_{\mathrm{Re},k}^{\mathrm{eo}}-V_{\mathrm{Re},k}\right\|
+\left\|V_{\mathrm{I},k}^{\mathrm{ST(eo)}}-V_{\mathrm{I},k}\right\|
+\left\|U_k^{\mathrm{LT}}-U_k^{\mathrm{ex}}\right\|.
\label{eq:onestep_combined_STeo_bound}
\end{align}
\end{proposition}

\begin{proof}
Insert $V_{\mathrm{Re},k}V_{\mathrm{I},k}^{\mathrm{ST(eo)}}$ and $U_k^{\mathrm{LT}}=V_{\mathrm{Re},k}V_{\mathrm{I},k}$ \hk{and} apply the triangle inequality.
\end{proof}

\begin{theorem}[$N$-step accumulation for the Suzuki-Trotter even-odd implementation]
\label{theo:Nstep_STeo}
Taking the ensemble average over the independent Gaussian increments $\{\xi_{k,\ell}\}_{k,\ell}$ yields
\begin{align}
\mathbb{E}_{\boldsymbol{\xi}}
\left[\sum_{k=0}^{N-1}\left\|U^{\mathrm{impl,ST(eo)}}_k-U^{\mathrm{ex}}_k\right\|\right]
&\le
\frac{T\Delta t}{2} \bigl\|[H_{\mathrm{Re}}^{(\mathrm{even})},H_{\mathrm{Re}}^{(\mathrm{odd})}]\bigr\|
+
T\sum_{\ell\in E}\sqrt{\frac{\gamma_\ell\Delta t}{\pi}}
\bigl\|[H_{\mathrm{Re}},H_{\mathrm{I},\ell}]\bigr\|
\notag\\
&\quad
+
\frac{T\sqrt{\Delta t}}{6\sqrt{\pi}}
\Bigl(
2 \sum_{\ell,\ell'\in E_{\mathrm{odd}}}\sum_{m\in E_{\mathrm{even}}} \sqrt{\gamma_\ell\gamma_{\ell'}\gamma_m}  \bigl\|[H_{\mathrm{I},\ell'},[H_{\mathrm{I},\ell},H_{\mathrm{I},m}]]\bigr\|
\notag\\
&\quad\qquad\qquad\quad
+
\sum_{m,m'\in E_{\mathrm{even}}}\sum_{\ell\in E_{\mathrm{odd}}} \sqrt{\gamma_m\gamma_{m'}\gamma_\ell}  \bigl\|[H_{\mathrm{I},m'},[H_{\mathrm{I},m},H_{\mathrm{I},\ell}]]\bigr\|
\Bigr).
\label{eq:Nstep_expectation_STeo}
\end{align}
\end{theorem}

\begin{proof}
The inequality in Eq.~\eqref{eq:Nstep_expectation_STeo} for a fixed realization of $\{\xi_{k,\ell}\}_{k,\ell}$ follows from Proposition~\ref{prop:onestep_combined_STeo}, together with Lemmas~\ref{lem:LT_HrHi}, \ref{lem:LT_evenodd_Hr}, and~\ref{lem:VI_evenodd_ST_onestep}.

We now take the ensemble average.
For the even-odd term of $H_{\mathrm{Re}}$, the bound in Lemma~\ref{lem:LT_evenodd_Hr} is
\begin{align}
\sum_{k=0}^{N-1}\mathbb{E}_{\boldsymbol{\xi}}\bigl[\|V_{\mathrm{Re},k}^{\mathrm{eo}}-V_{\mathrm{Re},k}\|\bigr]
\le \frac{T\Delta t}{2}\|[H_{\mathrm{Re}}^{(\mathrm{even})},H_{\mathrm{Re}}^{(\mathrm{odd})}]\|.
\end{align}

For the $H_{\mathrm{Re}}$-$H_{\mathrm{I}}(k)$ Lie-Trotter term, using $\|[H_{\mathrm{Re}},H_{\mathrm{I}}(k)]\|\le \sum_{\ell\in E}|\xi_{k,\ell}|\|[H_{\mathrm{Re}},H_{\mathrm{I},\ell}]\|$ and Lemma~\ref{lem:gaussian_moments} gives
\begin{align}
\mathbb{E}_{\boldsymbol{\xi}}[|\xi_{k,\ell}|]=\sqrt{4\gamma_\ell\Delta t/\pi},
\end{align}
and therefore
\begin{align}
\sum_{k=0}^{N-1}\mathbb{E}_{\boldsymbol{\xi}}\bigl[\|U_k^{\mathrm{LT}}-U_k^{\mathrm{ex}}\|\bigr]
\le
T\sum_{\ell\in E}\sqrt{\gamma_\ell\Delta t/\pi} \|[H_{\mathrm{Re}},H_{\mathrm{I},\ell}]\|.
\end{align}

For the Suzuki-Trotter even-odd splitting of $V_{\mathrm{I},k}$, we expand
\begin{align}
[H_{\mathrm{I}}^{(\mathrm{odd})}(k),[H_{\mathrm{I}}^{(\mathrm{odd})}(k),H_{\mathrm{I}}^{(\mathrm{even})}(k)]]
=
\sum_{\ell,\ell'\in E_{\mathrm{odd}}}\sum_{m\in E_{\mathrm{even}}}
\xi_{k,\ell}\xi_{k,\ell'}\xi_{k,m}
[H_{\mathrm{I},\ell'},[H_{\mathrm{I},\ell},H_{\mathrm{I},m}]],
\end{align}
and similarly for the other nested commutator.
Taking operator norms and using the triangle inequality yields an upper bound by the corresponding triple sum with factors $|\xi_{k,\ell}||\xi_{k,\ell'}||\xi_{k,m}|$.
Using Lemma~\ref{lem:gaussian_moments} and the independence of the Gaussian increments, one obtains for $\ell,\ell'$ in the same layer and $m$ in the opposite layer
\begin{align}
\mathbb{E}_{\boldsymbol{\xi}}[|\xi_{k,\ell}||\xi_{k,\ell'}||\xi_{k,m}|]
\le \frac{4}{\sqrt{\pi}}\sqrt{\gamma_\ell\gamma_{\ell'}\gamma_m} (\Delta t)^{3/2}.
\end{align}
Substituting this into Lemma~\ref{lem:VI_evenodd_ST_onestep} and summing over $k$ gives the last term in Eq.~\eqref{eq:Nstep_expectation_STeo}.
Combining the three contributions completes the proof.
\end{proof}

\begin{corollary}[Vanishing of the ensemble-averaged error for symmetric even-odd splitting]
\label{coro:STeo_vanish}
Assume that all components are Hermitian and that the within-layer commutativity holds for $\{H_{\mathrm{I},\ell}\}_{\ell\in E}$.
Then there exists a constant $C^{\hk{\mathrm{ST(eo)}}}>0$ (independent of $\Delta t$ and $T$) such that
\begin{align}
\mathbb{E}_{\boldsymbol{\xi}}\left[\sum_{k=\hk{0}}^{N-1}\left\|U^{\mathrm{impl,ST(eo)}}_k-U^{\mathrm{ex}}_k\right\|\right]
\le C^{\hk{\mathrm{ST(eo)}}} T \sqrt{\Delta t},
\label{eq:STeo_vanish_bound}
\end{align}
and hence the left-hand side \hk{vanishes} as $\Delta t\to 0$ for fixed $T$.
In contrast to the Lie-Trotter even-odd splitting discussed in Eq.~\eqref{eq:Nstep_expectation_combined_eo_prop}, the Suzuki-Trotter even-odd splitting yields an error that is controlled by the step size.
\end{corollary}

\hk{\section{Two-qubit benchmark}}
\begin{figure}[ht]
  \centering
  \includegraphics[width=0.6\linewidth]{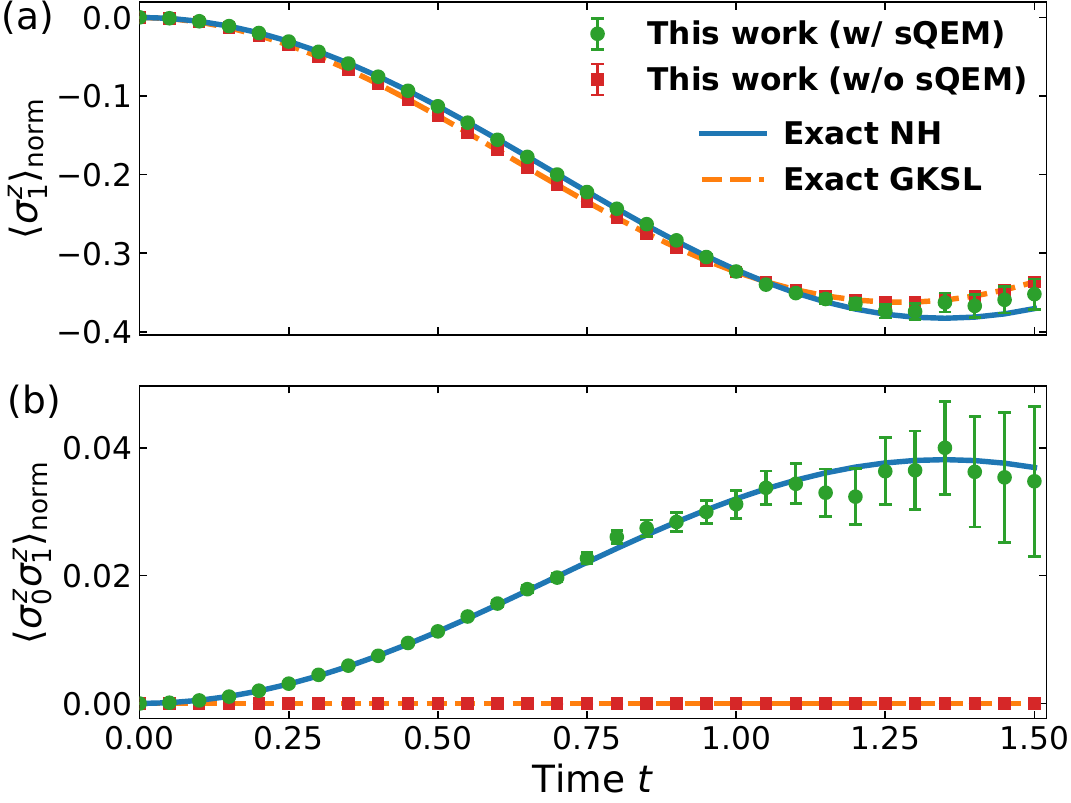}
  \caption{Normalized expectation values (a) $\langle \sigma^z_1\rangle_{\mathrm{norm}}(t)$ and (b) $\langle \sigma^z_0\sigma^z_1\rangle_{\mathrm{norm}}(t)$ \hk{are computed from the initial state} $|\psi_0\rangle=|+\rangle^{\otimes 2}$.
\hk{Solid blue and dashed orange curves show the exact non-Hermitian and GKSL dynamics}, respectively.
Green circles and red squares show Monte Carlo estimates from our protocol with and without sQEM.
\hk{Parameters are} $J=1.0$, $g=0.1$, $U=2.0$, $\gamma=1.0$, $(h_0,h_1)=(-0.8071,0.3890)$, $\Delta t=10^{-3}$, and $t\in\{0,0.05,\ldots,1.50\}$.
We use $B=8000$ independent runs, each averaging over $M_{\mathrm{num}}=2000$ stochastic trajectories (total $BM_{\mathrm{num}}=1.6\times 10^7$ trajectories).
Error bars show jackknife standard errors \hk{over the $B$ runs}.}
  \label{fig:L2_fourcurves}
\end{figure}

\hk{Figure~\ref{fig:L2_fourcurves} provides a minimal two-qubit validation of our protocol.}
We evaluate $\langle \sigma^z_1\rangle_{\mathrm{norm}}(t)$ and $\langle \sigma^z_0\sigma^z_1\rangle_{\mathrm{norm}}(t)$, starting from $|\psi_0\rangle=|+\rangle^{\otimes 2}$ (see Fig.~\ref{fig:L2_fourcurves} for parameters).
The data labeled ``This work (w/o sQEM)'' are obtained by running the same \hk{Gaussian-noise} sampler while omitting the sQEM step; they therefore estimate the corresponding GKSL evolution, including the jump contribution.
As shown in Fig.~\ref{fig:L2_fourcurves}, the w/o sQEM data agree with the exact GKSL curve, whereas the w/ sQEM data match the exact non-Hermitian target \hk{dynamics} within statistical uncertainty.
\hk{This comparison separately validates the noise-averaged GKSL sampler and the sQEM cancellation of the jump contribution.}

\section{\hk{Site occupations in the full-bond} benchmark}
\begin{figure}[H]
  \centering
  \includegraphics[width=0.7\linewidth]{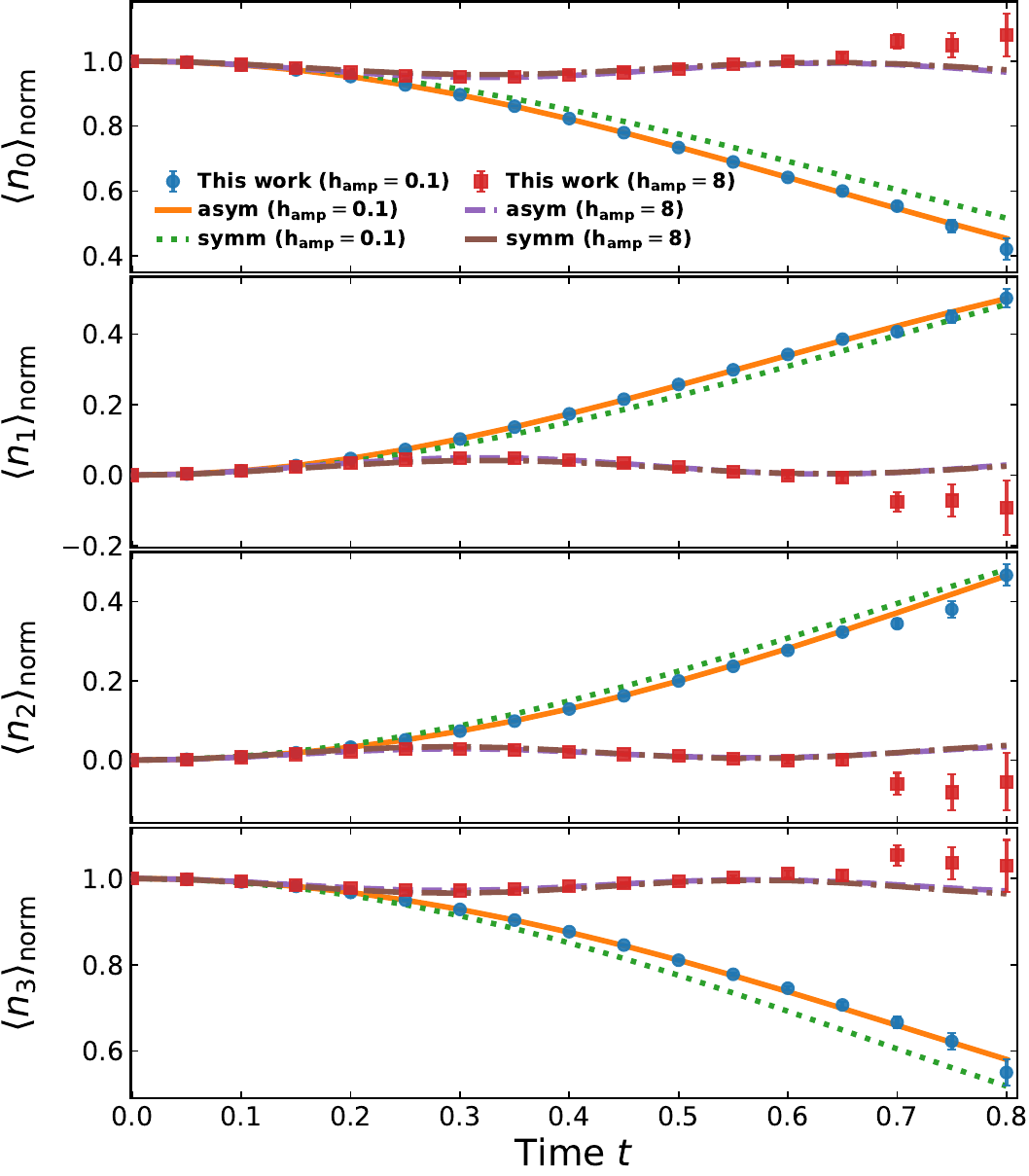}
  \caption{\hk{S}ite occupations $\langle n_i(t)\rangle_{\mathrm{norm}}$ ($i=0,1,2,3$) \hk{in the four-site full-bond non-Hermitian benchmark corresponding to Fig.~\ref{fig:L4_merged}.
Results are shown} for weak disorder ($h_{\mathrm{amp}}=0.1$) and strong disorder ($h_{\mathrm{amp}}=8.0$), where the initial state is $|0110\rangle$.
Solid orange and dashed purple lines show the exact target non-Hermitian dynamics (asym, $g=0.1$) for $h_{\mathrm{amp}}=0.1$ and $h_{\mathrm{amp}}=8.0$, respectively.
Dotted green and dash-dotted brown lines show the corresponding symmetric-hopping model (symm, $g=0$).
Blue circles ($h_{\mathrm{amp}}=0.1$) and red squares ($h_{\mathrm{amp}}=8.0$) show Monte Carlo estimates from our protocol.
Parameters \hk{are} $J=1.0$, $g=0.1$, $U=1.0$, $\gamma=1.0$, $\Delta t=1.0\times 10^{-3}$, and $t\in\{0,0.05,\ldots,0.80\}$.
We use $B=210240$ independent runs, each averaging over $M_{\mathrm{num}}=2000$ stochastic trajectories (total $BM_{\mathrm{num}}=4.2048\times 10^8$ trajectories).
Error bars show jackknife standard errors over the $B$ runs.}
  \label{fig:L4_site_occupation}
\end{figure}
\hk{Figure~\ref{fig:L4_site_occupation} resolves the site occupations underlying the edge-imbalance signal shown in the main text.}
\hk{F}or weak disorder ($h_{\mathrm{amp}}=0.1$), particles spread into the bulk and $\langle n_0(t)\rangle_{\mathrm{norm}}$ and $\langle n_3(t)\rangle_{\mathrm{norm}}$ decrease, whereas for strong disorder ($h_{\mathrm{amp}}=8.0$) the occupations remain close to their initial values, indicating inhibited transport.
\hk{Because the finite-sample estimator is a ratio of signed sample averages, its values are not constrained to the physical interval $[0,1]$, which explains the late-time excursions in the strong-disorder data.}

\end{document}